\documentclass[11pt]{article}
\usepackage[utf8]{inputenc}

\usepackage[margin=1in]{geometry}
\usepackage{amsmath,amsfonts,amssymb,amsthm}
\usepackage{aliascnt}
\usepackage{mathtools}
\mathtoolsset{centercolon}
\usepackage{mathrsfs}
\usepackage{bm}

\usepackage{float}

\usepackage{tikz}
\usetikzlibrary{arrows.meta,positioning}

\usepackage{enumitem}
\usepackage{booktabs}
\usepackage{xcolor}
\usepackage[colorlinks]{hyperref}
\usepackage[capitalize]{cleveref}

\allowdisplaybreaks
\numberwithin{equation}{section}

\newtheorem{theorem}{Theorem}[section]

\newaliascnt{lemma}{theorem}
\newtheorem{lemma}[lemma]{Lemma}
\aliascntresetthe{lemma}

\newaliascnt{corollary}{theorem}
\newtheorem{corollary}[corollary]{Corollary}
\aliascntresetthe{corollary}

\newaliascnt{proposition}{theorem}
\newtheorem{proposition}[proposition]{Proposition}
\aliascntresetthe{proposition}

\newaliascnt{assumption}{theorem}

\aliascntresetthe{assumption}

\theoremstyle{definition}

\newaliascnt{definition}{theorem}
\newtheorem{definition}[definition]{Definition}
\aliascntresetthe{definition}

\newaliascnt{notation}{theorem}

\aliascntresetthe{notation}

\newaliascnt{remark}{theorem}
\newtheorem{remark}[remark]{Remark}
\aliascntresetthe{remark}

\newaliascnt{example}{theorem}

\aliascntresetthe{example}

\crefname{problem}{problem}{problems}
\Crefname{problem}{Problem}{Problems}
\crefname{theorem}{theorem}{theorems}
\Crefname{theorem}{Theorem}{Theorems}
\crefname{lemma}{lemma}{lemmas}
\Crefname{lemma}{Lemma}{Lemmas}
\crefname{corollary}{corollary}{corollaries}
\Crefname{corollary}{Corollary}{Corollaries}
\crefname{proposition}{proposition}{propositions}
\Crefname{proposition}{Proposition}{Propositions}
\crefname{assumption}{assumption}{assumptions}
\Crefname{assumption}{Assumption}{Assumptions}
\crefname{definition}{definition}{definitions}
\Crefname{definition}{Definition}{Definitions}
\crefname{notation}{notation}{notations}
\Crefname{notation}{Notation}{Notations}
\crefname{remark}{remark}{remarks}
\Crefname{remark}{Remark}{Remarks}
\crefname{example}{example}{examples}
\Crefname{example}{Example}{Examples}

\newcommand{\C}{\mathbb{C}}
\newcommand{\R}{\mathbb{R}}
\newcommand{\I}{\mathrm{I}}
\newcommand{\spec}{\operatorname{spec}}

\DeclareMathOperator{\diag}{diag}

\newcommand{\abs}[1]{\left\lvert #1 \right\rvert}
\newcommand{\norm}[1]{\left\lVert #1 \right\rVert}

\newcommand{\RePart}{\operatorname{Re}}
\newcommand{\ImPart}{\operatorname{Im}}

\newcommand{\ket}[1]{\lvert #1 \rangle}
\newcommand{\bra}[1]{\langle #1 \rvert}

\newcommand{\secref}[1]{\hyperref[sec:#1]{Section~\ref*{sec:#1}}}
\newcommand{\appref}[1]{\hyperref[app:#1]{Appendix~\ref*{app:#1}}}
\newcommand{\thmref}[1]{\hyperref[thm:#1]{Theorem~\ref*{thm:#1}}}
\newcommand{\lemref}[1]{\hyperref[lem:#1]{Lemma~\ref*{lem:#1}}}
\newcommand{\corref}[1]{\hyperref[cor:#1]{Corollary~\ref*{cor:#1}}}
\newcommand{\fig}[1]{\hyperref[fig:#1]{Figure~\ref*{fig:#1}}}

\newcommand{\Disk}{\mathbb{D}}
\newcommand{\Id}{\I}
\newcommand{\ii}{\mathrm{i}}
\newcommand{\ee}{\mathrm{e}}
\newcommand{\dd}{\mathrm{d}}
\newcommand{\Ran}{\operatorname{Ran}}
\newcommand{\Pcal}{\mathcal{P}}
\newcommand{\DiskAlg}{\mathcal D(\Disk)}
\DeclareMathOperator{\Log}{Log}
\DeclareMathOperator{\Sign}{sign}
\DeclareMathOperator{\ReLU}{ReLU}
\DeclareMathOperator*{\Res}{Res}
\newcommand{\ip}[2]{\left\langle #1,#2\right\rangle}
\newcommand{\Projector}[1]{\ket{#1}\!\bra{#1}}
\newcommand{\Outer}[2]{\ket{#1}\!\bra{#2}}

\usepackage{tabularx}
\usepackage{ragged2e}
\usepackage{makecell}
\usepackage{placeins}
\newcolumntype{Y}{>{\RaggedRight\arraybackslash}X}

\usepackage{array,threeparttable,multirow,pifont}

\usepackage{graphicx}

\title{Optimal Quantum Eigenvalue Transformation via\\Linear Combinations of Hermitian Matrices}
\author{Yanqiao Wang$^{1,2,3}$, Yixuan Liang$^{1,2}$, Hongjia Chen$^{4}$, Jin-Peng Liu$^{1,5,6,\thanks{Corresponding author: liujinpeng@tsinghua.edu.cn}}$ \\
\footnotesize $^{1}$ Yau Mathematical Sciences Center, Tsinghua University\\
\footnotesize $^{2}$ Qiuzhen College, Tsinghua University\\
\footnotesize $^{3}$ Institute for AI Industry Research (AIR), Tsinghua University, Beijing, China \\
\footnotesize $^{4}$ School of Mathematical Sciences, Peking University\\
\footnotesize $^{5}$ Institute for Applied Mathematics, Tsinghua University\\
\footnotesize $^{6}$ Beijing Institute of Mathematical Sciences and Applications}
\date{}

\begin{document}
\maketitle

\begin{abstract}
We discover two complementary linear-combination-of-Hermitian-matrices (LCHM) formulations to achieve a general non-normal matrix eigenvalue transformation $g(A)$. Firstly, for
$A=L+\ii H$ with Hermitian $L$ and $H$, the vanilla LCHM formula represents $g(A)$ as a kernel integral of $g(\ii(H+kL))$, and it contains linear-combination-of-Hamiltonian-simulation (LCHS) [An,~Liu,~Lin, Phys.~Rev.~Lett.~2023] as the special case for matrix exponentials. Secondly, for the angular Hermitian $X_\theta = \cos\theta\,L+\sin\theta\,H$, the Weyl LCHM formula expresses $g(A)$ via integrating $g(\ee^{\ii\theta}
(X_\theta\pm\ii(\Id-X_\theta^2)^{1/2}))$ with minimal weights.
For the matrix power instance $g(A)=A^m$, the Fourier projection of Weyl LCHM gives
\[
A^m=\frac{2}{\pi}\int_0^\pi \text{e}^{\text{i} m\theta}T_m(X_\theta)\,\text{d}\theta = \frac{2}{N}\sum_{j=0}^{N-1}
\text{e}^{\text{i} m\theta_j}T_m(X_{\theta_j}),\quad\theta_j=\frac{\pi j}{N},\quad \text{for every}\ N>m
\]
with Chebyshev polynomial of Hermitian $T_m(X_\theta)$ and $N$ samples. The discrete formula is exact, introduces no truncation and angular quadrature error, and offers $\mathcal{O}(1)$ post-selection weights.  

LCHM formulas lead to new quantum eigenvalue transformation (QET) algorithms. As an exact and compact representation for a degree-$d$ polynomial transform $p_d(A)$ on $|\psi\rangle$, our QET algorithm can achieve optimal $\Theta(d)$ matrix-query circuit depth and optimal $\mathcal{O}(||p_d||_{\infty}/||p_d(A)|\psi\rangle||)$ post-selection repetitions. LCHM-based QETs unify various quantum linear algebraic problems with near-optimal $\mathcal{\widetilde O}(d\log(d/\epsilon))$ Clifford$+T$ gates, including matrix exponentials and driven ODEs (reduced to standard LCHS), affine iterative methods, first- and higher-order resolvents, $\Log(\Id+A)$, shifted fractional powers $(\lambda\Id+A)^\nu$, Sign and ReLU transforms, and Faber approximation on noncircular domains. Besides optimal asymptotic scaling, Weyl LCHM with numerical-radius rescaling of weighted-shift matrices can exponentially reduce the prefactor in query complexity over prior QET methods. %Vanilla and Weyl LCHM formulas may provide new mathematical insights as well as practical algorithmic design principles towards fault-tolerant quantum computing.
\end{abstract}

\tableofcontents

\section{Introduction}\label{sec:introduction}

\subsection{Matrix functions and quantum eigenvalue transformation}

Functions of non-normal matrices occur in differential equations, control, stability theory, non-Hermitian physics, matrix equations, and scientific computing.  Their behavior cannot in general be inferred from eigenvalues alone: Jordan derivatives, the numerical range, and resolvent growth can all matter.  Classical matrix-function theory defines $g(A)$ through the eigenvalues and Jordan structure of $A$~\cite{Higham2008}.  This functional calculus is generally distinct from the singular-value transformation implemented by QSVT, which is defined through the singular-value decomposition of $A$~\cite{GilyenEtAl2019}.

Quantum computing offers the prospect of exponential or superpolynomial improvements for several foundational scientific-computing primitives.  Representative developments include universal, LCU/Taylor-series, and QSP/qubitization algorithms for Hamiltonian simulation~\cite{Lloyd1996,ChildsWiebe2012,BerryEtAl2015,LowChuang2017,LowChuang2019}; HHL, high-precision LCU/QSVT, and optimal-scaling approaches to quantum linear systems and matrix arithmetic~\cite{HarrowEtAl2009,ChildsKothariSomma2017,GilyenEtAl2019,CostaEtAl2022}; and high-order, history-state, spectral, time-marching, and Carleman-linearization algorithms for differential equations~\cite{Berry2014,BerryChildsOstranderWang2017,ChildsLiu2020,Krovi2023,FangLinTong2023}.  A central unifying mechanism is quantum signal processing (QSP), which encodes a degree-$d$ scalar polynomial through a length-$d$ sequence of signal uses and single-qubit rotations.  Quantum singular value transformation (QSVT) lifts this mechanism to block-encoded matrices and implements polynomial transformations of their singular values with degree-optimal query complexity~\cite{GilyenEtAl2019}.

Ordinary quantum eigenvalue transformation (QET) for a general non-normal matrix is substantially harder than QSVT.  QSVT naturally alternates the signal with its adjoint and therefore acts on singular values, becoming an ordinary eigenvalue transformation only in special cases such as Hermitian inputs.  By contrast, $p(A)$ depends on the ordered powers of the principal block itself, including its Jordan structure.  Leakage out of and back into the block-encoding success subspace means that powers of a generic signal unitary do not directly produce powers of its principal block.  Generalized quantum signal processing (GQSP) implements bounded complex polynomials of a unitary signal~\cite{MotlaghWiebe2024}, but this block-encoding obstruction remains for a general matrix input.

Several recent methods address this difficulty.  Quantum eigenvalue processing (QEP) constructs Chebyshev and Faber history states for polynomial transforms of non-normal matrices~\cite{LowSu2026}.  Regular block encoding uses a counter register to prevent failed branches from returning to the success subspace, after which GQSP implements the desired polynomial, including for defective matrices~\cite{GutierrezEtAl2026}.  These methods realize ordinary matrix transformations through a history-state or regularization layer.

A second family of methods represents matrix functions as linear combinations of simpler unitary or inverse primitives.  The original linear combination of Hamiltonian simulation (LCHS) represents dissipative matrix exponentials $\ee^{-tA}$, for an accretive matrix $A=L+\ii H$ with $L\succeq0$, as an integral of unitary evolutions generated by $H+kL$ and achieves optimal state-preparation dependence~\cite{AnLiuLin2023}.  An, Childs, and Lin developed an LCHS algorithm for general linear nonunitary dynamics with optimal state-preparation cost and near-optimal matrix-query dependence, while Low and Somma obtained optimal dependence on all parameters for time-independent linear nonunitary dynamics~\cite{AnChildsLin2026Dynamics,LowSomma2025}.  Laplace-LCHS extends this route to functions with a suitable integrable inverse-Laplace transform, including shifted inverse powers and mass-matrix differential equations~\cite{AnChildsLinYing2026}.  Fourier and Weyl-calculus variants improve kernels and extend the construction to broader analytic symbols and discrete LCHS measures~\cite{HuangAn2025,NiYing2026}.  These methods are particularly effective when the transformed Hermitian matrices enter through Hamiltonian exponentials.

Schr\"odingerization maps general linear differential equations to Schr\"odinger dynamics through a warped-phase transformation in an auxiliary dimension~\cite{JinLiuYu2024,JinLiuYu2023}.  It also applies to continuous-time formulations of discrete linear systems~\cite{JinLiu2024}.

Contour-integral algorithms use Cauchy's formula to reduce $g(A)$ to shifted resolvents and then apply quantum linear-system primitives~\cite{HaleHighamTrefethen2008,TakahiraEtAl2020,TakahiraEtAl2022,JiangAn2026}.  Their function class is broad, and they can be advantageous for strongly stable differential equations, but their cost depends on contour geometry, shifted-resolvent conditioning, and quadrature.  Sign-embedding algorithms first compress structured matrix equations or matrix functions into a half-plane matrix sign and then use a canonical real-line resolvent representation~\cite{WangLiu2026}.  Classical Faber approximation gives another way to adapt the polynomial degree to a noncircular numerical range~\cite{MoretNovati2001,Trefethen2019}.

We develop two LCHM formulas that reduce non-normal matrix functions to Hermitian matrix transforms.  Vanilla LCHM is tailored to functions represented through Hermitian pencils, while Weyl LCHM gives an exact angular realization of polynomial transforms.  Together they lead to optimal quantum eigenvalue transformation algorithms and cover a broad range of matrix functions arising in quantum linear algebra.

\subsection{Vanilla and Weyl LCHM formulations}

Throughout the paper we use the Cartesian decomposition
\begin{equation}
L=\frac{A+A^\dagger}{2},
\qquad
H=\frac{A-A^\dagger}{2\ii},
\qquad
A=L+\ii H,
\label{eq:intro-cartesian}
\end{equation}
with $L=L^\dagger$ and $H=H^\dagger$.

\begin{theorem}[Informal vanilla LCHM; \Cref{thm:plain-entire,thm:plain-strip}]
\label{thm:intro-vanilla-lchm}
For the function classes treated in \Cref{sec:plain-lchm}, there is an associated scalar kernel $f$ such that
\begin{equation}
g(A)=\int_{\R}f(k)\,g\!\left(\ii(H+kL)\right)\,\dd k.
\label{eq:intro-plain}
\end{equation}
The hypotheses and the strip-shift formula are stated in \Cref{thm:plain-entire,thm:plain-strip}.
\end{theorem}

For every real $k$, the matrix $H+kL$ is Hermitian.  When $g(z)=\ee^{-tz}$ and $L\succeq0$, \Cref{thm:intro-vanilla-lchm} reduces to the original LCHS identity.  The representation is obtained from the residue at the complex parameter $k=-\ii$, for which $\ii(H+kL)=A$.

The Weyl formulation uses the angular Hermitian pencil
\begin{equation}
X_\theta
=\cos\theta\,L+\sin\theta\,H
=\frac{\ee^{-\ii\theta}A+\ee^{\ii\theta}A^\dagger}{2}
=\RePart(\ee^{-\ii\theta}A).
\label{eq:intro-angle-matrix}
\end{equation}

\begin{theorem}[Informal Weyl LCHM for matrix powers; \Cref{thm:exact-projection}]
\label{thm:intro-weyl-lchm}
For every integer $m\geq1$,
\begin{equation}
A^m=\frac{2}{\pi}\int_0^\pi
\ee^{\ii m\theta}T_m(X_\theta)\,\dd\theta.
\label{eq:intro-continuous}
\end{equation}
On the equispaced grid
\begin{equation}
\theta_j=\frac{\pi j}{N},
\qquad j=0,\ldots,N-1,
\label{eq:intro-grid}
\end{equation}
one also has, for every $N>m$,
\begin{equation}
A^m=\frac{2}{N}\sum_{j=0}^{N-1}
\ee^{\ii m\theta_j}T_m(X_{\theta_j}).
\label{eq:intro-discrete}
\end{equation}
\end{theorem}

The coefficient of $\ee^{-\ii m\theta}$ in $T_m(X_\theta)$ is exactly $A^m/2$, even though $A$ and $A^\dagger$ need not commute.  Hence every polynomial of degree at most $d$ is recovered by an \emph{exact} root-of-unity projection with any $N>d$.  No diagonalization, contour discretization, shifted inverse, integral truncation, or angular quadrature is involved.

The angular family is controlled precisely by the numerical radius:
\begin{equation}
\sup_{\theta\in[0,\pi]}\norm{X_\theta}=w(A).
\label{eq:intro-radius}
\end{equation}

\begin{theorem}[Informal Weyl LCHM in the disk algebra; \Cref{thm:disk-algebra}]
\label{thm:intro-disk-angular}
If $w(A)\leq1$ and $g$ is holomorphic in $\Disk$ and continuous on $\overline\Disk$, then the unitary matrices
\begin{equation}
U_{\theta,\pm}
=\ee^{\ii\theta}\left(X_\theta\pm\ii(\Id-X_\theta^2)^{1/2}\right)
\label{eq:intro-unitaries}
\end{equation}
satisfy
\begin{equation}
g(A)=\frac1\pi\int_0^\pi
\left[g(U_{\theta,+})+g(U_{\theta,-})-g(0)\Id\right] \,\dd\theta.
\label{eq:intro-disk}
\end{equation}
\end{theorem}

Define the lifted scalar function
\begin{equation}
F_g(z)=2g(z)-g(0).
\label{eq:intro-lift}
\end{equation}
If $P_d$ approximates $F_g$ uniformly on the closed unit disk and
$p_d=(P_d+P_d(0))/2$, then
\begin{equation}
\norm{g(A)-p_d(A)}
\leq\norm{F_g-P_d}_{\infty,\partial\Disk}.
\label{eq:intro-transfer}
\end{equation}
Uniform scalar approximation therefore transfers to operator norm with constant one; the subsequent angular realization is exact and adds no approximation error.

\begin{theorem}[Informal Faber--Weyl LCHM; \Cref{thm:faber-weyl-identity}]
\label{thm:intro-faber-weyl}
Let $K\subset\C$ be compact and convex with nonempty interior, let $W(A)\subset K$, and let $\Phi$ be the normalized exterior conformal map of $K$.  For $R>1$, set
$\Gamma_R=\{z:\abs{\Phi(z)}=R\}$, and suppose $f$ is holomorphic on a neighborhood of the compact region enclosed by $\Gamma_R$.  If $r_d$ is the degree-$d$ Faber truncation and
\begin{equation}
M_R=\max_{z\in\Gamma_R}\abs{f(z)},
\label{eq:intro-faber-MR}
\end{equation}
then
\begin{equation}
\norm{f(A)-r_d(A)}
\leq \frac{2(1+\sqrt2)M_R}{R-1}R^{-d}.
\label{eq:intro-faber-error}
\end{equation}
Writing $r_d(z)=\sum_{m=0}^d b_mz^m$, one has, for every $N>d$,
\begin{equation}
r_d(A)=\frac1N\sum_{j=0}^{N-1}
\left[b_0+2\sum_{m=1}^d
b_m\ee^{\ii m\theta_j}T_m(X_{\theta_j})\right],
\qquad \theta_j=\frac{\pi j}{N}.
\label{eq:intro-faber-identity}
\end{equation}
\end{theorem}

The Faber truncation contains all geometry-dependent approximation error, while the Weyl projection of the resulting polynomial is exact.  This identity applies equally to circular, elliptical, and general convex numerical-range enclosures.

\subsection{Optimal quantum eigenvalue transformation algorithms}

For a nonzero degree-$d$ polynomial $p_d$ and a nonzero target $g\in\DiskAlg$, define the lifted normalizations
\begin{equation}
C_p:=\norm{2p_d-p_d(0)}_{\infty,\partial\Disk},
\qquad
C_g:=\norm{2g-g(0)}_{\infty,\partial\Disk}.
\label{eq:intro-Cp-Cg}
\end{equation}
For a normalized input $\ket{\psi}$ with nonzero transformed output, define the polynomial and target QET constants
\begin{equation}
q_{\mathrm{et}}(p_d;A,\psi)
:=\frac{\norm{p_d}_{\infty,\overline\Disk}}
{\norm{p_d(A)\ket{\psi}}},
\qquad
q_{\mathrm{et}}(g;A,\psi)
:=\frac{\norm{g}_{\infty,\overline\Disk}}
{\norm{g(A)\ket{\psi}}}.
\label{eq:intro-qet-constant}
\end{equation}
These quantities separate approximation degree from post-selection conditioning.  In particular, if
$P_d=2p_d-p_d(0)$ approximates $F_g=2g-g(0)$ with error
$\varepsilon_{\mathrm{app}}$ on $\partial\Disk$, then
\begin{equation}
\abs{C_p-C_g}\leq\varepsilon_{\mathrm{app}}.
\label{eq:intro-Cp-Cg-control}
\end{equation}
Thus $C_p$ is controlled by the boundary modulus of the target lift and does not intrinsically grow with the polynomial degree.

\begin{theorem}[Informal Weyl LCHM implementation; \Cref{thm:coherent-algorithm,prop:query-optimality,prop:angle-lcu-optimality}]
\label{thm:intro-optimal-qet}
Let $g\in\DiskAlg$, let $\norm{A}\leq1$, let $U_A$ be a $(1,a,\delta_A)$ block encoding of $A$, and let
$P_d(z)=\sum_{m=0}^dc_mz^m$ satisfy
\begin{equation}
\norm{F_g-P_d}_{\infty,\partial\Disk}\leq\varepsilon_{\mathrm{app}}.
\label{eq:intro-qet-data}
\end{equation}
Set
\begin{equation}
p_d=\frac{P_d+P_d(0)}2,
\qquad
C_p=\norm{P_d}_{\infty,\partial\Disk}
=\norm{2p_d-p_d(0)}_{\infty,\partial\Disk}.
\label{eq:intro-weyl-normalization}
\end{equation}
For $p_d\not\equiv0$, a coherent Weyl LCHM circuit gives a $C_p$-normalized block encoding of $g(A)$ with unnormalized error at most
\begin{equation}
\varepsilon_{\mathrm{app}}
+\frac{\delta_A}{2}\sum_{m=1}^d m\abs{c_m}
+C_p\varepsilon_{\mathrm{circ}}.
\label{eq:intro-qet-error}
\end{equation}
It uses
\begin{equation}
Q_A=\mathcal{O}(d),
\qquad
Q_{\mathrm{ang}}=\mathcal{O}\!\left(d\log(d+1)\right),
\qquad
n_{\mathrm{anc}}=a+\left\lceil\log_2(d+1)\right\rceil+2.
\label{eq:intro-qet-resources}
\end{equation}
If $c_d\neq0$, the matrix-oracle circuit depth is $\Theta(d)$.  Ancilla-free Clifford+$T$ approximation of the axial rotations contributes
\begin{equation}
\mathcal{O}\!\left(
 d\log(d+1)
 \log\!\frac{d\log(d+1)}{\varepsilon_{\mathrm{circ}}}
\right)
\label{eq:intro-qet-T-count}
\end{equation}
$T$ gates apart from the matrix-oracle cost.  For a normalized input $\ket{\psi}$, amplitude amplification prepares the normalized polynomial output using
\begin{equation}
\mathcal{O}\!\left(
 d\,q_{\mathrm{et}}(p_d;A,\psi)
\right)
=\mathcal{O}\!\left(
\frac{dC_p}{\norm{p_d(A)\ket{\psi}}}
\right)
=\mathcal{O}\!\left(
\frac{d\norm{g}_{\infty,\overline\Disk}}
{\norm{g(A)\ket{\psi}}}
\right)
\label{eq:intro-qet-state-prep}
\end{equation}
controlled matrix-oracle calls.  

Consequently, Weyl LCHM has degree-optimal circuit depth and optimal
$\Theta(q_{\mathrm{et}})$ post-selection dependence up to universal constants and the prescribed approximation error; see
\Cref{thm:coherent-algorithm,prop:query-optimality,prop:angle-lcu-optimality}.
\end{theorem}

\begin{theorem}[Informal Faber--Weyl LCHM implementation; \Cref{thm:faber-weyl-implementation}]
\label{thm:intro-faber-weyl-implementation}
Under the assumptions of \Cref{thm:intro-faber-weyl}, choose
\begin{equation}
A=c\Id+sB,
\qquad c\in\C,\quad s>0,\quad \norm{B}\leq1,
\label{eq:intro-faber-affine}
\end{equation}
and write
\begin{equation}
2r_d(c+sz)-r_d(c)=\sum_{m=0}^d\widehat c_mz^m,
\qquad
C_{r_d}
:=\norm{2r_d(c+s\,\cdot)-r_d(c)}_{\infty,\partial\Disk}.
\label{eq:intro-faber-normalization}
\end{equation}
Given a $(1,a,\delta_B)$ block encoding of $B$, a coherent Faber--Weyl circuit has a $C_{r_d}$-normalized success block for $r_d(A)$ and unnormalized error at most
\begin{equation}
\frac{2(1+\sqrt2)M_R}{R-1}R^{-d}
+\frac{\delta_B}{2}\sum_{m=1}^d m\abs{\widehat c_m}
+C_{r_d}\varepsilon_{\mathrm{circ}}.
\label{eq:intro-faber-implementation-error}
\end{equation}
It uses $\mathcal{O}(d)$ controlled block-encoding calls,
$\mathcal{O}(d\log(d+1))$ angle rotations, and
$a+\lceil\log_2(d+1)\rceil+2$ logical ancillas.  Clifford+$T$ rotation approximation contributes
\begin{equation}
\mathcal{O}\!\left(
 d\log(d+1)
 \log\!\frac{d\log(d+1)}{\varepsilon_{\mathrm{circ}}}
\right)
\label{eq:intro-faber-T-count}
\end{equation}
$T$ gates apart from the oracle cost.  For a normalized input $\ket{\psi}$, amplitude amplification uses
\begin{equation}
\mathcal{O}\!\left(
\frac{dC_{r_d}}{\norm{r_d(A)\ket{\psi}}}
\right)
\label{eq:intro-faber-state-prep}
\end{equation}
controlled matrix-oracle calls.  The angular projection is exact and contributes no quadrature error.
\end{theorem}

Weyl LCHM is optimal in two independent senses: its matrix-oracle circuit depth is degree-optimal, and its post-selection dependence is optimal in the uniform contraction model.  The angular LCU coefficients have unit $\ell_1$-norm, so the angle register contributes neither a factor of $N$ nor a factor growing with $d$.  The remaining QET constant is intrinsic because any uniformly valid success block must accommodate scalar contractions throughout $\overline\Disk$.  For the exact power $p_m(z)=z^m$, one has $C_{p_m}=2$ and
$q_{\mathrm{et}}(p_m;A,\psi)=1/\norm{A^m\ket{\psi}}$; the Weyl success-amplitude ratio is $2q_{\mathrm{et}}(p_m;A,\psi)$, and the coherent repetition count is constant whenever the power output has constant norm.

\subsection{Summary of matrix function examples}

The examples in \Cref{sec:examples-extensions} show how approximation degree, lift normalization, and output amplitude determine the choice between vanilla and Weyl LCHM.  For each polynomial row, $p$, $C_p$, and $\zeta_p$ denote the implemented polynomial, its lift normalization, and its output norm on $\ket{\psi}$ in the corresponding result of \Cref{sec:examples-extensions}.  The dominant controlled matrix-oracle complexities are summarized in \Cref{tab:function-complexities}.  Here $\varepsilon_{\mathrm{op}}$ is an unnormalized operator-approximation error; normalized-state error $\varepsilon_{\mathrm{st}}$ is obtained by choosing operator error $\mathcal O(\varepsilon_{\mathrm{st}}\zeta_p)$.  In the driven-ODE row, one attempt block-encodes the normalized driven propagator $\Phi_{\mu,t}(A)$ from \eqref{eq:driven-average}; $\kappa_\mu:=\norm{A}/\mu$, and \eqref{eq:driven-amplitude-lower-bound} gives $\zeta_{\mathrm{drv}}\geq1/\kappa_\mu$, uniformly for $t>0$.

\begin{table}[H]
\centering
\scriptsize
\renewcommand{\arraystretch}{1.28}
\setlength{\tabcolsep}{1.6pt}
\begin{tabularx}{\textwidth}{|>{\raggedright\arraybackslash}p{2.35cm}|>{\centering\arraybackslash}p{3.25cm}|>{\centering\arraybackslash}p{3.15cm}|>{\centering\arraybackslash}X|}
\hline
Target & Regime & Circuit depth & Normalized output-state preparation \\
\hline
\makecell[l]{Driven ODE\\$\dot x=-Ax+b$}
& $\displaystyle L\succeq\mu\Id>0,\ x(0)=0$
& $\displaystyle
\mathcal O\!\left(\alpha_A t\log\frac1{\varepsilon_{\mathrm{op}}}\right)$
& $\displaystyle
\mathcal O\!\left(
\kappa_\mu\alpha_A t
\log\!\frac{\kappa_\mu}{\varepsilon_{\mathrm{st}}}
\right)$ \\
\hline
Exact power $A^m$
& $\norm A\leq1$, $m\geq1$
& $\mathcal O(m)$
& $\displaystyle
\mathcal O\!\left(\frac{m}{\norm{A^m\ket{\psi}}}\right)$ \\
\hline
\makecell[l]{Affine iterate\\$x^{(k)}$}
& $\norm A\leq1$
& $\mathcal O(k)$
& $\displaystyle
\mathcal O\!\left(\frac{k\alpha_k}{\norm{x^{(k)}}}\right)$ \\
\hline
\makecell[l]{Resolvent\\$(\lambda\Id-A)^{-k}$}
& $\norm A\leq1$, $\abs{\lambda}>1$
& $\displaystyle
\mathcal O_{\lambda,k}\!\left(
\log\frac1{\varepsilon_{\mathrm{op}}}
\right)$
& $\displaystyle
\mathcal O_{\lambda,k}\!\left(
\frac{\norm{p}_{\infty,\overline\Disk}}{\zeta_p}
\log\!\frac1{\varepsilon_{\mathrm{st}}\zeta_p}
\right)$ \\
\hline
\makecell[l]{Logarithm\\$\Log(\Id+A)$}
& $\norm A\leq\rho<1$
& $\displaystyle
\mathcal O_\rho\!\left(
\log\frac1{\varepsilon_{\mathrm{op}}}
\right)$
& $\displaystyle
\mathcal O_\rho\!\left(
\frac{\norm{p}_{\infty,\overline\Disk}}{\zeta_p}
\log\!\frac1{\varepsilon_{\mathrm{st}}\zeta_p}
\right)$ \\
\hline
\makecell[l]{Fractional power\\$(\lambda\Id+A)^\nu$}
& $\norm A\leq\rho<\lambda$
& $\displaystyle
\mathcal O_{\rho,\lambda,\nu}\!\left(
\log\frac1{\varepsilon_{\mathrm{op}}}
\right)$
& $\displaystyle
\mathcal O_{\rho,\lambda,\nu}\!\left(
\frac{\norm{p}_{\infty,\overline\Disk}}{\zeta_p}
\log\!\frac1{\varepsilon_{\mathrm{st}}\zeta_p}
\right)$ \\
\hline
\makecell[l]{Sign, projectors,\\ReLU}
& $\norm A\leq1$
& $\displaystyle
\mathcal O\!\left(\deg(s)+1\right)$
& $\displaystyle
\mathcal O\!\left(
\frac{(\deg(s)+1)\norm{p}_{\infty,\overline\Disk}}{\zeta_p}
\right)$ \\
\hline
\makecell[l]{Partial fractions\\(aggregate $\mathfrak p$)}
& $\norm A\leq1$, $\abs{\lambda_\ell}>1$
& $\displaystyle
\mathcal O\!\left(\deg(\mathfrak p)\right)$
& $\displaystyle
\mathcal O\!\left(
\frac{\deg(\mathfrak p)C_{\mathfrak p}}
{\norm{\mathfrak p(A)\ket{\psi}}}
\right)$ \\
\hline
\end{tabularx}
\caption{Representative matrix function examples with complexities proved in \Cref{sec:examples-extensions}.  The last column includes amplitude amplification.  In the driven-ODE row, $b\neq0$ is constant and the displayed bound follows from $\zeta_{\mathrm{drv}}\geq1/\kappa_\mu$ in \eqref{eq:driven-amplitude-lower-bound}; hence its amplification factor is uniform in $t$ rather than exponential in $t$.  The quantity $\alpha_k$ is defined in \eqref{eq:affine-iteration-normalization}.  A polynomial $p$ also uses $\mathcal O(\deg(p)\log(\deg(p)+1))$ angle rotations and $a+\lceil\log_2(\deg(p)+1)\rceil+2$ logical ancillas.  The affine-iteration construction uses one additional outer-LCU selector.}
\label{tab:function-complexities}
\end{table}

\subsection{Contributions and comparison with prior methods}

The main contributions are as follows.
\begin{enumerate}[leftmargin=2.1em,label=\textup{(\roman*)}]
\item We prove that every matrix power $A^m$ is recovered as an exact noncommutative Fourier coefficient of a Chebyshev transform of $X_\theta$.  The continuous and root-of-unity formulas apply to arbitrary matrices, including non-normal ones, and extend to analytic and disk-algebra functions.
\item A single coherent GQSP circuit evaluates all angular branches in superposition and performs the root-of-unit projection. It implements a degree-$d$ polynomial with $\Theta(d)$ matrix-oracle circuit depth, $\mathcal{O}(d\log d)$ angle rotations, $\mathcal{O}(\log d)$ angle qubits, and no truncation and angular quadrature error.
\item The angular LCU coefficients have unit $\ell_1$-norm. For a polynomial $p_d(A)$, the normalization $C_p$ satisfies
$\norm{p_d}_{\infty,\overline\Disk}\leq C_p\leq3\norm{p_d}_{\infty,\overline\Disk}$ and, for $p_d$ as an approximation to $g$, obeys $\abs{C_p-C_g}\leq\varepsilon_{\mathrm{app}}$.  The resulting QET constant $q_{\mathrm{et}}$ is unavoidable up to universal constants for arbitrary contractions.
\item The identity $\sup_\theta\norm{X_\theta}=w(A)$ enables numerical-radius rescaling.  For the weighted-shift family in \Cref{app:rescaled-weyl-comparison}, it yields an exponential-in-degree reduction in total state-preparation queries for matrix powers relative to $d$-regular GQSP, including coherent-amplification overhead.
\item Vanilla LCHM supplies a residue formula for entire functions of subexponential order and a strip-shift formula for half-plane analytic functions; its matrix-exponential specialization recovers LCHS, and a Duhamel LCU extends it to constant-source driven ODEs.
\item Explicit complexity bounds cover powers, affine stationary iterations, exponentials, driven ODEs, resolvents, logarithms, shifted fractional powers, matrix sign, spectral projectors, ReLU transforms, and partial-fraction approximations.  Faber--Weyl further combines geometry-adapted Faber approximation with exact angular projection, retaining $\mathcal O(d)$ matrix-oracle complexity on noncircular numerical-range enclosures.
\end{enumerate}

\begin{table}[H]
\centering
\scriptsize
\renewcommand{\arraystretch}{1.5}
\setlength{\tabcolsep}{0.8pt}
\begin{tabularx}{\textwidth}{|>{\centering\arraybackslash}p{2.3cm}|>{\centering\arraybackslash}X|>{\centering\arraybackslash}p{2.2cm}|>{\centering\arraybackslash}p{3.35cm}|>{\centering\arraybackslash}p{3.1cm}|}
\hline
Method & Assumptions & Circuit depth & Amplified repetitions & Ancilla qubits \\
\hline
\makecell[c]{Contour integral\\\cite{TakahiraEtAl2020,TakahiraEtAl2022,JiangAn2026}}
& $\displaystyle
\begin{gathered}
\spec(A)\subset\operatorname{int}(\Gamma),\quad
R_\Gamma=\max_{z\in\Gamma}\abs z,\\[-1pt]
B_{p,\Gamma}=\max_{z\in\Gamma}\abs{p_d(z)},\quad
\ell_\Gamma=\operatorname{len}(\Gamma),\\[-1pt]
\gamma_\Gamma=\sup_{z\in\Gamma}\norm{(z\Id-A)^{-1}}
\end{gathered}$
& \makecell[c]{$\widetilde{\mathcal O}\!\left(\gamma_\Gamma(1+R_\Gamma)\right)$}
& $\displaystyle \mathcal O\!\left(
\frac{B_{p,\Gamma}\ell_\Gamma\gamma_\Gamma}
{\norm{p_d(A)\ket{\psi}}}
\right)$
& \makecell[c]{$a+\lceil\log_2 M_\Gamma\rceil +2$} \\
\hline
\makecell[c]{\\ QEP\\\cite{LowSu2026}\\}
& $ \displaystyle
\begin{gathered}
A=S\Lambda S^{-1},\quad \spec(A)\subset\R,\\[-1pt]
\kappa_S=\norm S\norm{S^{-1}}
\end{gathered}$
& $\widetilde{\mathcal O}(\kappa_S d)$
& $\displaystyle \widetilde{\mathcal O}\!\left(
\frac{\kappa_S\norm{p_d}_{\infty,[-1,1]}}
{\norm{p_d(A)\ket{\psi}}}
\right)$
& $a+\mathcal O(\log d)$ \\
\hline
\makecell[c]{\\ $d$-regular GQSP\\\cite{GutierrezEtAl2026}\\ }
& none
& $\widetilde{\mathcal O}(d)$
& $\displaystyle \mathcal O\!\left(
\frac{\norm{p_d}_{\infty,\partial\Disk}}
{\norm{p_d(A)\ket{\psi}}}
\right)$
& $a+\mathcal O(\log d)$ \\
\hline
\makecell[c]{\\Weyl LCHM\\ (this work)\\ }
& none
& $\widetilde{\mathcal O}(d)$
& $ \displaystyle \mathcal O\!\left(
\frac{\norm{p_d}_{\infty,\partial\Disk}}
{\norm{p_d(A)\ket{\psi}}}
\right)$
& \makecell[c]{$a+\lceil\log_2(d+1)\rceil +2$} \\
\hline
\end{tabularx}
\caption{Comparison for degree-$d$ polynomial matrix functions $p_d(A)$ with $d\geq1$ under the common condition $\norm{A}\leq1$ and given a $(1,a,\delta_A)$ block encoding of $A$. The $d$-regular entry uses the $\mathcal O(\log d)$ ancillary-qubit bound, and the Weyl entry uses $a+\lceil\log_2(d+1)\rceil+2$.  Here $M_\Gamma$ is the number of contour quadrature nodes, and the $\mathcal O(\log d)$ term in QEP includes its history and linear-system work registers.}
\label{tab:intro-comparison}
\end{table}

\textbf{Weyl LCHM vs.~$d$-regular GQSP}. As presented in \Cref{tab:intro-comparison}, two approaches share almost the same asymptotic scaling. The constant-factor comparison in the amplified repetitions demonstrates the Weyl-specific advantage. Weyl LCHM exposes the numerical radius through $\sup_\theta\norm{X_\theta}=w(A)$.  Given a certified bound $w(A)<s<||A||\leq1$, rescaling damps the degree-$k$ coefficient by $s^k$ while the angular blocks are amplified coherently.  For the explicit weighted-shift family in \Cref{app:rescaled-weyl-comparison}, this gives an exponential-in-degree reduction in the prefactor of total state-preparation queries for matrix powers relative to $d$-regular GQSP~\cite{GutierrezEtAl2026}. This yields an exponential advantage for Weyl LCHM.

\textbf{Relation with Laplace and generalized LCHS}.
Laplace-LCHS treats targets with an integrable inverse-Laplace representation, including shifted inverse powers and functions arising in mass-matrix differential equations~\cite{AnChildsLinYing2026}.  Weyl LCHM complements this class on bounded matrix domains: it implements polynomials exactly and approximates logarithms and shifted powers through explicit finite-degree formulas.  Generalized LCHS via Weyl calculus approach covers a broad analytic class through two-dimensional plane-wave expansions and can obtain target-optimized discrete weights from a constrained convex program~\cite{NiYing2026}.  For polynomial targets, Weyl LCHM provides the nodes, weights, and normalization in closed form through a one-angle root-of-unity projection.

The remainder of the paper is organized as follows.  \Cref{sec:plain-lchm} develops vanilla LCHM and its normalization behavior.  \Cref{sec:angular-calculus} proves the Weyl LCHM identities.  \Cref{sec:quantum-realization} gives the coherent quantum circuit and the degree and normalization lower bounds.  \Cref{sec:examples-extensions} presents matrix function examples with complexity results.  \Cref{sec:faber} develops the Faber--Weyl LCHM extension. \Cref{sec:comparison} compares related methods.  
\Cref{sec:conclusion} concludes the paper by summarizing the main contributions. \Cref{app:rescaled-weyl-comparison} develops numerical-radius rescaling and the weighted-shift separation. \Cref{app:boundary} states the boundary semisimplicity of accretive matrices. \Cref{app:qubitization} proves the two-dime nsional qubitized Chebyshev identity. \Cref{app:error-resources} collects perturbation and fault-tolerant resource estimates.

\section{Vanilla LCHM: generalization of LCHS}
\label{sec:plain-lchm}

For real $k$, define the Hermitian generator and anti-Hermitian matrix
\begin{equation}
K_k=H+kL,
\qquad
M_k=\ii K_k=\ii(H+kL).
\label{eq:plain-generators}
\end{equation}
The identity $M_{-\ii}=A$ is the residue mechanism behind vanilla LCHM.

\subsection{An exact residue formula for entire functions}

Fix $0<\beta<1$ and use the principal branch of the fractional power on the open right half-plane.  Define
\begin{equation}
f_\beta(z)
=\frac{1}{2\pi}
\frac{\exp\!\left(2^\beta-(1+\ii z)^\beta\right)}{1-\ii z}.
\label{eq:plain-kernel}
\end{equation}
In the lower half-plane, $1+\ii z$ has positive real part.  Hence $f_\beta$ is holomorphic there except for a simple pole at $z=-\ii$, with
\begin{equation}
\Res_{z=-\ii}f_\beta(z)=\frac{\ii}{2\pi}.
\label{eq:plain-residue}
\end{equation}
Moreover, for real $k$ and for $z$ on a sufficiently large lower semicircle,
\begin{equation}
\abs{f_\beta(z)}
\leq \frac{C_\beta}{1+\abs{z}}
\exp\!\left[-c_\beta\abs{z}^\beta\right]
\label{eq:plain-kernel-decay}
\end{equation}
with positive constants $C_\beta,c_\beta$.

\begin{theorem}[Vanilla LCHM identity for entire functions]
\label{thm:plain-entire}
Let $A=L+\ii H\in\C^{n\times n}$ with $L,H$ Hermitian.  Suppose $g$ is entire and there exist $C,c>0$ and $0\leq\beta'<\beta<1$ such that
\begin{equation}
\abs{g(z)}\leq C\exp\!\left(c\abs{z}^{\beta'}\right),
\qquad z\in\C.
\label{eq:plain-growth}
\end{equation}
Then the operator-norm integral converges absolutely and
\begin{equation}
 g(A)=\int_{\R}f_\beta(k)\,
 g\!\left(\ii(H+kL)\right)\,\dd k.
\label{eq:plain-master}
\end{equation}
In particular, the theorem applies to every matrix polynomial.
\end{theorem}

\begin{proof}
Let
\[
G(z)=f_\beta(z)g\!\left(\ii(H+zL)\right).
\]
The matrix-valued factor is entire in $z$.  Integrate $G$ over the clockwise contour consisting of $[-R,R]$ and the lower semicircle.  The only enclosed pole is $-\ii$, and \eqref{eq:plain-residue} gives
\begin{equation}
\int_{-R}^{R}G(k)\,\dd k+\int_{\Gamma_R}G(z)\,\dd z
=-2\pi\ii\Res_{z=-\ii}G(z)=g(A).
\label{eq:plain-contour}
\end{equation}
For $z\in\Gamma_R$, the norm of $\ii(H+zL)$ is $\mathcal{O}(R)$.  The Crouzeix--Palencia estimate~\cite{CrouzeixPalencia2017} and \eqref{eq:plain-growth} therefore give
\[
\norm{g\!\left(\ii(H+zL)\right)}
\leq C'\exp(c'R^{\beta'}).
\]
Equation~\eqref{eq:plain-kernel-decay} supplies the stronger factor
$\exp(-c_\beta R^\beta)$, so the semicircle integral tends to zero.  The same estimates on the real axis give absolute convergence.  Letting $R\to\infty$ in \eqref{eq:plain-contour} proves the result.
\end{proof}

\begin{remark}
For $g(z)=\ee^{-tz}$, accretivity supplies the required lower-half-plane decay even though the function has order one.  When $L\succeq0$ and $t\geq0$, dissipativity suppresses the semicircular arc and yields the LCHS identity in \Cref{subsec:matrix-exponential}.
\end{remark}

\subsection{A strip-shift formula for half-plane analytic functions}

Functions with branch points or poles outside the right half-plane are handled by the following strip-shift formula, which separates truncation error from the contour-shift remainder.

\begin{theorem}[Vanilla LCHM strip formula]
\label{thm:plain-strip}
Let $L\succeq0$, let $y_0>1$, and let $\varphi$ be meromorphic on the closed strip
\[
\mathcal S_{y_0}=\{z\in\C:-y_0\leq\ImPart z\leq0\}
\]
with a unique pole at $-\ii$ and residue $\ii/(2\pi)$.  Let $g$ be holomorphic on an open set containing
\[
\bigcup_{z\in\mathcal S_{y_0}}
W\!\left(\ii(H+zL)\right),
\]
and assume the two horizontal integrals below converge absolutely and the vertical sides vanish as their real parts tend to $\pm\infty$.  Then
\begin{equation}
\begin{aligned}
g(A)={}&\int_{\R}\varphi(k)
 g\!\left(\ii(H+kL)\right)\,\dd k\\
&-\int_{\R}\varphi(x-\ii y_0)
 g\!\left(y_0L+\ii(H+xL)\right)\,\dd x.
\end{aligned}
\label{eq:plain-strip-master}
\end{equation}
For a finite truncation radius $R$,
\begin{equation}
\begin{aligned}
&\int_{-R}^{R}\varphi(k)g\!\left(\ii(H+kL)\right)\,\dd k-g(A)\\
={}&-\int_{\R\setminus[-R,R]}\varphi(k)g\!\left(\ii(H+kL)\right)\,\dd k
+\int_{\R}\varphi(x-\ii y_0)
 g\!\left(y_0L+\ii(H+xL)\right)\,\dd x.
\end{aligned}
\label{eq:plain-strip-error}
\end{equation}
\end{theorem}

\begin{proof}
For $z=x-\ii y$ with $0\leq y\leq y_0$,
\[
\ii(H+zL)=yL+\ii(H+xL)
\]
has numerical range in the closed right half-plane.  The stated holomorphy assumption makes the matrix-valued integrand analytic in the strip away from $-\ii$.  Apply the residue theorem to the clockwise rectangle with top edge $[-R,R]$ and bottom edge $[R-\ii y_0,-R-\ii y_0]$.  The residue contribution is $g(M_{-\ii})=g(A)$.  Letting $R\to\infty$ gives \eqref{eq:plain-strip-master}; adding and subtracting the real-axis tail gives \eqref{eq:plain-strip-error}.
\end{proof}

\begin{remark}[Boundary functional calculus]
If $g$ is only continuous on the imaginary axis and holomorphic in the open right half-plane, one may apply the theorem to $g_\eta(z)=g(z+\eta)$ and let $\eta\downarrow0$.  Accretive matrices have semisimple eigenvalues on the imaginary axis, so no unavailable boundary derivatives are required; see \Cref{lem:accretive-boundary}.
\end{remark}

\subsection{Discretization and normalization structure}

After choosing a truncation interval and quadrature, vanilla LCHM has the form
\begin{equation}
g(A)\approx\sum_{j=1}^{M}a_j\,
 g\!\left(\ii(H+k_jL)\right),
\qquad \abs{k_j}\leq K.
\label{eq:plain-discrete}
\end{equation}
Set
\begin{equation}
s_K=\norm{H}+K\norm{L},
\qquad
Y_j=\frac{H+k_jL}{s_K},
\qquad
\norm{Y_j}\leq1.
\label{eq:plain-normalized-angle}
\end{equation}
A block encoding of $Y_j$ is obtained by a linear combination of unitaries (LCU) of block encodings of $H$ and $L$, and QSVT can then implement a bounded polynomial of $Y_j$.

\begin{proposition}[Normalization in a standard vanilla LCHM implementation]
\label{prop:plain-normalization}
For the monomial $g(z)=z^m$, the standard QSVT--LCU realization of \eqref{eq:plain-discrete} has block-encoding normalization
\begin{equation}
\alpha_{\mathrm{vanilla},m}
=s_K^m\sum_{j=1}^{M}\abs{a_j}.
\label{eq:plain-monomial-normalization}
\end{equation}
For a polynomial $p(z)=\sum_{\ell=0}^{d}b_\ell z^\ell$, a direct coefficientwise implementation obeys
\begin{equation}
\alpha_{\mathrm{vanilla},p}
\leq
\left(\sum_{\ell=0}^{d}\abs{b_\ell}s_K^\ell\right)
\sum_{j=1}^{M}\abs{a_j}.
\label{eq:plain-polynomial-normalization}
\end{equation}
Thus even when the kernel has bounded $\ell_1$-norm, the generic normalization can grow as $s_K^d$.
\end{proposition}

\begin{proof}
Since
\[
\left(\ii(H+k_jL)\right)^m=(\ii s_K)^mY_j^m,
\]
QSVT implements $Y_j^m$ with unit normalization and restoring the physical scale contributes $s_K^m$.  The outer LCU contributes the sum of absolute weights.  The polynomial estimate follows by the triangle inequality over its monomials.
\end{proof}

Because the truncation radius $K$ generally increases as the target error decreases, the normalization in \eqref{eq:plain-polynomial-normalization} is an important cost parameter for generic polynomials.  Matrix exponentials display the most favorable structure: if $g(z)=\ee^{-tz}$ and $L\succeq0$, every transform in \eqref{eq:plain-discrete} is the unitary Hamiltonian evolution $\ee^{-\ii t(H+k_jL)}$, and the factor $s_K^d$ is replaced by Hamiltonian-simulation time.  For polynomial and disk-algebra transforms, the bounded Weyl angle matrices of the next section provide an alternative with no $K$-dependent normalization.

\section{Weyl LCHM: exact angular projection}
\label{sec:angular-calculus}

The Weyl LCHM identities are the algebraic core of the paper.  With
\begin{equation}
A=L+\ii H,
\qquad
X_\theta=\cos\theta\,L+\sin\theta\,H
=\frac{\ee^{-\ii\theta}A+\ee^{\ii\theta}A^\dagger}{2},
\label{eq:weyl-main-angle}
\end{equation}
one has, for every integer $m\geq1$,
\begin{equation}
A^m=\frac{2}{\pi}\int_0^\pi
\ee^{\ii m\theta}T_m(X_\theta)\,\dd\theta
=\frac{2}{N}\sum_{j=0}^{N-1}
\ee^{\ii m\theta_j}T_m(X_{\theta_j}),
\quad
\theta_j=\frac{\pi j}{N},\quad N>m.
\label{eq:weyl-master-box}
\end{equation}
Thus the angular discretization is exact at the polynomial degree of interest.  If $w(A)\leq1$, then every angular Hermitian matrix $X_\theta$ has norm at most one and the boundary matrices
\begin{equation}
U_{\theta,\pm}
=\ee^{\ii\theta}\left(
X_\theta\pm\ii(\Id-X_\theta^2)^{1/2}
\right)
\label{eq:weyl-unitary-boundary-box}
\end{equation}
are unitary.  For $g$ in the disk algebra this gives the second master identity
\begin{equation}
 g(A)=\frac1\pi\int_0^\pi
 \left[g(U_{\theta,+})+g(U_{\theta,-})-g(0)\Id\right] \,\dd\theta.
\label{eq:weyl-disk-box}
\end{equation}
\subsection{Angular Hermitian matrices and Fourier projection}
\label{subsec:projection}

Let $A\in\C^{n\times n}$ and define $L$, $H$, and $X_\theta$ by \eqref{eq:intro-cartesian} and \eqref{eq:intro-angle-matrix}.  All matrix norms below are operator $2$-norms.  The numerical range and numerical radius are
\begin{equation}
W(A)=\{\ip{x}{Ax}:\norm{x}=1\},
\qquad
w(A)=\max_{z\in W(A)}\abs{z}.
\label{eq:numerical-range}
\end{equation}
We write $\Disk=\{z:\abs{z}<1\}$ and use $\partial\Disk$ for its boundary.  For a scalar function $f$ on a compact set $E$, let
\begin{equation}
\norm{f}_{\infty,E}=\sup_{z\in E}\abs{f(z)},
\qquad
\Pcal_d=\{P\in\C[z]:\deg P\leq d\}.
\label{eq:scalar-norm-polynomial-space}
\end{equation}

The angle matrices encode the support function of the numerical range.

\begin{lemma}[Numerical radius from angle matrices]
\label{lem:angle-radius}
For every $A\in\C^{n\times n}$,
\begin{equation}
\sup_{\theta\in[0,\pi]}\norm{X_\theta}=w(A).
\label{eq:angle-radius}
\end{equation}
In particular, $w(A)\leq1$ if and only if every angular Hermitian matrix $X_\theta$ has norm at most one.
\end{lemma}

\begin{proof}
Since $X_{\theta+\pi}=-X_\theta$,
\begin{align*}
\sup_{\theta\in[0,\pi]}\norm{X_\theta}
&=\sup_{\theta\in[0,2\pi]}\lambda_{\max}(X_\theta)\\
&=\sup_{\norm{x}=1}\sup_{\theta\in[0,2\pi]}
\RePart\!\left(\ee^{-\ii\theta}\ip{x}{Ax}\right)\\
&=\sup_{\norm{x}=1}\abs{\ip{x}{Ax}}
=w(A).
\end{align*}
\end{proof}

The exact angular projection has continuous and discrete forms.

\begin{theorem}[Exact Chebyshev angular projection]
\label{thm:exact-projection}
Let $A\in\C^{n\times n}$ and $m\geq1$.  Then
\begin{equation}
A^m=\frac{2}{\pi}\int_0^\pi
\ee^{\ii m\theta}T_m(X_\theta)\,\dd\theta.
\label{eq:continuous-monomial}
\end{equation}
If $N>m$ and $\theta_j=\pi j/N$, then
\begin{equation}
A^m=\frac{2}{N}\sum_{j=0}^{N-1}
\ee^{\ii m\theta_j}T_m(X_{\theta_j}).
\label{eq:discrete-monomial}
\end{equation}
\end{theorem}

\begin{proof}
The leading term of the first-kind Chebyshev polynomial is
\[
T_m(x)=2^{m-1}x^m+
\text{terms of degree at most $m-2$ with the same parity}.
\]
Because $X_\theta$ contains only the Fourier modes $\ee^{-\ii\theta}$ and $\ee^{\ii\theta}$, one can write
\begin{equation}
T_m(X_\theta)=
\sum_{\substack{\ell=-m\\ \ell\equiv m\; (2)}}^m
\ee^{\ii\ell\theta}C_{m,\ell},
\label{eq:fourier-expansion}
\end{equation}
where the coefficients are noncommutative polynomials in $A$ and $A^\dagger$.  The mode $-m$ can only arise from the leading term, with every factor of $X_\theta$ contributing $2^{-1}\ee^{-\ii\theta}A$.  Hence
\begin{equation}
C_{m,-m}=2^{m-1}2^{-m}A^m=\frac12A^m.
\label{eq:extremal-coefficient}
\end{equation}
After multiplication by $2\ee^{\ii m\theta}$, the frequencies are $0,2,\ldots,2m$.  Integration over $[0,\pi]$ retains only the zero mode, proving \eqref{eq:continuous-monomial}.

For the discrete formula, set $\omega_N=\ee^{2\pi\ii/N}$.  Since $m<N$,
\[
\frac1N\sum_{j=0}^{N-1}\ee^{2\ii r\theta_j}
=\frac1N\sum_{j=0}^{N-1}\omega_N^{rj}
=\delta_{r0},
\qquad 0\leq r\leq m.
\]
The root-of-unity average therefore performs the same projection.
\end{proof}

Let
\begin{equation}
p(z)=\sum_{m=0}^d a_m z^m.
\label{eq:poly-p}
\end{equation}
Applying \Cref{thm:exact-projection} term by term gives, for every $N>d$,
\begin{equation}
p(A)=\frac1N\sum_{j=0}^{N-1}q_{p,\theta_j}(X_{\theta_j}),
\label{eq:polynomial-angular}
\end{equation}
where
\begin{equation}
q_{p,\theta}(x)=a_0+2\sum_{m=1}^d a_m\ee^{\ii m\theta}T_m(x).
\label{eq:angle-polynomial}
\end{equation}
Define the lifted polynomial
\begin{equation}
F_p(z)=2p(z)-p(0).
\label{eq:polynomial-lift}
\end{equation}
For $x=\cos\varphi$,
\begin{align}
q_{p,\theta}(\cos\varphi)
&=p(\ee^{\ii(\theta+\varphi)})
+p(\ee^{\ii(\theta-\varphi)})-p(0)
\label{eq:boundary-p}\\
&=\frac12F_p(\ee^{\ii(\theta+\varphi)})
+\frac12F_p(\ee^{\ii(\theta-\varphi)}).
\label{eq:boundary-lift}
\end{align}
Equation~\eqref{eq:boundary-lift} is the link between unit-circle approximation and all angle matrices.

\subsection{Analytic and disk-algebra functions}
\label{subsec:functional-calculus}

The angular projection identity first extends to functions whose Taylor coefficients dominate the Chebyshev growth.  Define
\begin{equation}
\chi(r)=
\begin{cases}
1, & 0\leq r\leq1,\\
r+\sqrt{r^2-1}, & r>1.
\end{cases}
\label{eq:chi}
\end{equation}
If $X=X^\dagger$ and $\norm{X}\leq r$, then the spectral theorem gives
\begin{equation}
\norm{T_m(X)}\leq\chi(r)^m.
\label{eq:chebyshev-growth}
\end{equation}

\begin{theorem}[Angular series for analytic functions]
\label{thm:analytic-angular}
Let $g(z)=\sum_{m=0}^\infty a_mz^m$ be holomorphic in $\abs{z}<R_0$, and let $A\in\C^{n\times n}$ satisfy $R_0>\chi(w(A))$.  Then the series
\begin{equation}
a_0\Id+2\sum_{m=1}^\infty
a_m\ee^{\ii m\theta}T_m(X_\theta)
\label{eq:analytic-angle-series}
\end{equation}
converges uniformly in operator norm for $\theta\in[0,\pi]$, and
\begin{equation}
g(A)=\frac1\pi\int_0^\pi
\left[
a_0\Id+2\sum_{m=1}^\infty
a_m\ee^{\ii m\theta}T_m(X_\theta)
\right] \,\dd\theta.
\label{eq:analytic-angular}
\end{equation}
In particular, the result holds for every entire function and every finite-dimensional $A$.
\end{theorem}

\begin{proof}
Choose $R$ with $\chi(w(A))<R<R_0$ and set
$M_g(R)=\max_{\abs{z}=R}\abs{g(z)}$.  Since the spectral radius $r(A)$ satisfies
$r(A)\leq w(A)<R_0$, the Taylor series of $g$ converges in operator norm at $A$ and represents the standard holomorphic functional calculus:
\[
g(A)=\sum_{m=0}^\infty a_mA^m.
\]
Cauchy's estimate gives $\abs{a_m}\leq M_g(R)R^{-m}$, while \Cref{lem:angle-radius} and \eqref{eq:chebyshev-growth} give
\[
\sup_\theta\norm{T_m(X_\theta)}
\leq\chi(w(A))^m.
\]
The series in \eqref{eq:analytic-angle-series} is therefore uniformly absolutely convergent and may be integrated term by term.  Applying \Cref{thm:exact-projection} to every Taylor coefficient yields \eqref{eq:analytic-angular}.
\end{proof}

The associated truncation error satisfies
\begin{equation}
\begin{aligned}
&\sup_{\theta\in[0,\pi]}
\norm{
2\sum_{m=d+1}^\infty
a_m\ee^{\ii m\theta}T_m(X_\theta)
}\\
&\qquad\leq
\frac{2M_g(R)}{1-\chi(w(A))/R}
\left(\frac{\chi(w(A))}{R}\right)^{d+1}.
\end{aligned}
\label{eq:analytic-tail}
\end{equation}

When $w(A)\leq1$, the angular representation admits a unitary boundary form.  Since $(\Id-X_\theta^2)^{1/2}$ is a continuous function of $X_\theta$, the matrices
\begin{equation}
X_\theta\pm\ii(\Id-X_\theta^2)^{1/2}
\label{eq:boundary-unitaries}
\end{equation}
are unitary and mutually inverse.  The scalar identity
$T_m(x)=\tfrac12[(x+\ii\sqrt{1-x^2})^m+(x-\ii\sqrt{1-x^2})^m]$ implies
\begin{equation}
T_m(X_\theta)
=\frac12\left[
\left(X_\theta+\ii(\Id-X_\theta^2)^{1/2}\right)^m
+\left(X_\theta-\ii(\Id-X_\theta^2)^{1/2}\right)^m
\right].
\label{eq:chebyshev-unitary}
\end{equation}

Let
\begin{equation}
\DiskAlg=\{g\in C(\overline\Disk):g\text{ is holomorphic in }\Disk\}
\label{eq:disk-algebra}
\end{equation}
be the disk algebra.

\begin{definition}[Disk-algebra functional calculus]
\label{def:disk-calculus}
Let $w(A)\leq1$ and $g\in\DiskAlg$.  Choose polynomials $r_\ell$ such that
\begin{equation}
\norm{r_\ell-g}_{\infty,\overline\Disk}\longrightarrow0.
\label{eq:disk-polynomial-approximation}
\end{equation}
We define
\begin{equation}
g(A)=\lim_{\ell\to\infty}r_\ell(A)
\label{eq:disk-calculus-definition}
\end{equation}
in operator norm.  The Okubo--Ando inequality
\begin{equation}
\norm{r(A)}\leq2\norm{r}_{\infty,\overline\Disk}
\label{eq:okubo-ando-polynomial}
\end{equation}
for polynomials $r$ makes the limit well defined and independent of the approximating sequence~\cite{OkuboAndo1975}.  If $g$ is holomorphic on a neighborhood of $\overline\Disk$, this definition agrees with the standard holomorphic matrix functional calculus.
\end{definition}

\begin{theorem}[Weyl LCHM identity in the disk algebra]
\label{thm:disk-algebra}
Let $w(A)\leq1$ and $g\in\DiskAlg$.  Then
\begin{equation}
\begin{aligned}
g(A)=\frac1\pi\int_0^\pi
\Bigl[{}
&g\!\left(\ee^{\ii\theta}
 \left(X_\theta+\ii(\Id-X_\theta^2)^{1/2}\right)\right)\\
&+g\!\left(\ee^{\ii\theta}
 \left(X_\theta-\ii(\Id-X_\theta^2)^{1/2}\right)\right)
-g(0)\Id
\Bigr] \,\dd\theta.
\end{aligned}
\label{eq:disk-algebra-formula}
\end{equation}
With $F_g(z)=2g(z)-g(0)$,
\begin{equation}
\norm{g(A)}
\leq\min\left\{
2\norm{g}_{\infty,\overline\Disk},
\norm{F_g}_{\infty,\partial\Disk}
\right\}.
\label{eq:disk-norm-bounds}
\end{equation}
\end{theorem}

\begin{proof}
Combining \eqref{eq:continuous-monomial} with \eqref{eq:chebyshev-unitary} gives
\begin{equation}
\begin{aligned}
A^m=\frac1\pi\int_0^\pi
\Bigl[{}
&\left(\ee^{\ii\theta}
 \left(X_\theta+\ii(\Id-X_\theta^2)^{1/2}\right)\right)^m\\
&+\left(\ee^{\ii\theta}
 \left(X_\theta-\ii(\Id-X_\theta^2)^{1/2}\right)\right)^m
\Bigr] \,\dd\theta,
\qquad m\geq1.
\end{aligned}
\label{eq:monomial-boundary}
\end{equation}
Thus \eqref{eq:disk-algebra-formula} holds for polynomials.  Let $r_\ell$ be as in \eqref{eq:disk-polynomial-approximation}.  The two matrices in the integrand are unitary, so continuous functional calculus gives uniform convergence of the corresponding boundary terms.  Passing to the limit yields \eqref{eq:disk-algebra-formula}, while \eqref{eq:okubo-ando-polynomial} gives the first bound in \eqref{eq:disk-norm-bounds}.

For the second bound, rewrite \eqref{eq:disk-algebra-formula} as
\begin{equation}
\begin{aligned}
g(A)=\frac1{2\pi}\int_0^\pi
\Bigl[{}
&F_g\!\left(\ee^{\ii\theta}
 \left(X_\theta+\ii(\Id-X_\theta^2)^{1/2}\right)\right)\\
&+F_g\!\left(\ee^{\ii\theta}
 \left(X_\theta-\ii(\Id-X_\theta^2)^{1/2}\right)\right)
\Bigr] \,\dd\theta.
\end{aligned}
\label{eq:pure-hermitian-form}
\end{equation}
The spectral theorem and continuous functional calculus bound each summand by $\norm{F_g}_{\infty,\partial\Disk}$.
\end{proof}

\section{Quantum eigenvalue transformation algorithm}
\label{sec:quantum-realization}

\subsection{Lifted approximation and exact angular realization}
\label{subsec:lifted-approximation}

The disk-algebra formula turns approximation of the scalar lift into approximation of the matrix function with no loss in the constant.

\begin{theorem}[Lifted approximation and exact angular realization]
\label{thm:lifted-approximation}
Let $w(A)\leq1$, let $g\in\DiskAlg$, and let
\begin{equation}
P_d(z)=\sum_{m=0}^d c_mz^m
\label{eq:Pd-coefficients}
\end{equation}
be a polynomial of degree at most $d$.  Define
\begin{equation}
p_d(z)=\frac{P_d(z)+P_d(0)}2
=c_0+\frac12\sum_{m=1}^dc_mz^m.
\label{eq:pd-from-Pd}
\end{equation}
Then $F_{p_d}=P_d$ and
\begin{equation}
\norm{g(A)-p_d(A)}
\leq\norm{F_g-P_d}_{\infty,\partial\Disk}.
\label{eq:constant-one-transfer}
\end{equation}
Moreover, for every $N>d$ and $\theta_j=\pi j/N$,
\begin{equation}
p_d(A)=\frac1N\sum_{j=0}^{N-1}
q_{d,\theta_j}(X_{\theta_j}),
\label{eq:exact-finite-realization}
\end{equation}
where
\begin{equation}
q_{d,\theta}(x)=c_0+\sum_{m=1}^dc_m\ee^{\ii m\theta}T_m(x).
\label{eq:qd-theta}
\end{equation}
\end{theorem}

\begin{proof}
The identity $F_{p_d}=P_d$ follows directly from \eqref{eq:pd-from-Pd}.  Applying the second estimate in \eqref{eq:disk-norm-bounds} to $h=g-p_d$ gives
\[
\norm{h(A)}\leq\norm{F_h}_{\infty,\partial\Disk}
=\norm{F_g-P_d}_{\infty,\partial\Disk},
\]
which is \eqref{eq:constant-one-transfer}.  Equation~\eqref{eq:exact-finite-realization} follows from \eqref{eq:discrete-monomial} applied to the coefficients in \eqref{eq:pd-from-Pd}.
\end{proof}

\begin{definition}[Lift normalizations and QET constants]
\label{def:qet-constant}
Let $p_d\not\equiv0$ be a polynomial, let $g\in\DiskAlg$ be nonzero, and let $\ket{\psi}$ be normalized.  Define
\begin{equation}
C_p:=\norm{2p_d-p_d(0)}_{\infty,\partial\Disk},
\qquad
C_g:=\norm{2g-g(0)}_{\infty,\partial\Disk}.
\label{eq:Cp-Cg-definition}
\end{equation}
Whenever the corresponding output is nonzero, define
\begin{equation}
q_{\mathrm{et}}(p_d;A,\psi)
:=\frac{\norm{p_d}_{\infty,\overline\Disk}}
{\norm{p_d(A)\ket{\psi}}},
\qquad
q_{\mathrm{et}}(g;A,\psi)
:=\frac{\norm{g}_{\infty,\overline\Disk}}
{\norm{g(A)\ket{\psi}}}.
\label{eq:qet-constant-definition}
\end{equation}
\end{definition}

If $P_d=2p_d-p_d(0)$ satisfies
$\norm{F_g-P_d}_{\infty,\partial\Disk}\leq\varepsilon_{\mathrm{app}}$, then the reverse triangle inequality gives
\begin{equation}
\abs{C_p-C_g}\leq\varepsilon_{\mathrm{app}}.
\label{eq:Cp-controlled-by-Cg}
\end{equation}
This is the main reason for indexing the normalization by the transformed polynomial or target function rather than by the approximation degree: for a convergent disk-algebra approximation, $C_p$ remains within the prescribed approximation error of the target-controlled quantity $C_g$.

If
\begin{equation}
E_d(F_g)=\inf_{P\in\Pcal_d}\norm{F_g-P}_{\infty,\partial\Disk},
\label{eq:best-error}
\end{equation}
then \Cref{thm:lifted-approximation} gives a degree-$d$ polynomial $p_d$ satisfying
\begin{equation}
\norm{g(A)-p_d(A)}\leq E_d(F_g).
\label{eq:best-error-matrix}
\end{equation}
The passage from $P_d$ to the finite family of angle matrices adds no angular error.

For analytic lifts, standard Taylor approximation gives an explicit degree estimate.  Suppose $R>1$ and $F_g$ is holomorphic on a neighborhood of the closed disk $\{z:\abs{z}\leq R\}$, and set
\begin{equation}
M_F(R)=\max_{\abs{z}=R}\abs{F_g(z)}.
\label{eq:MF}
\end{equation}
The degree-$d$ Taylor polynomial satisfies
\begin{equation}
\norm{F_g-P_d}_{\infty,\overline\Disk}
\leq\frac{M_F(R)}{R-1}R^{-d}.
\label{eq:Taylor-lift-bound}
\end{equation}
Thus a sufficient choice for error $\varepsilon_{\mathrm{app}}$ is
\begin{equation}
d\geq
\max\!\left\{0,
\left\lceil
\frac{\log\!\left(M_F(R)/((R-1)\varepsilon_{\mathrm{app}})\right)}
{\log R}
\right\rceil\right\}.
\label{eq:degree-bound}
\end{equation}

\begin{remark}[Direct approximation for contractions]
\label{rem:direct-approximation}
If $\norm{A}\leq1$, von Neumann's inequality~\cite{SzNagyFoias2010} gives
\begin{equation}
\norm{g(A)-r_d(A)}
\leq\norm{g-r_d}_{\infty,\overline\Disk}
\label{eq:von-neumann}
\end{equation}
for every analytic polynomial $r_d$.  The angular construction still acts on the lifted unit-circle polynomial $2r_d-r_d(0)$.  One may therefore compare direct approximation of $g$ with direct approximation of $F_g$ and retain the lower-degree choice.
\end{remark}

\subsection{Coherent angle-GQSP construction}
\label{subsec:coherent-construction}

For the quantum construction, assume $\norm{A}\leq1$ and let the supplied block encoding have normalization $\alpha=1$.  The circuit therefore acts directly on $A$ in the standard block-encoding oracle model~\cite{GilyenEtAl2019}.

\begin{definition}[Unit-normalized block encoding]
\label{def:block-encoding}
A unitary $U_A$ acting on $a$ ancilla qubits and a system register is a $(1,a,\delta_A)$ block encoding of $A$ if
\begin{equation}
\norm{A-(\bra{0^a}\otimes\Id)U_A(\ket{0^a}\otimes\Id)}
\leq\delta_A.
\label{eq:block-encoding}
\end{equation}
\end{definition}

Denote the principal block of $U_A$ by
\begin{equation}
\widetilde A=(\bra{0^a}\otimes\Id)U_A(\ket{0^a}\otimes\Id),
\qquad
\norm{\widetilde A-A}\leq\delta_A.
\label{eq:Atilde}
\end{equation}
Since $\widetilde A$ is a compression of a unitary, $\norm{\widetilde A}\leq1$.

Choose
\begin{equation}
N=2^b>d,
\qquad
b=\left\lceil\log_2(d+1)\right\rceil,
\qquad
\theta_j=\frac{\pi j}{N}.
\label{eq:angle-register}
\end{equation}
The diagonal phase operations
\begin{equation}
\Phi_\pm\ket{j}=\ee^{\pm\ii\pi j/N}\ket{j}
\label{eq:angle-phases}
\end{equation}
are implemented by $\mathcal{O}(b)$ single-bit-controlled rotations.

For each angle branch, define
\begin{equation}
\widetilde X_j
=\frac{\ee^{-\ii\theta_j}\widetilde A
+\ee^{\ii\theta_j}\widetilde A^\dagger}{2}.
\label{eq:Xtilde-j}
\end{equation}
Introduce a direct-standard-form qubit $h$, set
$B_j=\ee^{-\ii\theta_j}U_A$, and define
\begin{align}
\mathcal V_j
&=\Outer{0}{1}_h\otimes B_j
 +\Outer{1}{0}_h\otimes B_j^\dagger
\notag\\
&=\left(
\Projector{0}_h\otimes B_j
 +\Projector{1}_h\otimes B_j^\dagger
\right)(X_h\otimes\Id).
\label{eq:Vj}
\end{align}
The angle-coherent version is
\begin{equation}
\sum_{j=0}^{N-1}\Projector{j}\otimes\mathcal V_j
=
\Phi_-\otimes\Outer{0}{1}_h\otimes U_A
+
\Phi_+\otimes\Outer{1}{0}_h\otimes U_A^\dagger .
\label{eq:coherent-Vj}
\end{equation}
Thus the direct dilation uses a constant number of selected controlled
calls to $U_A$ or $U_A^\dagger$ and $\mathcal O(b)$ angle-dependent phase
rotations.  Since $B_j$ is unitary, the two off-diagonal blocks in
\eqref{eq:Vj} are adjoints of one another and
\begin{equation}
\mathcal V_j^\dagger=\mathcal V_j,
\qquad
\mathcal V_j^2
=\Projector{0}_h\otimes U_AU_A^\dagger
 +\Projector{1}_h\otimes U_A^\dagger U_A
=\Id .
\label{eq:Vj-involution}
\end{equation}
For
\begin{equation}
\Pi_X=\Projector{+}_h\otimes\Projector{0^a}\otimes\Id,
\label{eq:PiX}
\end{equation}
direct compression gives
\begin{align}
&(\bra{+}_h\otimes\bra{0^a}\otimes\Id)
\mathcal V_j
(\ket{+}_h\otimes\ket{0^a}\otimes\Id)
\notag\\
&\qquad=
\frac{\ee^{-\ii\theta_j}\widetilde A
+\ee^{\ii\theta_j}\widetilde A^\dagger}{2}
=\widetilde X_j .
\label{eq:Vj-compression}
\end{align}
Equivalently,
\begin{equation}
\Pi_X\mathcal V_j\Pi_X=\widetilde X_j\Pi_X.
\label{eq:standard-form}
\end{equation}
The corresponding signal reflection requires no additional register because
\begin{equation}
2\Pi_X-\Id
=(\mathsf H_h\otimes\Id)
\left(
2\Projector{0}_h\otimes\Projector{0^a}\otimes\Id-\Id
\right)
(\mathsf H_h\otimes\Id).
\label{eq:PiX-reflection}
\end{equation}
Following qubitization~\cite{LowChuang2019}, define the walk
\begin{equation}
W_j=(2\Pi_X-\Id)\mathcal V_j.
\label{eq:Wj}
\end{equation}
Its two-dimensional invariant-subspace structure gives
\begin{equation}
\Pi_XW_j^m\Pi_X=T_m(\widetilde X_j)\Pi_X,
\qquad m\geq0.
\label{eq:qubitized-chebyshev}
\end{equation}
A self-contained proof is included in \Cref{app:qubitization}.

Let $P_d$ be as in \eqref{eq:Pd-coefficients}, assume $p_d\not\equiv0$, and define
\begin{equation}
C_p
:=\norm{P_d}_{\infty,\partial\Disk}
=\norm{2p_d-p_d(0)}_{\infty,\partial\Disk}.
\label{eq:normalization-factor}
\end{equation}
Define the angle-controlled signal unitary
\begin{equation}
\mathsf Z
=\sum_{j=0}^{N-1}\Projector{j}\otimes
\ee^{\ii\theta_j}W_j.
\label{eq:global-signal}
\end{equation}
It is unitary and satisfies
\begin{equation}
P_d(\mathsf Z)
=\sum_{j=0}^{N-1}\Projector{j}\otimes
P_d(\ee^{\ii\theta_j}W_j).
\label{eq:global-polynomial}
\end{equation}
GQSP implements $P_d(\mathsf Z)/C_p$ in a designated ancilla block using $\mathcal{O}(d)$ calls to $\mathsf Z$ or $\mathsf Z^\dagger$~\cite{MotlaghWiebe2024}.  On the $j$th angle branch,
\begin{align}
\Pi_XP_d(\ee^{\ii\theta_j}W_j)\Pi_X
&=\left[c_0\Id+
\sum_{m=1}^dc_m\ee^{\ii m\theta_j}T_m(\widetilde X_j)\right]\Pi_X
\label{eq:GQSP-angle}\\
&=q_{d,\theta_j}(\widetilde X_j)\Pi_X.
\end{align}

Prepare the uniform angle state
\begin{equation}
\ket{\mathrm{ang}}=\frac1{\sqrt N}\sum_{j=0}^{N-1}\ket{j}
=\mathsf H^{\otimes b}\ket{0^b}.
\label{eq:angle-state}
\end{equation}
Equations~\eqref{eq:global-polynomial} and \eqref{eq:GQSP-angle} give
\begin{align}
&(\bra{\mathrm{ang}}\otimes\Pi_X)
P_d(\mathsf Z)
(\ket{\mathrm{ang}}\otimes\Pi_X)
\notag\\
&\qquad=\frac1N\sum_{j=0}^{N-1}
q_{d,\theta_j}(\widetilde X_j)\Pi_X
=p_d(\widetilde A)\Pi_X.
\label{eq:coherent-average}
\end{align}
The last equality is exact because the grid is equispaced and $N>d$.

\begin{theorem}[Weyl LCHM implementation]
\label{thm:coherent-algorithm}
Let $\norm{A}\leq1$, and let $U_A$ be a $(1,a,\delta_A)$ block encoding with coherent access to controlled $U_A$, $U_A^\dagger$, and the required reflections.  Let $P_d$ be a nonzero polynomial of degree at most $d$, define $p_d=(P_d+P_d(0))/2$, and define $C_p$ by \eqref{eq:normalization-factor}.  There is an ideal quantum circuit whose success block is exactly
\begin{equation}
\frac{p_d(\widetilde A)}{C_p}.
\label{eq:ideal-success-block}
\end{equation}
The circuit uses
\begin{align}
Q_A(d)&=\mathcal{O}(d)
&&\text{controlled calls to $U_A$ or $U_A^\dagger$},
\label{eq:coherent-query-count}\\
Q_{\mathrm{rot}}(d)&=\mathcal{O}\!\left(d\log(d+1)\right)
&&\text{angle-dependent one-qubit rotations},
\label{eq:coherent-rotation-count}
\end{align}
and
\begin{equation}
a+\left\lceil\log_2(d+1)\right\rceil+2
\label{eq:ancilla-count}
\end{equation}
logical ancilla qubits.  The two qubits beyond the original block-encoding and angle registers are the direct-standard-form qubit and the single GQSP processing qubit.  Temporary workspace internal to the supplied oracles, reflections, controls, and gate synthesis is excluded.  On a normalized input $\ket{\psi}$, amplitude amplification uses
\begin{equation}
\mathcal O\!\left(q_{\mathrm{et}}(p_d;\widetilde A,\psi)\right)
=\mathcal O\!\left(
\frac{C_p}{\norm{p_d(\widetilde A)\ket{\psi}}}
\right)
\label{eq:coherent-qet-repetitions}
\end{equation}
coherent circuit repetitions.  \Cref{prop:angle-lcu-optimality} shows that this dependence is optimal up to a universal constant in the uniform contraction model.

If, in addition, $g\in\DiskAlg$ and
\begin{equation}
\norm{F_g-P_d}_{\infty,\partial\Disk}\leq\varepsilon_{\mathrm{app}},
\label{eq:approximation-assumption}
\end{equation}
then an implementation with circuit error $\varepsilon_{\mathrm{circ}}$ has normalized success block $\widetilde G$ satisfying
\begin{equation}
\norm{\widetilde G-\frac{g(A)}{C_p}}
\leq
\frac{\varepsilon_{\mathrm{app}}}{C_p}
+\frac{\delta_A}{2C_p}\sum_{m=1}^d m\abs{c_m}
+\varepsilon_{\mathrm{circ}},
\label{eq:coherent-total-error}
\end{equation}
where $P_d(z)=\sum_{m=0}^dc_mz^m$.  In particular, if
\begin{equation}
\varepsilon_{\mathrm{app}}
+\frac{\delta_A}{2}\sum_{m=1}^d m\abs{c_m}
+C_p\varepsilon_{\mathrm{circ}}
\leq\varepsilon,
\label{eq:unnormalized-error-budget}
\end{equation}
then the circuit is a $C_p$-normalized block encoding of $g(A)$ with unnormalized operator error at most $\varepsilon$.  Clifford+$T$ approximation with uniform per-rotation error of order $\varepsilon_{\mathrm{circ}}/[d\log(d+1)]$ contributes
\begin{equation}
\mathcal{O}\!\left(
 d\log(d+1)
 \log\!\frac{d\log(d+1)}{\varepsilon_{\mathrm{circ}}}
\right)
\label{eq:coherent-T-count}
\end{equation}
$T$ gates apart from the matrix-oracle cost.

If $F_g$ is holomorphic on $\abs{z}\leq R$ for some $R>1$, the Taylor choice \eqref{eq:degree-bound} gives
\begin{equation}
Q_A
=\mathcal{O}\!\left(
\frac{\log\!\left(M_F(R)/((R-1)\varepsilon_{\mathrm{app}})\right)}{\log R}
\right).
\label{eq:analytic-query-complexity}
\end{equation}
Thus matrix-query precision dependence is inherited from scalar approximation, whereas Clifford+$T$ rotation-approximation precision enters logarithmically through \eqref{eq:coherent-T-count}.
\end{theorem}

\begin{proof}[Proof of \Cref{thm:coherent-algorithm}]
Equations~\eqref{eq:Vj}--\eqref{eq:standard-form} give a Hermitian involutory standard form for every angle branch using one direct-standard-form qubit, and \eqref{eq:qubitized-chebyshev} gives the corresponding Chebyshev transformation.  The controlled unitary \eqref{eq:global-signal} packages all angle branches into one signal.  GQSP applied to $P_d/C_p$, followed by preparation and unpreparation of \eqref{eq:angle-state}, yields \eqref{eq:coherent-average} divided by $C_p$.

Each call to $\mathsf Z$ or $\mathsf Z^\dagger$ contains one application of the walk or its inverse.  By \eqref{eq:coherent-Vj}, this requires a constant number of selected controlled calls to $U_A$ or $U_A^\dagger$ and $\mathcal{O}(\log(d+1))$ angle-dependent phase rotations; the reflection \eqref{eq:PiX-reflection} requires no matrix-oracle query.  The logical ancillary registers consist of the original $a$ block-encoding qubits, the $\lceil\log_2(d+1)\rceil$-qubit angle register, one direct-standard-form qubit, and one GQSP processing qubit.  Preparing the $\ket{+}_h$ signal state requires only a Hadamard and no additional register.

The uniform angle state contributes amplitudes $1/\sqrt N$ on both sides of the matrix element, producing the exact average $N^{-1}\sum_j$.  Its coefficient $\ell_1$-norm is therefore $N\cdot N^{-1}=1$.  The success-amplitude statement follows by applying the success block to $\ket{\psi}$.  Because \eqref{eq:Vj-compression} is exact, the direct dilation introduces neither an additional normalization nor an additional approximation error.  The approximation and block-encoding terms in \eqref{eq:coherent-total-error} therefore follow unchanged from \Cref{thm:lifted-approximation} and the telescoping estimate \eqref{eq:polynomial-perturbation}; the final term is the circuit implementation error.  Multiplication by $C_p$ gives \eqref{eq:unnormalized-error-budget}.  The rotation count and ancilla-free Clifford+$T$ approximation of the axial rotations give \eqref{eq:coherent-T-count}, and the Taylor degree estimate \eqref{eq:degree-bound} gives \eqref{eq:analytic-query-complexity}.
\end{proof}

\begin{corollary}[Normalized output-state preparation]
\label{cor:output-state-preparation}
Assume an exact block encoding and an ideal circuit, and suppose reflections about the prepared input and the success subspace are available.  Let $\ket{\psi}$ be normalized and suppose $p_d(A)\ket{\psi}\neq0$.  One circuit application followed by postselection prepares
\begin{equation}
\frac{p_d(A)\ket{\psi}}{\norm{p_d(A)\ket{\psi}}}
\label{eq:normalized-polynomial-output}
\end{equation}
with probability
\begin{equation}
P_{\mathrm{succ}}
=\frac{\norm{p_d(A)\ket{\psi}}^2}{C_p^2}.
\label{eq:success-prob}
\end{equation}
There is no additional factor $1/N$ from the coherent angular average.  Repeating independent postselection therefore takes an expected
\begin{equation}
\Theta\!\left(
\frac{C_p^2}{\norm{p_d(A)\ket{\psi}}^2}
\right)
\label{eq:independent-postselection-repetitions}
\end{equation}
circuit applications.

Amplitude amplification~\cite{BrassardEtAl2002} prepares the same state at constant accuracy using
\begin{equation}
\mathcal{O}\!\left(
\frac{C_p}{\norm{p_d(A)\ket{\psi}}}
\right)
=\mathcal O\!\left(q_{\mathrm{et}}(p_d;A,\psi)\right)
\label{eq:amplification-repetitions}
\end{equation}
coherent circuit repetitions and hence
\begin{equation}
\mathcal{O}\!\left(
\frac{dC_p}{\norm{p_d(A)\ket{\psi}}}
\right)
\label{eq:state-prep-cost}
\end{equation}
controlled matrix-oracle calls.  The same statement holds with $p_d(A)$ replaced by $g(A)$ when the operator-approximation error is small relative to $\norm{g(A)\ket{\psi}}$; an explicit normalized-state bound is given in \eqref{eq:normalized-state-error}.
\end{corollary}

\begin{proposition}[Degree-optimal sequential signal-query scaling]
\label{prop:query-optimality}
Suppose the leading coefficient of $P_d$ is nonzero.  In a sequential signal-processing circuit composed of signal-independent unitaries interleaved with $q$ calls to $\mathsf Z$ or $\mathsf Z^\dagger$, every matrix element is a Laurent polynomial in the signal eigenvalue with exponents between $-q$ and $q$.  Exact realization of $P_d$ therefore requires $q\geq d$.  The $\mathcal{O}(d)$ signal-query depth in \Cref{thm:coherent-algorithm} is consequently optimal as a function of degree in this model.
\end{proposition}

\begin{proof}
The claim follows by induction on the number of signal calls.  Signal-independent unitaries do not change the Laurent degree, while multiplication by $\mathsf Z$ or $\mathsf Z^\dagger$ changes each exponent by at most one.  A target with a nonzero $z^d$ coefficient cannot be represented when $q<d$.
\end{proof}

The $N$ equispaced angular branches are evaluated in superposition rather than in a classical loop.  A direct implementation uses $\mathcal{O}(d\log(d+1))$ angle-dependent controlled rotations over the complete signal-processing sequence; this count is separate from the internal gate cost of the supplied matrix oracle.

\begin{proposition}[Optimal QET constant and post-selection overhead]
\label{prop:angle-lcu-optimality}
For every nonzero polynomial $p_d$,
\begin{equation}
\sup_{\norm{B}\leq1}\norm{p_d(B)}
=\norm{p_d}_{\infty,\overline\Disk},
\label{eq:minimax-normalization}
\end{equation}
and its Weyl lift normalization satisfies
\begin{equation}
\norm{p_d}_{\infty,\overline\Disk}
\leq C_p
\leq3\norm{p_d}_{\infty,\overline\Disk}.
\label{eq:weyl-three-competitive}
\end{equation}
Likewise,
\begin{equation}
\norm{g}_{\infty,\overline\Disk}
\leq C_g
\leq3\norm{g}_{\infty,\overline\Disk}.
\label{eq:target-lift-three-competitive}
\end{equation}
Consequently, any post-selection-based QET whose success block is valid uniformly for every contraction must, in the worst case, use
\begin{equation}
\Omega\!\left(q_{\mathrm{et}}(p_d;A,\psi)\right)=
\Omega\!\left(
\frac{\norm{p_d}_{\infty,\overline\Disk}}{\norm{p_d(A)\ket{\psi}}}
\right)=
\Omega\!\left(
\frac{C_p}{\norm{p_d(A)\ket{\psi}}}
\right)
\label{eq:polynomial-qet-lower-bound}
\end{equation}
coherent success-amplitude amplification steps.

Suppose in addition that $g\in\DiskAlg$ and
\begin{equation}
\norm{F_g-P_d}_{\infty,\partial\Disk}
\leq\varepsilon_{\mathrm{app}}.
\label{eq:optimality-approximation-assumption}
\end{equation}
Then
\begin{equation}
\begin{gathered}
\abs{C_p-C_g}\leq\varepsilon_{\mathrm{app}},
\qquad
\abs{
\norm{p_d}_{\infty,\overline\Disk}
-\norm{g}_{\infty,\overline\Disk}
}
\leq\varepsilon_{\mathrm{app}},\\
\abs{
\norm{p_d(A)\ket{\psi}}
-\norm{g(A)\ket{\psi}}
}
\leq\varepsilon_{\mathrm{app}}.
\end{gathered}
\label{eq:optimality-approximation-transfer}
\end{equation}
If $\norm{g(A)\ket{\psi}}>\varepsilon_{\mathrm{app}}$, it follows that
\begin{equation}
\frac{\max\{\norm{g}_{\infty,\overline\Disk}
-\varepsilon_{\mathrm{app}},0\}}
{\norm{g(A)\ket{\psi}}+\varepsilon_{\mathrm{app}}}
\leq
q_{\mathrm{et}}(p_d;A,\psi)
\leq
\frac{\norm{g}_{\infty,\overline\Disk}
+\varepsilon_{\mathrm{app}}}
{\norm{g(A)\ket{\psi}}-\varepsilon_{\mathrm{app}}}.
\label{eq:optimality-target-ratio}
\end{equation}
Any success-block normalization valid uniformly for all contractions is at least
$\norm{g}_{\infty,\overline\Disk}$, because scalar contractions $B=x\Id$ with $x\in\overline\Disk$ realize every scalar value $g(x)$.  Hence, when only $\norm{A}\leq1$ is known and no further information about $A$ or $\ket{\psi}$ is available, every such method has worst-case post-selection complexity
\begin{equation}
\Omega\!\left(
\frac{\norm{g}_{\infty,\overline\Disk}}
{\norm{g(A)\ket{\psi}}}
\right)
=\Omega\!\left(q_{\mathrm{et}}(g;A,\psi)\right).
\label{eq:target-qet-lower-bound}
\end{equation}
Thus Weyl LCHM attains the optimal post-selection dependence
\begin{equation}
\Theta\!\left(\frac{C_p}{\norm{p_d(A)\ket{\psi}}}\right)
=\Theta\!\left(q_{\mathrm{et}}(p_d;A,\psi)\right),
\qquad\text{or asymptotically}\qquad
\Theta\!\left(\frac{C_g}{\norm{g(A)\ket{\psi}}}\right)
=\Theta\!\left(q_{\mathrm{et}}(g;A,\psi)\right).
\label{eq:optimal-qet-scaling-summary}
\end{equation}

\end{proposition}

\begin{proof}
For every contraction $B$, von Neumann's inequality gives
$\norm{p_d(B)}\leq\norm{p_d}_{\infty,\overline\Disk}$.  Conversely, scalar contractions $B=z\Id$ attain the scalar supremum, proving \eqref{eq:minimax-normalization}.  Since $P_d=2p_d-p_d(0)$,
\begin{equation}
\norm{p_d}_{\infty,\overline\Disk}
\leq\norm{P_d}_{\infty,\partial\Disk}=C_p,
\qquad
C_p
\leq2\norm{p_d}_{\infty,\overline\Disk}+\abs{p_d(0)}
\leq3\norm{p_d}_{\infty,\overline\Disk}.
\label{eq:Cp-three-proof}
\end{equation}
The same calculation with $F_g=2g-g(0)$ proves
\eqref{eq:target-lift-three-competitive}.
A uniformly valid success block $p_d(B)/\alpha$ must be a contraction for every contraction $B$, so
$\alpha\geq\norm{p_d}_{\infty,\overline\Disk}\geq C_p/3$.  Standard success-amplitude amplification lower bounds then give \eqref{eq:polynomial-qet-lower-bound}~\cite{Zalka1999}, while \Cref{cor:output-state-preparation} gives the matching upper bound.

The first estimate in \eqref{eq:optimality-approximation-transfer} is \eqref{eq:Cp-controlled-by-Cg}.  Applying \eqref{eq:constant-one-transfer} to scalar contractions gives
$\norm{g-p_d}_{\infty,\overline\Disk}\leq\varepsilon_{\mathrm{app}}$, so the reverse triangle inequality gives the second estimate.  The third follows from \Cref{thm:lifted-approximation} and the reverse triangle inequality.  Combining the second and third estimates gives \eqref{eq:optimality-target-ratio}.

Finally, a success block for $g(B)/\alpha$ must be a contraction for every contraction $B$.  Taking $B=x\Id$ gives
$\alpha\geq\abs{g(x)}$ for every $x\in\overline\Disk$, hence
$\alpha\geq\norm{g}_{\infty,\overline\Disk}$.  The amplitude-amplification lower bound proves \eqref{eq:target-qet-lower-bound}.
\end{proof}

\section{Examples of matrix functions}
\label{sec:examples-extensions}

Unless stated otherwise, the Weyl LCHM results below assume a unit-normalized block encoding of a contraction $A$.  The circuit in \Cref{thm:coherent-algorithm} implements a degree-$d$ polynomial approximation using $\mathcal{O}(d)$ controlled matrix-oracle calls, $\mathcal{O}(d\log(d+1))$ angle-dependent rotations, and $a+\lceil\log_2(d+1)\rceil+2$ logical ancillas.  Its amplitude-amplified state-preparation cost is the per-attempt query count multiplied by
$C_p/\norm{p_d(A)\ket{\psi}}=\Theta(q_{\mathrm{et}}(p_d;A,\psi))$.  For disk-algebra targets approximated on $\partial\Disk$, the normalization obeys $C_p\leq C_g+\varepsilon_{\mathrm{app}}$ and is therefore controlled by the target function rather than by the degree itself.  Each corollary below states the relevant degree, lift normalization, input-oracle precision, and transformed-output norm explicitly.

\subsection{Matrix exponentials and driven ODEs via vanilla LCHM/LCHS}
\label{subsec:matrix-exponential}

Let $A=L+\ii H$ with $L\succeq0$ and let $t\geq0$.  The exponential specialization of vanilla LCHM is the LCHS identity
\begin{equation}
\ee^{-tA}
=\int_{\R}f_\beta(k)\,
\ee^{-\ii t(H+kL)}\,\dd k,
\qquad 0<\beta<1,
\label{eq:exp-plain-lchm}
\end{equation}
where $f_\beta$ is given by \eqref{eq:plain-kernel}.  Each integrand is unitary.  The proof is the same lower-half-plane residue argument as in \Cref{thm:plain-entire}, but the order-one growth of the exponential is controlled by accretivity: on $z=x+\ii y$ with $y\leq0$, the logarithmic norm of
$-\ii t(H+zL)$ is at most $ty\lambda_{\min}(L)\leq0$.  Hence the lower semicircle does not acquire the factor $\ee^{tR\norm{L}}$ that would appear in a norm-only estimate.

For long-time evolution, \eqref{eq:exp-plain-lchm} is the natural representation.  Every LCHS formula for $\ee^{-tA}$ is a specialization of vanilla LCHM in which the inner transforms are Hamiltonian evolutions.  Combining this identity with the optimal approximate-LCHS construction gives the following complexity for a time-independent accretive matrix~\cite{LowSomma2025}.

\begin{corollary}[Optimal LCHS complexity for dissipative evolution]
\label{cor:lchs-exponential}
Let $A=L+\ii H$ with $L\succeq0$, let $t>0$, and let $U_A$ be an exact $(\alpha_A,a,0)$ block encoding of $A$.  For every operator error $0<\varepsilon_{\mathrm{op}}<1$, optimal LCHS gives a block encoding of an operator $\widetilde E_t$ with normalization $\alpha_{\exp}=\Theta(1)$ such that
\begin{equation}
\norm{\widetilde E_t-\ee^{-tA}}
\leq\varepsilon_{\mathrm{op}}
\label{eq:optimal-lchs-block-error}
\end{equation}
using
\begin{equation}
\mathcal{O}\!\left(
\alpha_A t\log\!\frac1{\varepsilon_{\mathrm{op}}}
\right)
\label{eq:optimal-lchs-block-query}
\end{equation}
controlled queries to $U_A$ or $U_A^\dagger$, up to the gate cost of coefficient-state preparation and multiplexing.

For a normalized input $\ket{\psi}$ with $\ee^{-tA}\ket{\psi}\neq0$, preparing the normalized evolved state to error $\varepsilon_{\mathrm{st}}$ uses
\begin{equation}
\mathcal{O}\!\left(
\frac{\alpha_A t}{\norm{\ee^{-tA}\ket{\psi}}}
\log\!\frac1{\varepsilon_{\mathrm{st}}\norm{\ee^{-tA}\ket{\psi}}}
\right)
\label{eq:optimal-lchs-state-query}
\end{equation}
controlled matrix-oracle queries and
\begin{equation}
\mathcal{O}\!\left(\frac1{\norm{\ee^{-tA}\ket{\psi}}}\right)
\label{eq:optimal-lchs-input-query}
\end{equation}
queries to the state-preparation oracle for $\ket{\psi}$.  The block-encoding query complexity and the state-preparation-oracle dependence are optimal in their respective parameters.  In particular, vanilla LCHM inherits the optimal LCHS complexity for matrix exponentials and introduces no polynomial-degree normalization factor.
\end{corollary}

The output-amplitude factor in
\eqref{eq:optimal-lchs-state-query} can contain the known global decay
$\ee^{-\mu t}$.  Indeed, if $L\succeq\mu\Id$, then
\[
\ee^{-tA}=\ee^{-\mu t}\ee^{-t(A-\mu\Id)}.
\]
For normalized-state preparation, this scalar factor cancels, and the same
identity shift also sharpens the Duhamel LCU normalization below.  Consider
\begin{equation}
\frac{\dd}{\dd s}x(s)=-Ax(s)+b,
\qquad
x(0)=x_0,
\label{eq:inhomogeneous-ode}
\end{equation}
where $x_0=\xi_0\ket{x_0}$ and $b=\xi_b\ket{b}$, with normalized
$\ket{x_0},\ket{b}$ and $\xi_0,\xi_b\geq0$.  For any certified
$\mu\geq0$ such that $L\succeq\mu\Id$, set
\begin{equation}
A_\mu:=A-\mu\Id,
\qquad
\phi_\mu(t):=
\begin{cases}
(1-\ee^{-\mu t})/\mu,&\mu>0,\\
t,&\mu=0.
\end{cases}
\label{eq:inhomogeneous-shift-data}
\end{equation}
The Duhamel formula, followed by this scalar shift, gives
\begin{equation}
x(t)
=\ee^{-\mu t}\ee^{-tA_\mu}x_0
+\int_0^t\ee^{-\mu s}\ee^{-sA_\mu}b\,\dd s.
\label{eq:inhomogeneous-duhamel-shift}
\end{equation}
Because the Hermitian part of $A_\mu$ is $L-\mu\Id\succeq0$, all propagators
in \eqref{eq:inhomogeneous-duhamel-shift} have constant-normalization LCHS
blocks.  The positive outer LCU has total weight
\begin{equation}
\eta_\mu(t)
:=\ee^{-\mu t}\xi_0+\phi_\mu(t)\xi_b,
\qquad
q_{\mathrm{ode}}(t;\mu)
:=\frac{\eta_\mu(t)}{\norm{x(t)}}.
\label{eq:inhomogeneous-normalization}
\end{equation}
Contractivity of $\ee^{-sA_\mu}$ implies
$\norm{x(t)}\leq\eta_\mu(t)$, so $q_{\mathrm{ode}}(t;\mu)\geq1$.

\begin{corollary}[LCHS for driven ODE]
\label{cor:lchs-inhomogeneous-ode}
Let $A=L+\ii H$, let $L\succeq\mu\Id$ for a known $\mu\geq0$, and let
$U_A$ be an exact $(\alpha_A,a,0)$ block encoding of $A$.  Assume coherent
state-preparation oracles for $\ket{x_0}$ and $\ket{b}$, and suppose
$x(t)\neq0$.  A coherent positive quadrature of
\eqref{eq:inhomogeneous-duhamel-shift}, combined with optimal
time-independent LCHS, prepares $x(t)/\norm{x(t)}$ to state error
$0<\varepsilon_{\mathrm{st}}<1$ using
\begin{equation}
\mathcal{O}\!\left(
q_{\mathrm{ode}}(t;\mu)\,\alpha_A t
\log\!\frac{q_{\mathrm{ode}}(t;\mu)}
{\varepsilon_{\mathrm{st}}}
\right)
\label{eq:inhomogeneous-state-query}
\end{equation}
controlled queries to $U_A$ or $U_A^\dagger$, and
\begin{equation}
\mathcal{O}\!\left(q_{\mathrm{ode}}(t;\mu)\right)
\label{eq:inhomogeneous-input-query}
\end{equation}
queries to the two state-preparation oracles.  The displayed matrix-query
count is up to the gate cost of preparing the time-quadrature coefficients
and multiplexing the variable evolution times.
\end{corollary}

\begin{proof}
Equation~\eqref{eq:inhomogeneous-duhamel-shift} expresses the solution as an
LCU of contractive propagators with coefficient $\ell_1$-norm
$\eta_\mu(t)$.  A positive quadrature preserves this normalization up to its
chosen approximation error.  Moreover,
\begin{equation}
H+k(L-\mu\Id)=H+kL-k\mu\Id,
\label{eq:inhomogeneous-known-phase}
\end{equation}
so the identity shift is implemented by a known phase and does not increase
the asymptotic matrix-query cost.  The largest propagation time is $t$.
Applying \Cref{cor:lchs-exponential} with block error
$\delta=\Theta(\varepsilon_{\mathrm{st}}/q_{\mathrm{ode}}(t;\mu))$
therefore costs
$\mathcal O(\alpha_A t\log(q_{\mathrm{ode}}(t;\mu)/
\varepsilon_{\mathrm{st}}))$ queries in one attempt.  Its success amplitude
is $\Theta(1/q_{\mathrm{ode}}(t;\mu))$. 
\end{proof}

\paragraph{Comparison with prior LCHS bounds.}
For $\mu=0$ and a constant source,
\[
q_{\mathrm{ode}}(t;0)
=
\frac{\norm{x_0}+\norm{b}_{L^1(0,t)}}{\norm{x(t)}},
\]
which is exactly the inhomogeneous LCHS amplification factor in
Refs.~\cite{AnLiuLin2023,AnChildsLin2026Dynamics}.  Since
$\ee^{-\mu t}\leq1$ and $\phi_\mu(t)\leq t$, one has
\[
q_{\mathrm{ode}}(t;\mu)\leq q_{\mathrm{ode}}(t;0), 
\]
hence the present
factor is the identity-shifted refinement of the standard one.  In the homogeneous limit $b=0$ and
$\mu=0$, it reduces to $u_{\mathrm{lchs}}=u_r$ in Table~1 of
Ref.~\cite{LowSomma2025}.  The
$\mathcal O(q_{\mathrm{ode}})$ dependence of the initial-state and source
queries is the standard LCHS amplitude-amplification dependence, while the
linear matrix-query dependence and logarithmic precision dependence in
\eqref{eq:inhomogeneous-state-query} follow from the optimal
time-independent LCHS primitive of Ref.~\cite{LowSomma2025}.

The removal of the homogeneous exponential-decay factor is explicit in the
purely driven case.  Let $x_0=0$, $t>0$, $\mu>0$, and $b\neq0$, and define the
normalized driven propagator
\begin{equation}
\Phi_{\mu,t}(A)
:=\frac{1}{\phi_\mu(t)}\int_0^t\ee^{-sA}\,\dd s
=\frac{\mu}{1-\ee^{-\mu t}}
\int_0^t\ee^{-\mu s}\ee^{-sA_\mu}\,\dd s,
\qquad
\zeta_{\mathrm{drv}}
:=\norm{\Phi_{\mu,t}(A)\ket{b}}.
\label{eq:driven-average}
\end{equation}
This is a positive, unit-weight average of contractions and hence has a
constant-normalization LCHS--LCU block encoding.  To operator error
$0<\varepsilon_{\mathrm{op}}<1$, one coherent attempt uses
\begin{equation}
\mathcal O\!\left(
\alpha_A t\log\!\frac1{\varepsilon_{\mathrm{op}}}
\right)
\label{eq:driven-block-query}
\end{equation}
controlled matrix-oracle queries, up to coefficient-state preparation and
multiplexing gates.  Since
$x(t)=(\Id-\ee^{-tA})A^{-1}b$, strict accretivity implies that $A$ is
invertible.  Writing $y=A^{-1}b$, the reverse triangle inequality and
$\norm{\ee^{-tA}}\leq\ee^{-\mu t}$ give
\begin{align}
\norm{x(t)}
&=\norm{(\Id-\ee^{-tA})y}
\geq(1-\ee^{-\mu t})\norm{y},
\label{eq:driven-solution-lower-bound}\\
\zeta_{\mathrm{drv}}
&=\frac{\norm{x(t)}}{\phi_\mu(t)\norm{b}}
\geq\frac{\mu\norm{A^{-1}b}}{\norm{b}}
\geq\frac{\mu}{\norm{A}}.
\label{eq:driven-amplitude-lower-bound}
\end{align}
Consequently the matrix-query complexity for the normalized output-state of \eqref{eq:inhomogeneous-duhamel-shift} is
\begin{align}
Q_A
&=\mathcal{O}\!\left(
\kappa_\mu\alpha_A t
\log\!\frac{\kappa_\mu}{\varepsilon_{\mathrm{st}}}
\right),
\qquad
\kappa_\mu:=\frac{\norm{A}}{\mu}.
\label{eq:driven-state-query-uniform}
\end{align}
The corresponding number of queries to the $\ket{b}$-preparation oracle is
\begin{equation}
\mathcal O\!\left(\frac1{\zeta_{\mathrm{drv}}}\right)
=\mathcal O(\kappa_\mu).
\label{eq:driven-input-query}
\end{equation}
Unlike $1/\norm{\ee^{-tA}\ket{\psi}}$, the amplification factor
$1/\zeta_{\mathrm{drv}}\leq\kappa_\mu$ is uniform in $t$.  This does not
remove the unavoidable small-output dependence for arbitrary initial data:
destructive interference between the homogeneous and driven terms can still
make $q_{\mathrm{ode}}(t;\mu)$ large.

A direct Weyl treatment of $g_t(z)=\ee^{-tz}$ requires the lifted target
\begin{equation}
F_{g_t}(z)=2\ee^{-tz}-1.
\label{eq:exp-lift}
\end{equation}
For real $t\geq0$,
\begin{equation}
\norm{F_{g_t}}_{\infty,\partial\Disk}\geq\abs{F_{g_t}(-1)}=2\ee^{t}-1.
\label{eq:exp-lift-growth}
\end{equation}
Accordingly, vanilla LCHM/LCHS provides constant-normalization unitary primitives for dissipative evolutions, while Weyl LCHM is tailored to polynomial and controlled-lift targets.

The generalized LCHS method of Ni and Ying~\cite{NiYing2026} gives a two-dimensional plane-wave representation of the same exponential eigenvalue transform through Weyl calculus.  It also turns scalar Fourier modes into Hermitian linear combinations, but it is algebraically distinct from the one-angle Chebyshev projection \eqref{eq:weyl-master-box}: generalized LCHS constructs or optimizes a two-dimensional LCHS measure, whereas Weyl LCHM gives a closed-form $N>d$ projection for polynomials.

\subsection{Matrix powers and affine iterate methods}
\label{subsec:matrix-powers-iterations}

Matrix powers and finite iterates admit several quantum realizations.  A quantum linear-system algorithm can be applied to a block lower-bidiagonal history system whose solution stores the sequence
\begin{equation}
\ket{\psi},\quad A\ket{\psi},\quad\ldots,\quad A^m\ket{\psi}
\label{eq:power-history-sequence}
\end{equation}
\cite{Berry2014}.  When $A$ is available as a one-step propagator, a coherent time-marching construction can concatenate the step maps and control the accumulated success amplitude~\cite{FangLinTong2023}.  Quantum-simulation methods based on Schr\"odingerisation embed a discrete recurrence, or an associated continuous-time system, into a larger Hamiltonian dynamics~\cite{JinLiu2024,HuEtAl2025}.  Weyl LCHM gives a direct polynomial route: it block-encodes $A^m$ itself through an exact root-of-unity projection.

For $m\geq1$, choose
\begin{equation}
P_m(z)=2z^m,
\qquad
p_m(z)=z^m,
\qquad
C_{p_m}=2.
\label{eq:power-lift-data}
\end{equation}
The circuit coherently implements the exact discrete identity \eqref{eq:discrete-monomial}; all angle branches are evaluated in superposition.  Since the minimax uniform normalization of $z^m$ over contractions is one, the Weyl normalization is within a factor two of optimal and is independent of $m$.

\begin{corollary}[Exact matrix powers]
\label{cor:exact-matrix-powers}
Let $\norm{A}\leq1$ and let $U_A$ be an exact unit-normalized block encoding.  For every integer $m\geq1$, Weyl LCHM gives a $2$-normalized block encoding of $A^m$ using
\begin{equation}
\Theta(m)\ \text{sequential controlled matrix-oracle calls},
\qquad
\mathcal{O}\!\left(m\log(m+1)\right)\ \text{angle rotations},
\label{eq:power-resources}
\end{equation}
and $a+\lceil\log_2(m+1)\rceil+2$ logical ancillas.  The $\Theta(m)$ signal-oracle depth is degree-optimal.  On a normalized input $\ket{\psi}$,
\begin{equation}
P_{\mathrm{succ}}^{(m)}
=\frac{\norm{A^m\ket{\psi}}^2}{4}.
\label{eq:power-success}
\end{equation}
Thus the normalization-dependent postselection overhead is constant.  Independent postselection takes
$\Theta(1/\norm{A^m\ket{\psi}}^2)$ attempts, whereas amplitude amplification prepares the normalized power-method state using
\begin{equation}
\mathcal{O}\!\left(
\frac{m}{\norm{A^m\ket{\psi}}}
\right)
\label{eq:power-state-preparation}
\end{equation}
controlled matrix-oracle calls.  In particular, the total number of coherent repetitions is $\mathcal{O}(1)$ whenever $\norm{A^m\ket{\psi}}=\Omega(1)$.
\end{corollary}

Consider the affine stationary iteration
\begin{equation}
x^{(k+1)}=Ax^{(k)}+b,
\qquad k\geq0.
\label{eq:affine-stationary-iteration}
\end{equation}
Its $k$th iterate is the exact matrix polynomial
\begin{equation}
x^{(k)}=A^kx^{(0)}+S_k(A)b,
\qquad
S_k(z)=\sum_{\ell=0}^{k-1}z^\ell.
\label{eq:affine-iterate-polynomial}
\end{equation}
For the geometric-sum term,
\begin{equation}
P_{S,k}(z)=2S_k(z)-1
=1+2\sum_{\ell=1}^{k-1}z^\ell,
\qquad
C_{S_k}\leq2k-1.
\label{eq:geometric-sum-lift}
\end{equation}
This directly realizes any prescribed finite iterate.

\begin{corollary}[Exact finite-step affine iteration]
\label{cor:affine-iteration}
Assume coherent state-preparation oracles for
$x^{(0)}=\xi_0\ket{x_0}$ and $b=\xi_b\ket{b}$, where $\ket{x_0}$ and $\ket{b}$ are normalized and $\xi_0,\xi_b\geq0$.  For every $k\geq1$, combine one Weyl LCHM block encoding of $A^k$ with one of $S_k(A)$ by a two-branch LCU.  The resulting circuit prepares the exact vector $x^{(k)}$ with safe normalization
\begin{equation}
\alpha_k=2\xi_0+(2k-1)\xi_b.
\label{eq:affine-iteration-normalization}
\end{equation}
A single attempt uses $\mathcal{O}(k)$ controlled calls to $U_A$ or $U_A^\dagger$, $\mathcal{O}(k\log(k+1))$ angle rotations, and
$a+\lceil\log_2(k+1)\rceil+3$ logical ancillas.  The Weyl work registers are reused by the two branches, and the outer two-branch LCU adds one selector qubit.  The attempt succeeds with probability
\begin{equation}
P_{\mathrm{succ}}^{\mathrm{iter}}
=\frac{\norm{x^{(k)}}^2}{\alpha_k^2}.
\label{eq:affine-iteration-success}
\end{equation}
Amplitude amplification therefore prepares $x^{(k)}/\norm{x^{(k)}}$ using
\begin{equation}
\mathcal{O}\!\left(
\frac{k\alpha_k}{\norm{x^{(k)}}}
\right)
\label{eq:affine-iteration-state-prep}
\end{equation}
controlled matrix-oracle calls.  The depth and normalization factors are separately optimal up to constants in the uniform sequential polynomial model.  Indeed,
\begin{equation}
\alpha_k
=2\xi_0+(2k-1)\xi_b
\leq2(\xi_0+k\xi_b).
\label{eq:affine-iteration-factor-two}
\end{equation}
For the scalar contractions $A=r\Id$ with $r\uparrow1$ and aligned input states $\ket{x_0}=\ket{b}$, the output norm approaches $\xi_0+k\xi_b$; hence any normalization valid uniformly for all $\norm{A}\leq1$ is at least this supremum.  The active polynomial has degree $k$ when $\xi_0>0$ and degree $k-1$ when $\xi_0=0<\xi_b$, giving the matching linear signal-query lower bound in the sequential model.  When $b=0$, the normalization reduces to the constant $2\xi_0$ from \Cref{cor:exact-matrix-powers}.
\end{corollary}

If only the fixed point is required and $\norm{A}\leq\rho<1$, then
$x^{(\infty)}=(\Id-A)^{-1}b$.  The resolvent construction below replaces explicit dependence on the iteration count by a degree determined by the contraction gap and target accuracy.

\subsection{Resolvents, higher-order resolvents, and rational approximation}
\label{subsec:resolvents}

Assume throughout this subsection that
\begin{equation}
\abs{\lambda}>1,
\label{eq:resolvent-parameters}
\end{equation}
so the pole lies outside the closed unit disk.

\paragraph{First-order resolvent.}
For $g(A)=(\lambda\Id-A)^{-1}$, take
\begin{equation}
p_{d,1}(z)
=\frac1\lambda\sum_{m=0}^d\left(\frac z\lambda\right)^m,
\qquad
P_{d,1}(z)
=\frac1\lambda\left[
1+2\sum_{m=1}^d\left(\frac z\lambda\right)^m
\right].
\label{eq:resolvent-polynomials}
\end{equation}
The remainder factorizes as
\begin{equation}
(\lambda\Id-A)^{-1}-p_{d,1}(A)
=\frac1\lambda\left(\frac A\lambda\right)^{d+1}
\left(\Id-\frac A\lambda\right)^{-1},
\label{eq:resolvent-remainder-factorization}
\end{equation}
and therefore
\begin{equation}
\norm{(\lambda\Id-A)^{-1}-p_{d,1}(A)}
\leq\frac{\abs{\lambda}^{-(d+1)}}{\abs{\lambda}-1}.
\label{eq:resolvent-error}
\end{equation}
A sufficient degree is
\begin{equation}
d_{\mathrm{res},1}(\varepsilon)
=
\max\!\left\{1,
\left\lceil
\frac{
\max\!\left\{0,
\log\!\left(\dfrac1{\varepsilon(\abs{\lambda}-1)}\right)
\right\}
}{\log\abs{\lambda}}
\right\rceil
\right\}-1.
\label{eq:resolvent-degree}
\end{equation}
By \Cref{thm:coherent-algorithm}, this degree requires $\mathcal{O}(d)$ controlled matrix-oracle calls, $\mathcal{O}(d\log(d+1))$ angle rotations, and $a+\lceil\log_2(d+1)\rceil+2$ logical ancillas.

The normalization and perturbation moment satisfy
\begin{align}
C_{p_{d,1}}
&=\frac1{\abs{\lambda}}
\left[
1+\frac{2}{\abs{\lambda}-1}
\left(1-\abs{\lambda}^{-d}\right)
\right]
\label{eq:resolvent-normalization}\\
&\leq\frac{\abs{\lambda}+1}
{\abs{\lambda}(\abs{\lambda}-1)},
\label{eq:resolvent-lift-norm}\\
\frac12\sum_{m=1}^d m\abs{c_m}
&=\frac1{(\abs{\lambda}-1)^2}
\left[
1-(d+1)\abs{\lambda}^{-d}
+d\abs{\lambda}^{-(d+1)}
\right]
\label{eq:resolvent-coefficient-moment}\\
&\leq\frac1{(\abs{\lambda}-1)^2}.
\end{align}
Hence an inexact block encoding contributes at most $\delta_A/(\abs{\lambda}-1)^2$ to the unnormalized operator error.

\begin{corollary}[First-order resolvent complexity]
\label{cor:first-resolvent-complexity}
Let $d=d_{\mathrm{res},1}(\varepsilon)$ from \eqref{eq:resolvent-degree}.  The coherent Weyl LCHM circuit gives a $C_{p_{d,1}}$-normalized block encoding of $(\lambda\Id-A)^{-1}$ with unnormalized error at most
\begin{equation}
\varepsilon
+\frac{\delta_A}{(\abs{\lambda}-1)^2}
+C_{p_{d,1}}\varepsilon_{\mathrm{circ}}.
\label{eq:first-resolvent-total-error}
\end{equation}
It uses $\mathcal{O}(d)$ controlled matrix-oracle calls,
$\mathcal{O}(d\log(d+1))$ angle rotations, and
$a+\lceil\log_2(d+1)\rceil+2$ logical ancillas, with
\begin{equation}
C_{p_{d,1}}\leq
\frac{\abs{\lambda}+1}{\abs{\lambda}(\abs{\lambda}-1)}.
\label{eq:first-resolvent-cor-normalization}
\end{equation}
For a normalized input $\ket{\psi}$, amplitude amplification prepares the polynomial output state using
\begin{equation}
\mathcal{O}\!\left(
\frac{dC_{p_{d,1}}}{\norm{p_{d,1}(A)\ket{\psi}}}
\right)
\label{eq:first-resolvent-state-prep}
\end{equation}
controlled matrix-oracle calls.  For fixed $\abs{\lambda}>1$, the degree and block-construction query complexity are $\mathcal{O}(\!\log(1/\varepsilon))$.
\end{corollary}

Using the generating function of $T_m$, the exact angular formula is
\begin{equation}
(\lambda\Id-A)^{-1}
=\frac1{\pi\lambda}\int_0^\pi
\left(1-\frac{\ee^{2\ii\theta}}{\lambda^2}\right)
\left(
\Id-\frac{2\ee^{\ii\theta}}\lambda X_\theta
+\frac{\ee^{2\ii\theta}}{\lambda^2}\Id
\right)^{-1}
\,\dd\theta.
\label{eq:resolvent-exact-angle}
\end{equation}

\paragraph{Higher-order resolvents.}
For an integer $k\geq1$, define
\begin{equation}
R_k(\lambda;A)=(\lambda\Id-A)^{-k}.
\label{eq:higher-resolvent-definition}
\end{equation}
The negative-binomial expansion is
\begin{equation}
R_k(\lambda;A)
=\frac1{\lambda^k}\sum_{m=0}^{\infty}
\binom{m+k-1}{k-1}\frac{A^m}{\lambda^m}.
\label{eq:higher-resolvent-series}
\end{equation}
Accordingly, define
\begin{align}
p_{d,k}(z)
&=\frac1{\lambda^k}\sum_{m=0}^d
\binom{m+k-1}{k-1}\left(\frac z\lambda\right)^m,
\label{eq:higher-resolvent-polynomial}\\
P_{d,k}(z)
&=\frac1{\lambda^k}\left[
1+2\sum_{m=1}^d
\binom{m+k-1}{k-1}\left(\frac z\lambda\right)^m
\right].
\label{eq:higher-resolvent-lift-polynomial}
\end{align}
For $n=d+1$,
\begin{equation}
\norm{R_k(\lambda;A)-p_{d,k}(A)}
\leq\frac1{\abs{\lambda}^k}
\sum_{m=n}^{\infty}
\binom{m+k-1}{k-1}\abs{\lambda}^{-m}.
\label{eq:higher-resolvent-tail-sum}
\end{equation}

\begin{lemma}[Exact negative-binomial tail]
\label{lem:negative-binomial-tail}
For $0<x<1$, $n\geq1$, and $k\geq1$,
\begin{equation}
\sum_{m=n}^{\infty}\binom{m+k-1}{k-1}x^m
=
\frac{x^n}{(1-x)^k}
\sum_{j=0}^{k-1}
\binom{n+j-1}{j}(1-x)^j.
\label{eq:negative-binomial-scalar}
\end{equation}
\end{lemma}

\begin{proof}
The series can be written as
\begin{equation}
\sum_{m=n}^{\infty}\binom{m+k-1}{k-1}x^m
=
\frac1{(k-1)!}
\frac{\dd^{k-1}}{\dd x^{k-1}}
\left(\frac{x^{n+k-1}}{1-x}\right).
\label{eq:negative-binomial-derivative}
\end{equation}
Leibniz' rule followed by collecting powers of $1-x$ gives \eqref{eq:negative-binomial-scalar}.
\end{proof}

Taking $x=\abs{\lambda}^{-1}$ and $n=d+1$ gives the exact finite form
\begin{equation}
\begin{aligned}
&\frac1{\abs{\lambda}^k}
\sum_{m=d+1}^{\infty}
\binom{m+k-1}{k-1}\abs{\lambda}^{-m}\\
&\qquad=
\frac{\abs{\lambda}^{-(d+1)}}{(\abs{\lambda}-1)^k}
\sum_{j=0}^{k-1}
\binom{d+j}{j}
\left(\frac{\abs{\lambda}-1}{\abs{\lambda}}\right)^j.
\end{aligned}
\label{eq:negative-binomial-tail}
\end{equation}
Thus the smallest degree certified by the exact tail is found by a monotone one-dimensional search.

A closed-form certificate follows from Cauchy's estimate:
\begin{equation}
\norm{R_k(\lambda;A)-p_{d,k}(A)}
\leq
\frac{(\abs{\lambda}+1)^{k+1}}
{\abs{\lambda}^k(\abs{\lambda}-1)^{k+1}}
\left(\frac{\abs{\lambda}+1}{2\abs{\lambda}}\right)^d.
\label{eq:higher-resolvent-cauchy-tail}
\end{equation}
Therefore
\begin{equation}
d_{\mathrm{res},k}^{\mathrm{cf}}(\varepsilon)
=
\max\!\left\{0,
\left\lceil
\frac{
\log\!\left(
\dfrac{(\abs{\lambda}+1)^{k+1}}
{\varepsilon\abs{\lambda}^k(\abs{\lambda}-1)^{k+1}}
\right)
}{
\log\!\left(\dfrac{2\abs{\lambda}}{\abs{\lambda}+1}\right)
}
\right\rceil
\right\}
\label{eq:higher-resolvent-closed-degree}
\end{equation}
is sufficient.  By \Cref{thm:coherent-algorithm}, the resulting circuit uses $\mathcal{O}(d)$ controlled matrix-oracle calls and $\mathcal{O}(d\log(d+1))$ angle rotations.  It uses $a+\lceil\log_2(d+1)\rceil+2$ logical ancillas, with either the closed-form degree or the smaller degree obtained from \eqref{eq:negative-binomial-tail}.

The normalization and perturbation sensitivity are
\begin{align}
C_{p_{d,k}}
&=\frac1{\abs{\lambda}^k}
\left[
1+2\sum_{m=1}^d
\binom{m+k-1}{k-1}\abs{\lambda}^{-m}
\right]
\label{eq:higher-resolvent-normalization-finite}\\
&\leq\frac{2}{(\abs{\lambda}-1)^k}
-\frac1{\abs{\lambda}^k},
\label{eq:higher-resolvent-normalization}\\
\frac12\sum_{m=1}^d m\abs{c_m}
&\leq\frac{k}{(\abs{\lambda}-1)^{k+1}}.
\label{eq:higher-resolvent-coefficient-moment}
\end{align}
Consequently, assigning at most $\varepsilon_A$ unnormalized error to the block encoding is guaranteed by
\begin{equation}
\delta_A\leq
\frac{\varepsilon_A(\abs{\lambda}-1)^{k+1}}{k}.
\label{eq:higher-resolvent-block-precision}
\end{equation}

\begin{corollary}[Higher-order resolvent complexity]
\label{cor:higher-resolvent-complexity}
Choose $d$ either by the exact tail \eqref{eq:negative-binomial-tail} or by the closed certificate
$d=d_{\mathrm{res},k}^{\mathrm{cf}}(\varepsilon)$ in \eqref{eq:higher-resolvent-closed-degree}.  Then Weyl LCHM gives a $C_{p_{d,k}}$-normalized block encoding of $R_k(\lambda;A)$ with unnormalized error at most
\begin{equation}
\varepsilon
+\frac{k\delta_A}{(\abs{\lambda}-1)^{k+1}}
+C_{p_{d,k}}\varepsilon_{\mathrm{circ}},
\label{eq:higher-resolvent-total-error}
\end{equation}
using $\mathcal{O}(d)$ controlled matrix-oracle calls and $\mathcal{O}(d\log(d+1))$ angle rotations.  The normalization obeys
\begin{equation}
C_{p_{d,k}}\leq
\frac{2}{(\abs{\lambda}-1)^k}-\frac1{\abs{\lambda}^k},
\label{eq:higher-resolvent-cor-normalization}
\end{equation}
and the normalized output-state query cost is
\begin{equation}
\mathcal{O}\!\left(
\frac{dC_{p_{d,k}}}{\norm{p_{d,k}(A)\ket{\psi}}}
\right).
\label{eq:higher-resolvent-state-prep}
\end{equation}
For fixed $k$ and fixed $\abs{\lambda}>1$, both $C_{p_{d,k}}$ and the perturbation sensitivity are independent of the approximation degree, while $d=\mathcal{O}(\log(1/\varepsilon))$.
\end{corollary}

The direct boundary formula is
\begin{equation}
\begin{aligned}
R_k(\lambda;A)
=\frac1{\pi\lambda^k}\int_0^\pi
\Biggl[{}
&\left(
\Id-\frac{\ee^{\ii\theta}}\lambda
\left(X_\theta+\ii(\Id-X_\theta^2)^{1/2}\right)
\right)^{-k}\\
&+\left(
\Id-\frac{\ee^{\ii\theta}}\lambda
\left(X_\theta-\ii(\Id-X_\theta^2)^{1/2}\right)
\right)^{-k}
-\Id
\Biggr] \,\dd\theta.
\end{aligned}
\label{eq:higher-resolvent-unitary-angle}
\end{equation}
For $k=1$, this is equivalent to \eqref{eq:resolvent-exact-angle}.

\paragraph{Partial-fraction rational approximations.}
Let
\begin{equation}
\mathfrak r(A)
=h(A)+\sum_{\ell=1}^{L_0}\sum_{k=1}^{K_\ell}
 w_{\ell,k}(\lambda_\ell\Id-A)^{-k},
\qquad \abs{\lambda_\ell}>1.
\label{eq:rational-partial-fraction-new}
\end{equation}
Approximate $h$ by a polynomial and each pole term by \eqref{eq:higher-resolvent-polynomial}.  Adding all coefficients classically gives one polynomial $\mathfrak p_D$ of degree
\begin{equation}
D=\max\{d_0,d_{\ell,k}:w_{\ell,k}\neq0\}.
\label{eq:rational-degree-new}
\end{equation}
The coherent Weyl LCHM circuit then uses $\mathcal{O}(D)$ matrix-oracle queries rather than the sum of the individual degrees.  A safe lifted normalization is the sum of the component bounds, while applying the grid certificate \eqref{eq:grid-bound} after coefficient aggregation captures cancellations.  Targets with poles near the unit disk can also be combined with direct QSVT inverse primitives or sign embedding.

\begin{corollary}[Partial-fraction implementation]
\label{cor:partial-fraction-implementation}
Let $\mathfrak p_D$ be the aggregated degree-$D$ polynomial and let
\begin{equation}
C_{\mathfrak p_D}
\geq\norm{2\mathfrak p_D-\mathfrak p_D(0)}_{\infty,\partial\Disk}.
\label{eq:rational-aggregate-normalization}
\end{equation}
If the sum of the polynomial and pole-truncation errors is at most $\varepsilon_{\mathrm{rat}}$, then one coherent Weyl LCHM circuit gives a $C_{\mathfrak p_D}$-normalized block encoding of $\mathfrak r(A)$ with unnormalized error
\begin{equation}
\varepsilon_{\mathrm{rat}}
+\frac{\delta_A}{2}\sum_{m=1}^{D}m\abs{\mathfrak c_m}
+C_{\mathfrak p_D}\varepsilon_{\mathrm{circ}},
\label{eq:rational-aggregate-error}
\end{equation}
where $2\mathfrak p_D(z)-\mathfrak p_D(0)=\sum_{m=0}^{D}\mathfrak c_mz^m$.  The query complexity is $\mathcal{O}(D)$, with $D$ the maximum component degree rather than their sum.
\end{corollary}

\subsection[Matrix logarithm]{Matrix logarithm: $\Log(\Id+A)$}
\label{subsec:matrix-logarithm}

Assume
\begin{equation}
0<\rho<1,
\qquad
\norm{A}\leq\rho,
\qquad
B:=\frac{A}{\rho},
\qquad
\norm{B}\leq1.
\label{eq:log-parameters}
\end{equation}
Since $z\mapsto\Log(1+z)$ is singular at $z=-1$, we instead apply the Weyl construction to the rescaled disk-algebra target
\begin{equation}
g_\rho(z):=\Log(1+\rho z),
\qquad
g_\rho(B)=\Log(\Id+A).
\label{eq:log-rescaled-target}
\end{equation}
Its Taylor polynomial and lifted polynomial are
\begin{align}
p_d^{\log}(z)
&=\sum_{m=1}^d\frac{(-1)^{m+1}}m\rho^m z^m,
\label{eq:log-polynomials}\\
P_d^{\log}(z)
&=2p_d^{\log}(z)
=2\sum_{m=1}^d\frac{(-1)^{m+1}}m\rho^m z^m.
\label{eq:log-lift-polynomial}
\end{align}
For $n=d+1$,
\begin{equation}
\norm{\Log(\Id+A)-p_d^{\log}(B)}
\leq\frac{\rho^n}{n(1-\rho)}.
\label{eq:log-error}
\end{equation}
Here $W_0$ denotes the principal real branch of the Lambert $W$-function, defined by $W_0(x)\ee^{W_0(x)}=x$ for $x\geq0$.  A sufficient degree is
\begin{equation}
d_{\log}(\rho,\varepsilon)
=
\max\!\left\{1,
\left\lceil
\frac{
W_0\!\left(
\dfrac{\log(1/\rho)}{\varepsilon(1-\rho)}
\right)
}{\log(1/\rho)}
\right\rceil
\right\}-1.
\label{eq:log-degree}
\end{equation}

Define
\begin{equation}
C_p:=\norm{P_d^{\log}}_{\infty,\partial\Disk},
\qquad
C_g:=\norm{2g_\rho}_{\infty,\partial\Disk}.
\label{eq:log-normalization}
\end{equation}
The scalar remainder bound on the full unit disk gives
\begin{equation}
\abs{C_p-C_g}
\leq\frac{2\rho^{d+1}}{(d+1)(1-\rho)}.
\label{eq:log-Cp-Cg}
\end{equation}
Moreover, the disk centered at $1$ with radius $\rho$ has modulus in $[1-\rho,1+\rho]$ and argument bounded by $\arcsin\rho$, so
\begin{equation}
C_g
\leq
2\sqrt{
\log^2\!\frac1{1-\rho}
+(\arcsin\rho)^2
}.
\label{eq:log-target-normalization}
\end{equation}
Thus the lift normalization is bounded by the modulus of the rescaled target and is independent of $d$ up to the approximation error.  If
$P_d^{\log}(z)=\sum_{m=0}^dc_mz^m$, then
\begin{equation}
\frac12\sum_{m=1}^dm\abs{c_m}
=\sum_{m=1}^d\rho^m
\leq\frac{\rho}{1-\rho}.
\label{eq:log-coefficient-moment}
\end{equation}

\begin{corollary}[Matrix-logarithm complexity]
\label{cor:matrix-logarithm-complexity}
Let $d=d_{\log}(\rho,\varepsilon)$ from \eqref{eq:log-degree}, and suppose a $(1,a,\delta_B)$ block encoding of $B=A/\rho$ is available.  The Weyl LCHM circuit gives a $C_p$-normalized block encoding of $\Log(\Id+A)$ with unnormalized error at most
\begin{equation}
\frac{\rho^{d+1}}{(d+1)(1-\rho)}
+\frac{\rho}{1-\rho}\delta_B
+C_p\varepsilon_{\mathrm{circ}}.
\label{eq:log-total-error}
\end{equation}
It uses $\mathcal{O}(d)$ controlled calls to the block encoding of $B$ and
$\mathcal{O}(d\log(d+1))$ angle rotations.  Its normalization satisfies
\begin{equation}
C_p
\leq
2\sqrt{
\log^2\!\frac1{1-\rho}
+(\arcsin\rho)^2
}
+\frac{2\rho^{d+1}}{(d+1)(1-\rho)}.
\label{eq:log-cor-normalization}
\end{equation}
For a normalized input $\ket{\psi}$ with $p_d^{\log}(B)\ket{\psi}\neq0$, the normalized polynomial output is prepared using
\begin{equation}
\mathcal{O}\!\left(
\frac{dC_p}{\norm{p_d^{\log}(B)\ket{\psi}}}
\right)
\label{eq:log-state-prep}
\end{equation}
controlled matrix-oracle calls.  For every fixed $\rho<1$, one has
$d=\mathcal{O}(\log(1/\varepsilon))$ while $C_p=\mathcal{O}_\rho(1)$.
\end{corollary}

Because $g_\rho\in\DiskAlg$, \Cref{thm:disk-algebra} gives the exact boundary formula
\begin{equation}
\begin{aligned}
\Log(\Id+A)
=\frac1\pi\int_0^\pi
\Biggl[{}
&\Log\!\left(
\Id+\rho\ee^{\ii\theta}
\left(X_\theta^{(B)}+\ii(\Id-(X_\theta^{(B)})^2)^{1/2}\right)
\right)\\
&+\Log\!\left(
\Id+\rho\ee^{\ii\theta}
\left(X_\theta^{(B)}-\ii(\Id-(X_\theta^{(B)})^2)^{1/2}\right)
\right)
\Biggr] \,\dd\theta,
\end{aligned}
\label{eq:log-exact-angle}
\end{equation}
where $X_\theta^{(B)}=\RePart(\ee^{-\ii\theta}B)$.  Both logarithms use the principal branch, since the disks centered at $1$ with radius $\rho<1$ lie in the open right half-plane.

\subsection{Shifted fractional powers}
\label{subsec:fractional-powers}

We consider $(\lambda\Id+A)^\nu$ under
\begin{equation}
0<\rho<\lambda,
\qquad
\norm{A}\leq\rho,
\qquad
B:=\frac{A}{\rho},
\qquad
\norm{B}\leq1,
\qquad
\nu\in\C,
\label{eq:fractional-shift-assumptions}
\end{equation}
and use the principal logarithm.  The generalized binomial coefficients are
\begin{equation}
\binom{\nu}{0}=1,
\qquad
\binom{\nu}{m}
=\frac{\nu(\nu-1)\cdots(\nu-m+1)}{m!}
=\frac{\Gamma(\nu+1)}{\Gamma(m+1)\Gamma(\nu-m+1)},
\label{eq:generalized-binomial}
\end{equation}
where the finite-product definition is used at removable singularities of the gamma quotient.  Define the disk-algebra target
\begin{equation}
g_{\rho,\lambda,\nu}(z)
:=(\lambda+\rho z)^\nu
=\lambda^\nu\left(1+\frac{\rho}{\lambda}z\right)^\nu,
\qquad
g_{\rho,\lambda,\nu}(B)=(\lambda\Id+A)^\nu.
\label{eq:shifted-fractional-target}
\end{equation}
Its degree-$d$ Taylor polynomial and lift are
\begin{align}
p_d^{\nu}(z)
&=\lambda^\nu\sum_{m=0}^{d}
\binom{\nu}{m}\left(\frac{\rho z}{\lambda}\right)^m,
\label{eq:shifted-fractional-polynomial}\\
P_d^{\nu}(z)
&=2p_d^{\nu}(z)-p_d^{\nu}(0)
=\lambda^\nu\left[
1+2\sum_{m=1}^{d}\binom{\nu}{m}
\left(\frac{\rho z}{\lambda}\right)^m
\right].
\label{eq:shifted-fractional-lift}
\end{align}

For $0<r<1$, let
\begin{equation}
M_\nu(r)=\max_{\abs{z}=r}\abs{(1+z)^\nu}.
\label{eq:fractional-Mr}
\end{equation}
If $\rho/\lambda<r<1$, Cauchy's coefficient estimate gives the uniform scalar and operator bounds
\begin{equation}
\norm{(\lambda\Id+A)^\nu-p_d^{\nu}(B)}
\leq
\abs{\lambda^\nu}
\frac{M_\nu(r)}{1-\rho/(\lambda r)}
\left(\frac{\rho}{\lambda r}\right)^{d+1}.
\label{eq:shifted-fractional-error}
\end{equation}
Thus it suffices to take
\begin{equation}
d\geq
\max\!\left\{0,
\left\lceil
\frac{
\log\!\left(
\dfrac{\abs{\lambda^\nu}M_\nu(r)}
{\varepsilon(1-\rho/(\lambda r))}
\right)
}{
\log(\lambda r/\rho)
}
\right\rceil-1
\right\}.
\label{eq:shifted-fractional-degree}
\end{equation}
The disk geometry gives the computable bound
\begin{equation}
M_\nu(r)
\leq
(1+r)^{\max\{\RePart\nu,0\}}
(1-r)^{-\max\{-\RePart\nu,0\}}
\exp\!\left(\abs{\ImPart\nu}\arcsin r\right).
\label{eq:shifted-fractional-M-bound}
\end{equation}

Define
\begin{equation}
C_p:=\norm{P_d^{\nu}}_{\infty,\partial\Disk},
\qquad
C_g:=\norm{2g_{\rho,\lambda,\nu}-g_{\rho,\lambda,\nu}(0)}_{\infty,\partial\Disk}.
\label{eq:fractional-Cd-sum}
\end{equation}
The same scalar tail estimate on $\partial\Disk$ gives
\begin{equation}
\abs{C_p-C_g}
\leq
2\abs{\lambda^\nu}
\frac{M_\nu(r)}{1-\rho/(\lambda r)}
\left(\frac{\rho}{\lambda r}\right)^{d+1},
\qquad
C_g\leq\abs{\lambda^\nu}
\left[1+2M_\nu\!\left(\frac{\rho}{\lambda}\right)\right].
\label{eq:fractional-Cd-uniform}
\end{equation}
Hence $C_p$ is controlled by the modulus of the target $g_{\rho,\lambda,\nu}$ for every $\lambda>\rho$, including $\lambda\leq1$.  If
$P_d^{\nu}(z)=\sum_{m=0}^dc_mz^m$, then
\begin{equation}
\frac12\sum_{m=1}^{d}m\abs{c_m}
\leq
\abs{\lambda^\nu}M_\nu(r)
\frac{\rho/(\lambda r)}{\left(1-\rho/(\lambda r)\right)^2}.
\label{eq:fractional-moment-uniform}
\end{equation}

For real exponents, the target bound can be made explicit.  If $0<\nu<1$, then
\begin{equation}
C_g
\leq\lambda^\nu
\left[1+2\left(1+\frac{\rho}{\lambda}\right)^\nu\right].
\label{eq:fractional-positive-C}
\end{equation}
If $\nu=-p$ with $p>0$, then
\begin{equation}
C_g
\leq\lambda^{-p}
\left[1+2\left(1-\frac{\rho}{\lambda}\right)^{-p}\right].
\label{eq:fractional-negative-C}
\end{equation}
For nonnegative integer $\nu$, the series terminates and the Weyl realization is exact with $d=\nu$.

\begin{corollary}[Shifted fractional-power complexity]
\label{cor:fractional-power-complexity}
Assume $0<\rho<\lambda$ and choose
\begin{equation}
\frac{\rho}{\lambda}<r<1.
\label{eq:fractional-cor-r-range}
\end{equation}
Let $d$ satisfy \eqref{eq:shifted-fractional-degree}, and suppose a $(1,a,\delta_B)$ block encoding of $B=A/\rho$ is available.  Then the coherent Weyl LCHM circuit gives a $C_p$-normalized block encoding of $(\lambda\Id+A)^\nu$ with unnormalized error at most
\begin{equation}
\abs{\lambda^\nu}
\frac{M_\nu(r)}{1-\rho/(\lambda r)}
\left(\frac{\rho}{\lambda r}\right)^{d+1}
+\delta_B\abs{\lambda^\nu}M_\nu(r)
\frac{\rho/(\lambda r)}{\left(1-\rho/(\lambda r)\right)^2}
+C_p\varepsilon_{\mathrm{circ}}.
\label{eq:fractional-total-error}
\end{equation}
A single attempt uses $\mathcal{O}(d)$ controlled calls to the block encoding of $B$ and $\mathcal{O}(d\log(d+1))$ angle rotations.  Its normalization obeys
\begin{equation}
C_p
\leq
C_g
+2\abs{\lambda^\nu}
\frac{M_\nu(r)}{1-\rho/(\lambda r)}
\left(\frac{\rho}{\lambda r}\right)^{d+1},
\qquad
C_g\leq\abs{\lambda^\nu}
\left[1+2M_\nu\!\left(\frac{\rho}{\lambda}\right)\right].
\label{eq:fractional-cor-normalization}
\end{equation}
For a normalized input $\ket{\psi}$ with $p_d^{\nu}(B)\ket{\psi}\neq0$, amplitude amplification prepares the normalized polynomial output using
\begin{equation}
\mathcal{O}\!\left(
\frac{dC_p}{\norm{p_d^{\nu}(B)\ket{\psi}}}
\right)
\label{eq:fractional-state-prep-general}
\end{equation}
controlled matrix-oracle calls.

For the square-root and inverse-square-root cases, let $s\in\{+1,-1\}$ and set $\nu=s/2$.  A sufficient degree is
\begin{equation}
d_s(\varepsilon;r)
=
\max\!\left\{0,
\left\lceil
\frac{
\log\!\left(
\dfrac{\lambda^{s/2}M_{s/2}(r)}
{\varepsilon(1-\rho/(\lambda r))}
\right)
}{\log(\lambda r/\rho)}
\right\rceil-1
\right\},
\label{eq:half-power-degrees}
\end{equation}
where
\begin{equation}
M_{1/2}(r)=\sqrt{1+r},
\qquad
M_{-1/2}(r)=\frac1{\sqrt{1-r}}.
\label{eq:sqrt-special-constants}
\end{equation}
The corresponding target lift normalizations satisfy
\begin{align}
C_{g_{\rho,\lambda,1/2}}
&\leq\sqrt\lambda\left[
1+2\sqrt{1+\frac{\rho}{\lambda}}
\right],
\label{eq:half-power-normalizations}\\
C_{g_{\rho,\lambda,-1/2}}
&\leq\lambda^{-1/2}\left[
1+\frac{2}{\sqrt{1-\rho/\lambda}}
\right].
\label{eq:invsqrt-special-constants}
\end{align}
Thus the total matrix-oracle complexity for preparing the normalized polynomial output is
\begin{equation}
\mathcal{O}\!\left(
\frac{d_s(\varepsilon;r)
\left(C_{g_{\rho,\lambda,s/2}}+2\varepsilon\right)}
{\norm{p_{d_s}^{s/2}(B)\ket{\psi}}}
\right).
\label{eq:half-power-total-complexity}
\end{equation}
For fixed $\lambda>\rho$, $\rho$, and $r$, both the block construction and the target-controlled normalization are independent of the approximation degree apart from the logarithmic degree $d_s=\mathcal O(\log(1/\varepsilon))$.
\end{corollary}

\subsection{Half-plane sign, spectral projectors, and ReLU}
\label{subsec:sign-relu}

Let $A$ have no eigenvalues on the imaginary axis.  The half-plane matrix sign is the holomorphic functional-calculus transform
\begin{equation}
\Sign(A)=A(A^2)^{-1/2},
\label{eq:matrix-sign-definition}
\end{equation}
where the square root is the principal one.  Equivalently, $\Sign(z)=1$ on the open right half-plane and $\Sign(z)=-1$ on the open left half-plane.  It yields the spectral projectors
\begin{equation}
\Pi_\pm(A)=\frac12\left(\Id\pm\Sign(A)\right).
\label{eq:sign-projectors}
\end{equation}

\paragraph{Direct polynomial approximation.}
Suppose a compact polynomially convex set $E=E_+\cup E_-$ contains the spectrum, with $E_+$ and $E_-$ separated from the imaginary axis, and suppose a polynomial $s_d$ satisfies
\begin{equation}
\max_{z\in E_+}\abs{s_d(z)-1}
\leq\varepsilon_s,
\qquad
\max_{z\in E_-}\abs{s_d(z)+1}
\leq\varepsilon_s.
\label{eq:direct-sign-approximation}
\end{equation}
For a diagonalizable matrix $A=V\Lambda V^{-1}$,
\begin{equation}
\norm{\Sign(A)-s_d(A)}
\leq\kappa(V)\varepsilon_s.
\label{eq:direct-sign-diagonalizable}
\end{equation}
If a stronger spectral-set or resolvent certificate is available, it may replace $\kappa(V)$.  Once $s_d$ is selected, Weyl LCHM implements $s_d(A)$ exactly as a polynomial, with no additional angular error.  This route is especially attractive for real spectra contained in two separated intervals, where Chebyshev polynomial approximants are available.  Zolotarev rational approximants can instead be handled through the resolvent construction~\cite{NakatsukasaFreund2016}.

When $W(A)$ crosses the imaginary axis, its convexity prevents a numerical-range estimate from certifying a nontrivial two-sided approximation to the discontinuous sign.  In this genuinely non-normal regime, the resolvent formulation below supplies the appropriate robust certificate.

\paragraph{Resolvent and sign-embedding route.}
The matrix sign has the canonical integral representation~\cite{KenneyLaub1995,Higham2008}
\begin{equation}
\Sign(A)
=\frac{2}{\pi}\int_0^\infty
A(t^2\Id+A^2)^{-1}\,\dd t
=\frac1\pi\int_0^\infty
\left[(A-\ii t\Id)^{-1}+(A+\ii t\Id)^{-1}\right] \,\dd t.
\label{eq:sign-resolvent}
\end{equation}
After the logarithmic substitution $t=\ee^x$, sinc or trapezoidal quadrature gives an exponentially convergent rational sum under a strip-resolvent bound.  Each node is a shifted inverse and can be implemented directly by QSVT, or polynomialized using the higher-order resolvent formulas in \Cref{subsec:resolvents} when the poles remain outside the unit disk.  The sign-embedding framework of Wang and Liu~\cite{WangLiu2026} goes further for structured matrix equations: it first embeds the target operator into an off-diagonal block of a matrix sign and then rebalances the shifted inverse families so that the aggregate LCU weight can remain bounded.

\paragraph{ReLU and related transforms.}
For matrices with spectrum separated from the imaginary axis, define the half-plane ReLU by
\begin{equation}
\ReLU_{\mathrm{hp}}(A)
=\frac12 A\left(\Id+\Sign(A)\right)
=A\Pi_+(A).
\label{eq:halfplane-relu}
\end{equation}
For Hermitian $A$, this is the usual positive part $\max(A,0)$.  If $r(A)$ approximates $\Sign(A)$, then
\begin{equation}
\norm{\ReLU_{\mathrm{hp}}(A)-\tfrac12A(\Id+r(A))}
\leq\frac{\norm{A}}2\norm{\Sign(A)-r(A)}.
\label{eq:relu-error}
\end{equation}
The same sign approximant simultaneously yields
\begin{equation}
\abs{A}_{\mathrm{hp}}=A\Sign(A),
\qquad
A_+=A\Pi_+(A),
\qquad
A_-=-A\Pi_-(A),
\qquad
\operatorname{step}_{\mathrm{hp}}(A)=\Pi_+(A).
\label{eq:sign-derived-functions}
\end{equation}
Analytic smoothings such as $\tanh(\tau A)$, the logistic filter $(\Id+\ee^{-\tau A})^{-1}$, and the error-function filter $\operatorname{erf}(\tau A)$ can be treated by the Weyl polynomial route when their poles or growth lead to a controlled lift.  When Hamiltonian exponentials or resolvents give a smaller normalization, one can instead use vanilla LCHM or rational formulas.

\begin{corollary}[Complexity of sign, projector, and ReLU transforms]
\label{cor:sign-relu-complexity}
Let $s_d$ satisfy \eqref{eq:direct-sign-approximation}, define
\begin{equation}
P_{s,d}(z)=2s_d(z)-s_d(0),
\qquad
C_{s_d}:=\norm{P_{s,d}}_{\infty,\partial\Disk},
\label{eq:sign-lift-normalization}
\end{equation}
and assume a functional-calculus certificate
$\norm{\Sign(A)-s_d(A)}\leq\mathcal K_A\varepsilon_s$, with
$\mathcal K_A=\kappa(V)$ in \eqref{eq:direct-sign-diagonalizable}.  Then Weyl LCHM gives a $C_{s_d}$-normalized block encoding of $\Sign(A)$ with unnormalized error at most
\begin{equation}
\mathcal K_A\varepsilon_s
+\frac{\delta_A}{2}\sum_{m=1}^{d}m\abs{c_m^{(s)}}
+C_{s_d}\varepsilon_{\mathrm{circ}},
\label{eq:sign-total-error}
\end{equation}
using $\mathcal{O}(d)$ controlled matrix-oracle calls and
$a+\lceil\log_2(d+1)\rceil+2$ logical ancillas.  The projector polynomials $(1\pm s_d)/2$ have the same degree and ancilla count, while the ReLU polynomial
\begin{equation}
r_{d+1}(z)=\frac z2\left(1+s_d(z)\right)
\label{eq:relu-polynomial}
\end{equation}
has degree at most $d+1$, so all sign-derived transforms in \eqref{eq:sign-derived-functions} require $\mathcal{O}(d)$ matrix-oracle calls.  In the worst case, the ReLU construction uses
$a+\lceil\log_2(d+2)\rceil+2$ logical ancillas.  Rational sign approximants can instead be handled by \Cref{cor:partial-fraction-implementation}.
\end{corollary}

\section{Faber--Weyl LCHM: identity and implementation}
\label{sec:faber}

Let $K\subset\C$ be a compact convex set with nonempty interior and connected complement, and suppose $W(A)\subset K$.  Let
\begin{equation}
\Phi:\widehat\C\setminus K\longrightarrow\{w:\abs{w}>1\}
\label{eq:faber-conformal-map}
\end{equation}
be the exterior conformal map normalized by $\Phi(\infty)=\infty$ and $\Phi'(\infty)>0$.  The $m$th Faber polynomial $F_m$ is the polynomial part of the Laurent expansion of $\Phi(z)^m$ at infinity.  For $R>1$, write
\begin{equation}
\Gamma_R=\{z:\abs{\Phi(z)}=R\},
\qquad
K_R=\text{the compact region enclosed by }\Gamma_R.
\label{eq:faber-level-curves}
\end{equation}
Suppose $f$ is holomorphic on a neighborhood of $K_R$ and let
\begin{equation}
f(z)=\sum_{m=0}^{\infty}a_mF_m(z),
\qquad
r_d(z)=\sum_{m=0}^{d}a_mF_m(z)
\label{eq:faber-series}
\end{equation}
be its Faber series and degree-$d$ truncation.

\subsection{Faber--Weyl LCHM identity}

\begin{theorem}[Faber--Weyl LCHM identity]
\label{thm:faber-weyl-identity}
Let the assumptions above hold and set
\begin{equation}
M_R=\max_{z\in\Gamma_R}\abs{f(z)}.
\label{eq:faber-MR}
\end{equation}
Then
\begin{equation}
\norm{f(A)-r_d(A)}
\leq\frac{2(1+\sqrt2)M_R}{R-1}R^{-d}.
\label{eq:faber-matrix-error}
\end{equation}
If
\begin{equation}
r_d(z)=\sum_{m=0}^{d}b_mz^m,
\label{eq:faber-monomial-expansion}
\end{equation}
then, for every $N>d$ and $\theta_j=\pi j/N$,
\begin{equation}
r_d(A)=\frac1N\sum_{j=0}^{N-1}
\left[b_0+2\sum_{m=1}^{d}
 b_m\ee^{\ii m\theta_j}T_m(X_{\theta_j})\right].
\label{eq:faber-weyl-exact}
\end{equation}
The only approximation error is the Faber truncation error in \eqref{eq:faber-matrix-error}; the angular formula is exact.
\end{theorem}

\begin{proof}
For convex $K$, the standard Faber estimates are
\begin{equation}
\norm{F_m}_{\infty,K}\leq2,
\qquad
\abs{a_m}\leq M_RR^{-m}.
\label{eq:faber-basic-bounds}
\end{equation}
Therefore
\begin{equation}
\norm{f-r_d}_{\infty,K}
\leq\frac{2M_R}{R-1}R^{-d}.
\label{eq:faber-scalar-error}
\end{equation}
The Crouzeix--Palencia spectral-set estimate gives \eqref{eq:faber-matrix-error}.  Expanding $r_d$ into monomials and applying \Cref{thm:exact-projection} term by term gives \eqref{eq:faber-weyl-exact}.
\end{proof}

A sufficient nonnegative degree for operator error at most $\varepsilon$ is
\begin{equation}
d\geq d_{\mathrm F}(R,\varepsilon)
:=\max\!\left\{0,
\left\lceil
\frac{
\log\!\left(\dfrac{2(1+\sqrt2)M_R}{(R-1)\varepsilon}\right)
}{\log R}
\right\rceil
\right\}.
\label{eq:faber-degree}
\end{equation}
Faber approximation chooses the polynomial degree according to the conformal geometry of $K$; Weyl LCHM then extracts the resulting polynomial through a one-dimensional exact Fourier projection.  All conversion from the Faber basis to monomials is classical.

\subsection{Faber--Weyl quantum implementation}

\begin{theorem}[Faber--Weyl LCHM implementation]
\label{thm:faber-weyl-implementation}
Under the assumptions of \Cref{thm:faber-weyl-identity}, choose
\begin{equation}
A=c\Id+sB,
\qquad c\in\C,\quad s>0,\quad \norm{B}\leq1,
\label{eq:faber-affine-normalization}
\end{equation}
and write
\begin{equation}
2r_d(c+sz)-r_d(c)=\sum_{m=0}^{d}\widehat c_mz^m,
\qquad
C_{r_d}:=
\norm{2r_d(c+s\,\cdot)-r_d(c)}_{\infty,\partial\Disk}.
\label{eq:faber-normalized-polynomial-data}
\end{equation}
Given a $(1,a,\delta_B)$ block encoding of $B$, the coherent Weyl LCHM circuit gives a $C_{r_d}$-normalized block encoding $\widetilde G_d^{\mathrm F}$ satisfying
\begin{equation}
\norm{C_{r_d}\widetilde G_d^{\mathrm F}-f(A)}
\leq
\frac{2(1+\sqrt2)M_R}{R-1}R^{-d}
+\frac{\delta_B}{2}\sum_{m=1}^{d}m\abs{\widehat c_m}
+C_{r_d}\varepsilon_{\mathrm{circ}}.
\label{eq:faber-weyl-total-error}
\end{equation}
The circuit uses
\begin{equation}
\mathcal{O}(d)\ \text{controlled calls to the block encoding of $B$},
\qquad
\mathcal{O}\!\left(d\log(d+1)\right)\ \text{angle rotations},
\label{eq:faber-gate-complexity}
\end{equation}
and $a+\lceil\log_2(d+1)\rceil+2$ logical ancillas.  Clifford+$T$ synthesis to total rotation error $\varepsilon_{\mathrm{circ}}$ uses
\begin{equation}
\mathcal{O}\!\left(
 d\log(d+1)
 \log\!\frac{d\log(d+1)}{\varepsilon_{\mathrm{circ}}}
\right)
\label{eq:faber-T-count}
\end{equation}
$T$ gates apart from the matrix-oracle cost.  With $d=d_{\mathrm F}(R,\varepsilon)$, the matrix-query complexity is
\begin{equation}
\mathcal{O}\!\left(1+
\frac{\max\!\left\{0,\log\!\left(2(1+\sqrt2)M_R/((R-1)\varepsilon)\right)\right\}}{\log R}
\right).
\label{eq:faber-query-complexity}
\end{equation}
For a normalized input $\ket{\psi}$, amplitude amplification prepares the normalized Faber approximation using
\begin{equation}
\mathcal{O}\!\left(
\frac{dC_{r_d}}{\norm{r_d(A)\ket{\psi}}}
\right)
\label{eq:faber-state-prep}
\end{equation}
controlled matrix-oracle calls.  No truncation and angular quadrature error is introduced.
\end{theorem}

\begin{proof}
The affine normalization converts $r_d(A)$ into the degree-$d$ polynomial $r_d(c\Id+sB)$.  Its lifted polynomial is the left-hand side of \eqref{eq:faber-normalized-polynomial-data}.  Applying \Cref{thm:coherent-algorithm} gives the gate counts and the block-encoding perturbation term, while \Cref{thm:faber-weyl-identity} supplies the Faber truncation error.  The state-preparation bound follows from \Cref{cor:output-state-preparation}.
\end{proof}

The implementation depends only on the final degree, the affine lift normalization
\begin{equation}
C_{r_d}
=\norm{2r_d(c+s\,\cdot)-r_d(c)}_{\infty,\partial\Disk},
\label{eq:faber-lift-normalization}
\end{equation}
and the transformed-output norm.  It requires neither a quantum Faber recurrence nor a generating-function linear system.  For elliptical $K$, the Faber polynomials reduce to scaled Chebyshev polynomials, giving an explicit noncircular specialization~\cite{MoretNovati2001,LowSu2026}.

\section[Comparison with prior methods]{Comparison with prior methods}
\label{sec:comparison}

\subsection[Polynomial matrix-function comparison]{Comparison of degree-$d$ polynomial matrix functions}

We reproduce the comparison table of different methods in \Cref{tab:intro-comparison} here. \Cref{tab:method-comparison} compares methods for degree-$d$ polynomial matrix functions under the common assumption $\norm{A}\leq1$ and and given a $(1,a,\delta_A)$ block encoding of $A$.  It records the additional method-specific hypotheses, the circuit depth of one coherent attempt, and the coherent repetition count after amplitude amplification.

\begin{table}[H]
\centering
\scriptsize
\renewcommand{\arraystretch}{1.5}
\setlength{\tabcolsep}{0.8pt}
\begin{tabularx}{\textwidth}{|>{\centering\arraybackslash}p{2.3cm}|>{\centering\arraybackslash}X|>{\centering\arraybackslash}p{2.2cm}|>{\centering\arraybackslash}p{3.35cm}|>{\centering\arraybackslash}p{3.1cm}|}
\hline
Method & Assumptions & Circuit depth & Amplified repetitions & Ancilla qubits \\
\hline
\makecell[c]{Contour integral\\\cite{TakahiraEtAl2020,TakahiraEtAl2022,JiangAn2026}}
& $\displaystyle
\begin{gathered}
\spec(A)\subset\operatorname{int}(\Gamma),\quad
R_\Gamma=\max_{z\in\Gamma}\abs z,\\[-1pt]
B_{p,\Gamma}=\max_{z\in\Gamma}\abs{p_d(z)},\quad
\ell_\Gamma=\operatorname{len}(\Gamma),\\[-1pt]
\gamma_\Gamma=\sup_{z\in\Gamma}\norm{(z\Id-A)^{-1}}
\end{gathered}$
& \makecell[c]{$\widetilde{\mathcal O}\!\left(\gamma_\Gamma(1+R_\Gamma)\right)$}
& $\displaystyle \mathcal O\!\left(
\frac{B_{p,\Gamma}\ell_\Gamma\gamma_\Gamma}
{\norm{p_d(A)\ket{\psi}}}
\right)$
& \makecell[c]{$a+\lceil\log_2 M_\Gamma\rceil +2$} \\
\hline
\makecell[c]{\\ QEP\\\cite{LowSu2026}\\}
& $ \displaystyle
\begin{gathered}
A=S\Lambda S^{-1},\quad \spec(A)\subset\R,\\[-1pt]
\kappa_S=\norm S\norm{S^{-1}}
\end{gathered}$
& $\widetilde{\mathcal O}(\kappa_S d)$
& $\displaystyle \widetilde{\mathcal O}\!\left(
\frac{\kappa_S\norm{p_d}_{\infty,[-1,1]}}
{\norm{p_d(A)\ket{\psi}}}
\right)$
& $a+\mathcal O(\log d)$ \\
\hline
\makecell[c]{\\ $d$-regular GQSP\\\cite{GutierrezEtAl2026}\\ }
& none
& $\widetilde{\mathcal O}(d)$
& $\displaystyle \mathcal O\!\left(
\frac{\norm{p_d}_{\infty,\partial\Disk}}
{\norm{p_d(A)\ket{\psi}}}
\right)$
& $a+\mathcal O(\log d)$ \\
\hline
\makecell[c]{\\Weyl LCHM\\ (this work)\\ }
& none
& $\widetilde{\mathcal O}(d)$
& $ \displaystyle \mathcal O\!\left(
\frac{\norm{p_d}_{\infty,\partial\Disk}}
{\norm{p_d(A)\ket{\psi}}}
\right)$
& \makecell[c]{$a+\lceil\log_2(d+1)\rceil +2$} \\
\hline
\end{tabularx}
\caption{Comparison for degree-$d$ polynomial matrix functions $p_d(A)$ with $d\geq1$ under the common condition $\norm{A}\leq1$ and given a $(1,a,\delta_A)$ block encoding of $A$.  Circuit depth counts sequential matrix-oracle layers in one coherent attempt, and the fourth column assumes amplitude amplification.  Ancilla counts include the $a$ qubits of the supplied block encoding but exclude the system register and workspace internal to supplied oracles.  The $d$-regular entry uses the $\mathcal O(\log d)$ ancillary-qubit bound, and the Weyl entry uses $a+\lceil\log_2(d+1)\rceil+2$.  Here $M_\Gamma$ is the number of contour quadrature nodes, and the $\mathcal O(\log d)$ term in QEP includes its history and linear-system work registers.}
\label{tab:method-comparison}
\end{table}

The notation $\widetilde{\mathcal O}$ suppresses logarithmic factors.  Representative polynomial regimes make the resource parameters in \Cref{tab:method-comparison} explicit.

\paragraph{Weyl LCHM vs.~contour-integral methods.}
For a prescribed polynomial, Weyl LCHM fixes the root-of-unity nodes and weights algebraically and gives an exact angular projection for every $N>d$.  Its cost depends directly on the polynomial degree and transformed-output norm, with no contour, quadrature rule, or shifted-resolvent condition.  Multiplying the two contour entries in \Cref{tab:method-comparison} gives the total matrix-oracle complexity
\begin{equation}
\widetilde{\mathcal O}\!\left(
\frac{
B_{p,\Gamma}\ell_\Gamma\gamma_\Gamma^2(1+R_\Gamma)
}{
\norm{p_d(A)\ket{\psi}}
}
\right),
\label{eq:comparison-contour-total}
\end{equation}
in addition to the choice of $M_\Gamma$ quadrature nodes.  The factor $\gamma_\Gamma$ can be large for a highly non-normal matrix even when its spectrum appears benign, because it probes the resolvent and hence the pseudospectral geometry.  Weyl LCHM removes this dependence for polynomial targets.  For example, when $p_m(z)=z^m$, it implements $A^m$ exactly with $\widetilde{\mathcal O}(m)$ circuit depth and
$\mathcal O(1/\norm{A^m\ket{\psi}})$ amplified repetitions, with neither polynomial-approximation nor contour-quadrature error.

\paragraph{Weyl LCHM vs.~QEP.}
The QEP row assumes a real spectrum and a diagonalization $A=S\Lambda S^{-1}$, and its depth and post-selection bounds inherit the eigenvector condition number $\kappa_S$.  Weyl LCHM requires only a block encoding of a contraction and applies directly to complex spectra and defective matrices, with no diagonalization or $\kappa_S$ dependence.  The nilpotent Jordan shift
\begin{equation}
A=\eta J_n,\qquad 0<\eta<1.
\label{eq:comparison-jordan-shift}
\end{equation}
has $\spec(A)=\{0\}$ but $A^m=\eta^mJ_n^m\neq0$ for every $m<n$.  Weyl LCHM retains this ordered-power information and constructs the powers exactly with linear depth up to logarithmic factors.  This feature is especially useful near defectivity, where the contraction block encoding remains well defined independently of an eigenbasis condition number.

\paragraph{Weyl LCHM vs.~$d$-regular GQSP.}
A strict numerical-radius gap strengthens the baseline constant-factor comparison into a Weyl-specific advantage.  The identity $\sup_\theta\norm{X_\theta}=w(A)$ allows Weyl LCHM to use information beyond the contraction bound, while the standard $d$-regular normalization is controlled by the unit-disk polynomial norm.  The $d$-regular normalization is
$\norm{p_d}_{\infty,\partial\Disk}$, while the explicit Weyl normalization satisfies
\begin{equation}
\norm{p_d}_{\infty,\partial\Disk}
\leq
\norm{2p_d-p_d(0)}_{\infty,\partial\Disk}
\leq
3\norm{p_d}_{\infty,\partial\Disk}.
\label{eq:comparison-direct-normalizations}
\end{equation}
The direct normalizations are therefore equivalent up to a universal constant.  Weyl LCHM additionally provides a closed-form one-angle Hermitian decomposition: its angular LCU coefficients have unit $\ell_1$-norm, the root-of-unity projection is exact, and the nodes and weights require neither regularization nor optimization.

A uniform numerical-radius gap can strengthen this constant-factor comparison to an \emph{exponential advantage} for Weyl LCHM.  Given a normalized matrix family with $\norm{A_d}=1$, $\sup_{m\geq2}w(A_m)<1$, and a nonzero, genuinely degree-$d$ power, the rescaled Weyl construction in \Cref{app:rescaled-weyl-comparison} coherently amplifies $X_\theta$ to $X_\theta/s$ while damping the degree-$m$ coefficient by $s^m$.  For the explicit weighted-shift family $A_{d,\eta}$ and the degree-$d$ polynomial $p_d(z)=z^d$ analyzed there, the total state-preparation query ratio, including the coherent-amplification cost, satisfies
\begin{equation}
\frac{Q_{\mathrm{Weyl}}^{(s)}}{Q_{d\text{-reg}}}
=\widetilde{\mathcal O}\!\left(
d\left(\frac{1+\eta}{2}\right)^d
\right),
\qquad 0<\eta<1,
\label{eq:comparison-rescaled-ratio}
\end{equation}
which decreases exponentially in $d$ for every fixed $\eta<1$ and reduces to
$\widetilde{\mathcal O}(d/2^d)$ in the small-$\eta$ limiting regime.  The weighted-shift construction therefore proves an explicit family separation: a uniform numerical-radius gap and nontrivial degree-$d$ action make numerical-radius-rescaled Weyl LCHM exponentially more efficient than $d$-regular GQSP in total state-preparation queries.

\subsection{Relation with Laplace and generalized LCHS}

Laplace-LCHS extends the original LCHS method from dissipative semigroups $\ee^{-tA}$ to targets with a suitable inverse-Laplace representation.  This class includes shifted inverse powers and functions arising in mass-matrix differential equations~\cite{AnChildsLinYing2026}.  Its cost is governed by the inverse-Laplace kernel, its truncation and discretization, and the norm of the transformed output.  Weyl LCHM complements this global representation on bounded matrix domains by implementing polynomials exactly and approximating logarithms and shifted powers through explicit finite-degree formulas.

The main Laplace-LCHS analysis assumes an $L^1$ inverse-Laplace density and also permits finitely many atomic terms~\cite{AnChildsLinYing2026}.  The following elementary limit distinguishes this global function class from the bounded-domain targets treated by Weyl LCHM.

\begin{proposition}[Separation from finite-measure inverse-Laplace representations]
\label{prop:laplace-separation}
Let $\mu$ be a finite complex Borel measure on $[0,\infty)$ and suppose
\begin{equation}
 f(x)=\int_{[0,\infty)}\ee^{-xt}\,\dd\mu(t),
 \qquad x>0.
\label{eq:integrable-inverse-laplace}
\end{equation}
Then
\begin{equation}
\lim_{x\to+\infty}f(x)=\mu(\{0\}).
\label{eq:finite-measure-laplace-limit}
\end{equation}
In particular, if $\dd\mu(t)=q(t)\dd t$ with $q\in L^1(0,\infty)$, then $f(x)\to0$.  Consequently, no nonconstant polynomial, $\Log(1+x)$, or $(\lambda+x)^\nu$ with $\lambda>0$ and $\RePart\nu>0$ admits such a finite-measure representation; in the purely $L^1$ class, no nonzero polynomial does.
\end{proposition}

\begin{proof}
For every $t\geq0$, $\ee^{-xt}\to\mathbf 1_{\{0\}}(t)$ as $x\to+\infty$, while $\abs{\ee^{-xt}}\leq1$.  Dominated convergence with respect to the total-variation measure $\abs{\mu}$ gives \eqref{eq:finite-measure-laplace-limit}.  The listed nonconstant functions are unbounded on the positive real axis.
\end{proof}

The proposition concerns global inverse-Laplace representations.  On the bounded matrix domains considered here, Weyl LCHM implements nonconstant polynomials exactly and approximates the logarithm and shifted fractional powers as described in \Cref{subsec:matrix-logarithm,subsec:fractional-powers}.

Generalized LCHS via Weyl calculus instead replaces the one-dimensional inverse-Laplace ansatz by Weyl calculus and provides a broader route to analytic matrix functions on bounded domains~\cite{NiYing2026}.  Its discrete representation is a two-dimensional plane-wave expansion
\begin{equation}
 \sum_{j}w_j\ee^{-\ii(u_jL+v_jH)},
 \qquad (u_j,v_j)\in\R^2,
\label{eq:plane-wave-lchs}
\end{equation}
where $A=L+\ii H$.

The Weyl-calculus theorem converts a scalar Fourier approximation on the numerical range into an operator approximation, and the closed-form construction gives upper bounds for general analytic functions.  After restricting the frequencies to a radius $R$ and discretizing, a coherent implementation has LCU normalization
\begin{equation}
 \lambda=\sum_j\abs{w_j},
 \qquad
 \Lambda=R\lambda ,
\label{eq:plane-wave-lchs-cost}
\end{equation}
where $R$ controls the longest Hamiltonian-simulation time and $\lambda$ controls the block normalization and amplification overhead.  The optimized quantum cost therefore depends on both the scalar approximation error and the admissible frequency measure through $(R,\lambda)$.

Explicit kernels or numerical calculations are given for selected targets, including the dissipative matrix exponential, a shifted square root, and the illustrative monomial $z^4$~\cite{NiYing2026}.  For fixed $R$, $\lambda$, and a frequency grid
\begin{equation}
\mathcal G_R=\{(u_j,v_j)\}_j
\subset\{(u,v):u^2+v^2\leq R^2\},
\label{eq:plane-wave-frequency-grid}
\end{equation}
the ansatz-free generalized-LCHS construction obtains target-optimized discrete weights for a prescribed analytic function, including an arbitrary input polynomial $p_d$, from the convex program
\begin{equation}
 \min_{\{w_j\}}\ 
 \norm{
 \sum_{(u_j,v_j)\in\mathcal G_R}
 w_j\ee^{-\ii(u_jx+v_jy)}
 -p_d(x+\ii y)
 }_{W(\Omega)}
\quad\text{subject to}\quad
 \sum_{(u_j,v_j)\in\mathcal G_R}\abs{w_j}\leq\lambda .
\label{eq:plane-wave-lchs-convex-program}
\end{equation}
Together with an outer search over $(R,\lambda)$, the convex program produces ansatz-free, target-optimized discrete formulas and an a posteriori residual certificate.  Generalized LCHS also provides closed-form continuous kernels that can be discretized by quadrature.  Weyl LCHM specializes the polynomial case further by eliminating both the radius search and the weight optimization.

The one-angle Weyl LCHM representation gives explicit polynomial weights and normalization.  If
$p_d(z)=\sum_{m=0}^d a_mz^m$ and $N>d$, the root-of-unity formula fixes the angular nodes and weights algebraically, with no truncation-radius search, quadrature, or convex optimization.  The angular LCU coefficients have unit $\ell_1$-norm, and the only target-dependent block normalization is the explicit lift
\begin{equation}
 C_p=\norm{2p_d-p_d(0)}_{\infty,\partial\Disk},
\label{eq:comparison-explicit-lift}
\end{equation}
which is within universal constants of the minimax normalization
$\norm{p_d}_{\infty,\partial\Disk}$ and yields the post-selection dependence displayed in \Cref{tab:method-comparison}.  Thus polynomial targets admit a fully explicit Weyl construction: a one-dimensional grid, unit angular $\ell_1$-weight, closed-form normalization, and an exact degree-$d$ projection.

\section{Conclusion and discussion}
\label{sec:conclusion}

The central Weyl LCHM identity realizes every matrix power as an exact noncommutative Fourier coefficient of Chebyshev transforms of the Hermitian family $X_\theta=\RePart(\ee^{-\ii\theta}A)$:
\begin{equation}
A^m=\frac{2}{N}\sum_{j=0}^{N-1}
\ee^{\ii m\theta_j}T_m(X_{\theta_j}),
\qquad N>m.
\label{eq:conclusion-power-identity}
\end{equation}
The formula holds for arbitrary matrices, including non-normal ones, uses closed-form equispaced angles, and introduces no truncation or angular quadrature error.  The disk-algebra extension transfers uniform approximation of the scalar lift $F_g=2g-g(0)$ to operator approximation with constant one, after which the angular realization remains exact.

The coherent implementation evaluates all angles in superposition, applies one GQSP circuit, and performs the root-of-unity projection by uncomputing the angle register.  Its matrix-oracle circuit depth is $\Theta(d)$ for a degree-$d$ target, matching the signal-query lower bound.  The angular LCU coefficients have unit $\ell_1$-norm, while the target-dependent lift normalization satisfies
\[
\norm{p_d}_{\infty,\overline\Disk}
\leq C_p
\leq3\norm{p_d}_{\infty,\overline\Disk},
\qquad
\abs{C_p-C_g}\leq\varepsilon_{\mathrm{app}}.
\]
The relevant state-preparation parameters are the QET constants
\begin{equation}
q_{\mathrm{et}}(p_d;A,\psi)
:=\frac{\norm{p_d}_{\infty,\overline\Disk}}
{\norm{p_d(A)\ket{\psi}}},
\qquad
q_{\mathrm{et}}(g;A,\psi)
:=\frac{\norm{g}_{\infty,\overline\Disk}}
{\norm{g(A)\ket{\psi}}}.
\label{eq:conclusion-qet-constants}
\end{equation}
Scalar contractions show that any post-selection-based QET valid uniformly for arbitrary contractions must incur
$\Omega(q_{\mathrm{et}}(g;A,\psi))$ coherent repetitions in the worst case, and the Weyl circuit attains the matching $\mathcal O(q_{\mathrm{et}})$ dependence.  For exact powers, $C_{p_m}=C_g=2$ while
$q_{\mathrm{et}}(p_m;A,\psi)=q_{\mathrm{et}}(g;A,\psi)=1/\norm{A^m\ket{\psi}}$; the factor-two Weyl lift is independent of $m$, and total constant repetitions occur whenever $\norm{A^m\ket{\psi}}=\Omega(1)$.

The equality $\sup_\theta\norm{X_\theta}=w(A)$ exposes numerical-radius information that norm-only polynomial bounds do not use.  Rescaling by this radius damps the degree-$k$ coefficient by $s^k$; for the explicit weighted-shift family, the resulting end-to-end state-preparation query ratio relative to $d$-regular GQSP decreases exponentially with $d$, including coherent-amplification overhead.

The framework gives explicit resource bounds for powers, finite stationary iterations, exponentials, constant-source driven ODEs, resolvents, logarithms, shifted fractional powers, sign/projector/ReLU transforms, and partial-fraction rational approximations.  Faber--Weyl adapts the approximation degree to the conformal geometry of a noncircular numerical-range enclosure, then implements the resulting polynomial with $\mathcal{O}(d)$ matrix-oracle calls and no truncation and angular quadrature error.

Vanilla and Weyl LCHM are complementary.  Vanilla LCHM reduces the target to functions of Hermitian pencils and specializes to Hamiltonian simulation for dissipative evolution; Weyl LCHM implements polynomial approximations through exact angular projection.  Both preserve the non-normal matrix structure without diagonalization.

Vanilla and Weyl LCHM formulas may lead to new understanding and insights in non-normal matrix theory, numerical analysis, and quantum linear algebra. LCHM-based QET algorithms may provide simple, ubiquitous, and practical algorithmic design principles towards the era of fault-tolerant quantum computing.

\section*{Acknowledgments}
JPL thanks valuable discussions with Dong An, Lin Lin, and Shi Jin. JPL acknowledges support from Quantum Science and Technology--National Science and Technology Major Project under Grant No.~2024ZD0300500, Excellent Young Scientists Fund Program, start-up funding from Tsinghua University and Beijing Institute of Mathematical Sciences and Applications.

Human authors came up with the LCHM formulas. All mathematical statements, proofs, algorithms, complexity estimates, citations, and final manuscript content were independently created, reviewed, verified, and approved by the human authors, who take full responsibility for the work.

\bibliographystyle{unsrturl}
\bibliography{references}

\appendix

\section{Numerical-radius-rescaled Weyl LCHM}
\label{app:rescaled-weyl-comparison}

Numerical-radius rescaling converts a uniform gap $w(A)<1$ into degree-dependent coefficient damping.  We derive the resulting end-to-end query bound and then specialize it to $p_d(z)=z^d$.  An explicit weighted-shift family has both a uniform numerical-radius gap and nonzero degree-$d$ action, and its state-preparation query ratio relative to $d$-regular GQSP decreases exponentially with $d$.

\subsection{Numerical-radius-rescaled Weyl LCHM: method and advantage}
\label{app:rescaled-weyl-method}

Let $A$ be a contraction, let
\begin{equation}
p_d(z)=a_0+a_1z+\cdots+a_dz^d,
\qquad a_d\neq0,
\label{eq:rescaled-weyl-polynomial}
\end{equation}
and fix an input state $\ket{\psi}$.  Write
\begin{equation}
X_\theta
:=\RePart(\ee^{-\ii\theta}A)
=\frac{\ee^{-\ii\theta}A+\ee^{\ii\theta}A^\dagger}{2}.
\label{eq:rescaled-weyl-definitions}
\end{equation}
Suppose that, in addition to $\norm A\leq1$, the numerical radius is known to satisfy $w(A)<1$.  Choose $s$ directly from
\begin{equation}
w(A)=\sup_\theta\norm{X_\theta}<s<1.
\label{eq:rescaled-weyl-radius-promise}
\end{equation}
Then $\norm{X_\theta/s}\leq w(A)/s<1$.  Applying the exact angular identity to the rescaled Hermitian matrices gives, for every $N>d$, $\theta_j=\pi j/N$, and $1\leq k\leq d$,
\begin{equation}
A^k
=\frac{2s^k}{N}\sum_{j=0}^{N-1}
\ee^{\ii k\theta_j}
T_k\!\left(\frac{X_{\theta_j}}s\right).
\label{eq:rescaled-weyl-power-identity}
\end{equation}
Consequently,
\begin{equation}
p_d(A)
=\frac1N\sum_{j=0}^{N-1}
\left[
a_0\Id+2\sum_{k=1}^d
a_ks^k\ee^{\ii k\theta_j}
T_k\!\left(\frac{X_{\theta_j}}s\right)
\right].
\label{eq:rescaled-weyl-polynomial-identity}
\end{equation}
The angular averaging weights retain unit $\ell_1$-norm.  The polynomial presented to GQSP is $2p_d(s z)-p_d(0)$, with normalization
\begin{equation}
\norm{2p_d(s\,\cdot)-p_d(0)}_{\infty,\partial\Disk}
=\max_{\abs z=1}
\abs{a_0+2a_1sz+\cdots+2a_ds^dz^d}.
\label{eq:rescaled-weyl-normalization}
\end{equation}
Thus every degree-$k$ coefficient acquires a factor $s^k$.  This coefficient damping is the source of the improvement, and it can be exponentially strong for high-degree targets.

The supplied oracle directly block-encodes $X_\theta$, rather than $X_\theta/s$.  Uniform singular-value amplification~\cite{GilyenEtAl2019} approximates the odd map $x\mapsto x/s$ on $\abs{x}\leq w(A)$ while remaining bounded by one on $[-1,1]$.  Up to logarithmic dependence on the amplification accuracy, its degree is
\begin{equation}
q_{\mathrm{amp}}
=\widetilde{\mathcal O}\!\left(\frac1{s-w(A)}\right).
\label{eq:rescaled-weyl-amplification-degree}
\end{equation}
The same amplification polynomial acts coherently on all values of the angle register.  One rescaled-Weyl attempt therefore uses
\begin{equation}
\widetilde{\mathcal O}\!\left(\frac d{s-w(A)}\right)
\label{eq:rescaled-weyl-per-attempt}
\end{equation}
queries to the original block encoding of $A$, and amplitude amplification gives
\begin{equation}
Q_{\mathrm{Weyl}}^{(s)}
=\widetilde{\mathcal O}\!\left(
\frac{
d\norm{2p_d(s\,\cdot)-p_d(0)}_{\infty,\partial\Disk}
}{
(s-w(A))\norm{p_d(A)\ket{\psi}}
}
\right).
\label{eq:rescaled-weyl-total}
\end{equation}
The logarithms hidden in $\widetilde{\mathcal O}$ include the accuracy needed for the inner amplification to remain valid through the outer degree-$d$ signal-processing sequence.

For comparison, the $d$-regular construction produces the success block
$p_d(A)/\norm{p_d}_{\infty,\partial\Disk}$ using exactly $d$ calls to the supplied block encoding, so its amplified state-preparation complexity is~\cite{GutierrezEtAl2026}
\begin{equation}
Q_{d\text{-reg}}
=\mathcal O\!\left(
\frac{
d\norm{p_d}_{\infty,\partial\Disk}
}{
\norm{p_d(A)\ket{\psi}}
}
\right).
\label{eq:rescaled-weyl-dregular-total}
\end{equation}
The unscaled Weyl circuit of \Cref{thm:coherent-algorithm} has complexity
\begin{equation}
Q_{\mathrm{Weyl}}^{(1)}
=\mathcal O\!\left(
\frac{
d\norm{2p_d-p_d(0)}_{\infty,\partial\Disk}
}{
\norm{p_d(A)\ket{\psi}}
}
\right).
\label{eq:rescaled-weyl-unscaled-total}
\end{equation}
Numerical-radius rescaling improves on the norm-only $d$-regular estimate whenever the reduction from
$\norm{p_d}_{\infty,\partial\Disk}$ to
$\norm{2p_d(s\,\cdot)-p_d(0)}_{\infty,\partial\Disk}$
dominates the coherent-amplification factor $(s-w(A))^{-1}$.

\subsection[Power polynomial and exponential advantage]{The power $p_d(z)=z^d$: exponential advantage from a numerical-radius gap}
\label{app:rescaled-weyl-power}

Consider a family of normalized contractions $A_d$ and input states $\ket{\psi_d}$ satisfying
\begin{equation}
\norm{A_d}=1,
\qquad
\sup_{m\geq2}w(A_m)<1,
\qquad
A_d^d\ket{\psi_d}\notin
\operatorname{span}\!\left\{
\ket{\psi_d},A_d\ket{\psi_d},\ldots,A_d^{d-1}\ket{\psi_d}
\right\}.
\label{eq:rescaled-weyl-power-hypotheses}
\end{equation}
The last condition ensures a nonzero, genuinely degree-$d$ action, and in particular
$\norm{A_d^d\ket{\psi_d}}>0$.  For $p_d(z)=z^d$, the three normalizations can be read off directly:
\begin{equation}
\norm{p_d}_{\infty,\partial\Disk}=1,
\qquad
\norm{2p_d-p_d(0)}_{\infty,\partial\Disk}=2,
\qquad
\norm{2p_d(s\,\cdot)-p_d(0)}_{\infty,\partial\Disk}=2s^d.
\label{eq:rescaled-weyl-power-normalizations}
\end{equation}
Their state-preparation query costs therefore obey
\begin{align}
Q_{d\text{-reg}}
&=\Theta\!\left(
\frac d{\norm{A_d^d\ket{\psi_d}}}
\right),
\label{eq:rescaled-weyl-power-dregular}\\
Q_{\mathrm{Weyl}}^{(1)}
&=\mathcal O\!\left(
\frac{2d}{\norm{A_d^d\ket{\psi_d}}}
\right),
\label{eq:rescaled-weyl-power-unscaled}\\
Q_{\mathrm{Weyl}}^{(s)}
&=\widetilde{\mathcal O}\!\left(
\frac{2ds^d}{
(s-w(A_d))\norm{A_d^d\ket{\psi_d}}
}
\right).
\label{eq:rescaled-weyl-power-scaled}
\end{align}
Any fixed
$s\in(\sup_{m\geq2}w(A_m),1)$ makes the third exponentially smaller.  A sharper choice that uses the numerical radius directly is
\begin{equation}
s_d=w(A_d)\left(1+\frac1d\right),
\label{eq:rescaled-weyl-power-s-choice}
\end{equation}
which is below one for all sufficiently large $d$.  Since
$s_d-w(A_d)=w(A_d)/d$, the common factor
$\norm{A_d^d\ket{\psi_d}}^{-1}$ cancels in the comparison and
\begin{equation}
\frac{Q_{\mathrm{Weyl}}^{(s_d)}}{Q_{d\text{-reg}}}
=\widetilde{\mathcal O}\!\left(
2d\,w(A_d)^{d-1}
\left(1+\frac1d\right)^d
\right)
=\widetilde{\mathcal O}\!\left(d\,w(A_d)^d\right).
\label{eq:rescaled-weyl-power-ratio}
\end{equation}
Here $\norm{A_d}=1$ implies $w(A_d)\geq1/2$, and
$(1+1/d)^d\leq\ee$, so the last simplification changes only a constant factor.  Because $\sup_{d\geq2}w(A_d)<1$, the ratio decreases exponentially in $d$ even after including coherent amplification.  It remains to exhibit a matrix family satisfying \eqref{eq:rescaled-weyl-power-hypotheses}; the following weighted shifts provide such a witness.

\subsection{A weighted-shift family realizing the hypotheses}
\label{app:rescaled-weyl-weighted-shift}

Fix $d\geq2$ and $0<\eta<1$, and consider
\begin{equation}
A_{d,\eta}=
\begin{pmatrix}
0&1&0&\cdots&0\\
0&0&\eta&\ddots&\vdots\\
\vdots&\ddots&\ddots&\ddots&0\\
0&\cdots&0&0&\eta\\
0&\cdots&\cdots&0&0
\end{pmatrix}
\in\C^{(d+1)\times(d+1)}.
\label{eq:rescaled-shift-matrix}
\end{equation}
Thus
\begin{equation}
A_{d,\eta}\ket{2}=\ket{1},
\qquad
A_{d,\eta}\ket{j}=\eta\ket{j-1}
\quad(3\leq j\leq d+1).
\label{eq:rescaled-shift-action}
\end{equation}
Its singular values are $1,\eta,\ldots,\eta,0$, and hence
\begin{equation}
\norm{A_{d,\eta}}=1.
\label{eq:rescaled-shift-operator-norm}
\end{equation}
Moreover,
\begin{equation}
A_{d,\eta}^{d}=\eta^{d-1}\ket{1}\!\bra{d+1}\neq0,
\qquad
A_{d,\eta}^{d+1}=0.
\label{eq:rescaled-shift-power}
\end{equation}
The minimal polynomial is therefore $z^{d+1}$, so a degree-$d$ polynomial does not collapse to a lower-degree polynomial when evaluated at $A_{d,\eta}$.

Let
\begin{equation}
X_\theta=\RePart(\ee^{-\ii\theta}A_{d,\eta})
=\frac{\ee^{-\ii\theta}A_{d,\eta}
+\ee^{\ii\theta}A_{d,\eta}^\dagger}{2}.
\label{eq:rescaled-shift-angle}
\end{equation}
For
\begin{equation}
D_\theta=\diag(1,\ee^{-\ii\theta},\ee^{-2\ii\theta},\ldots,\ee^{-d\ii\theta}),
\label{eq:rescaled-shift-diagonal-unitary}
\end{equation}
one has $X_\theta=D_\theta^\dagger X_0D_\theta$.  Thus all angular matrices are unitarily similar, and
\begin{equation}
w(A_{d,\eta})
=\norm{X_0}
=\frac12\lambda_{\max}
\begin{pmatrix}
0&1&0&\cdots&0\\
1&0&\eta&\ddots&\vdots\\
0&\eta&0&\ddots&0\\
\vdots&\ddots&\ddots&\ddots&\eta\\
0&\cdots&0&\eta&0
\end{pmatrix}.
\label{eq:rescaled-shift-numerical-radius-exact}
\end{equation}
The leading $2\times2$ principal submatrix gives the lower bound, and the maximum absolute row sum gives the upper bound,
\begin{equation}
\frac12\leq w(A_{d,\eta})
\leq\frac{1+\eta}{2}<1=\norm{A_{d,\eta}}.
\label{eq:rescaled-shift-numerical-radius-bound}
\end{equation}
In particular,
\begin{equation}
w(A_{d,\eta})\longrightarrow\frac12
\qquad(\eta\to0),
\label{eq:rescaled-shift-numerical-radius-limit}
\end{equation}
while the operator norm remains one.  For $d=2$, the exact value is
\begin{equation}
w(A_{2,\eta})=\frac{\sqrt{1+\eta^2}}{2}.
\label{eq:rescaled-shift-d2-radius}
\end{equation}

The remaining hypotheses of \eqref{eq:rescaled-weyl-power-hypotheses} hold with
\begin{equation}
p_d(z)=z^d,
\qquad
\ket{\psi_d}=\ket{d+1},
\qquad
s_{d,\eta}
=w(A_{d,\eta})\left(1+\frac1d\right).
\label{eq:rescaled-shift-power-data}
\end{equation}
The vectors $\ket{\psi_d},A_{d,\eta}\ket{\psi_d},\ldots,
A_{d,\eta}^d\ket{\psi_d}$ are nonzero multiples of distinct computational-basis states.  Hence the degree-$d$ action is genuinely new, and
\begin{equation}
p_d(A_{d,\eta})\ket{\psi_d}
=A_{d,\eta}^d\ket{d+1}
=\eta^{d-1}\ket1.
\label{eq:rescaled-shift-power-output}
\end{equation}
In particular,
$\norm{p_d(A_{d,\eta})\ket{\psi_d}}=\eta^{d-1}$.
For fixed $\eta\in(0,1)$, \eqref{eq:rescaled-shift-numerical-radius-bound} gives a uniform strict gap from one.  Therefore, for all sufficiently large $d$,
\begin{equation}
w(A_{d,\eta})<s_{d,\eta}<1,
\qquad
s_{d,\eta}-w(A_{d,\eta})
=\frac{w(A_{d,\eta})}{d}.
\label{eq:rescaled-shift-direct-gap}
\end{equation}
Substituting these expressions directly into
\eqref{eq:rescaled-weyl-power-dregular}--\eqref{eq:rescaled-weyl-power-scaled} gives
\begin{align}
Q_{d\text{-reg}}
&=\Theta\!\left(\frac d{\eta^{d-1}}\right),
\label{eq:rescaled-shift-power-dregular}\\
Q_{\mathrm{Weyl}}^{(1)}
&=\mathcal O\!\left(\frac{2d}{\eta^{d-1}}\right),
\label{eq:rescaled-shift-power-unscaled}\\
Q_{\mathrm{Weyl}}^{(s_{d,\eta})}
&=\widetilde{\mathcal O}\!\left(
\frac{
2d^2w(A_{d,\eta})^{d-1}
}{
\eta^{d-1}
}
\left(1+\frac1d\right)^d
\right).
\label{eq:rescaled-shift-power-scaled-explicit}
\end{align}
Consequently,
\begin{equation}
\frac{Q_{\mathrm{Weyl}}^{(s_{d,\eta})}}{Q_{d\text{-reg}}}
=\widetilde{\mathcal O}\!\left(
2d\,w(A_{d,\eta})^{d-1}
\left(1+\frac1d\right)^d
\right)
=\widetilde{\mathcal O}\!\left(
d\,w(A_{d,\eta})^d
\right)
=\widetilde{\mathcal O}\!\left(
d\left(\frac{1+\eta}{2}\right)^d
\right).
\label{eq:rescaled-shift-query-ratio}
\end{equation}
The first simplification uses $w(A_{d,\eta})\geq1/2$ and
$(1+1/d)^d\leq\ee$, while the final bound uses
$w(A_{d,\eta})\leq(1+\eta)/2$.  Thus the ratio decreases exponentially in $d$ for every fixed $\eta<1$.  In the small-$\eta$ limiting regime it becomes
\begin{equation}
\frac{Q_{\mathrm{Weyl}}^{(s_{d,\eta})}}{Q_{d\text{-reg}}}
=\widetilde{\mathcal O}\!\left(\frac d{2^d}\right).
\label{eq:rescaled-shift-query-ratio-limit}
\end{equation}
The common factor $\eta^{-(d-1)}$ is intrinsic to the small output norm
$\norm{A_{d,\eta}^d\ket{d+1}}=\eta^{d-1}$ and cancels from the relative comparison.  Thus the exponential advantage comes from the numerical-radius-dependent normalization, rather than from omitting the output-amplitude cost.

\subsection[General polynomial output on the weighted shift]{General degree-$d$ polynomial output on the weighted-shift witness}

Let
\begin{equation}
p_d(z)=a_0+a_1z+\cdots+a_dz^d,
\qquad a_d\neq0,
\label{eq:rescaled-shift-polynomial}
\end{equation}
and choose the input state $\ket{\psi}=\ket{d+1}$.  Directly from \eqref{eq:rescaled-shift-action},
\begin{equation}
A_{d,\eta}^k\ket{d+1}=
\begin{cases}
\eta^k\ket{d+1-k},&0\leq k\leq d-1,\\
\eta^{d-1}\ket{1},&k=d.
\end{cases}
\label{eq:rescaled-shift-basis-powers}
\end{equation}
Consequently,
\begin{equation}
\begin{aligned}
p_d(A_{d,\eta})\ket{d+1}
={}&a_0\ket{d+1}+a_1\eta\ket d+\cdots
+a_{d-1}\eta^{d-1}\ket2\\
&+a_d\eta^{d-1}\ket1,
\end{aligned}
\label{eq:rescaled-shift-polynomial-output}
\end{equation}
and
\begin{equation}
\norm{p_d(A_{d,\eta})\ket{d+1}}^2
=\abs{a_0}^2+
\sum_{k=1}^{d-1}\abs{a_k}^2\eta^{2k}
+\abs{a_d}^2\eta^{2d-2}.
\label{eq:rescaled-shift-output-norm}
\end{equation}
The same output norm appears in all three methods above.  It therefore cancels from their relative comparison, leaving the per-attempt cost and the block normalization as the distinguishing quantities.

\subsection{Regime of exponential advantage}

Suppose the matrix family has a uniform numerical-radius gap,
\[
\sup_d w(A_d)<1.
\]
Numerical-radius rescaling then trades the coherent-amplification factor $(s-w(A_d))^{-1}$ for the degree-dependent damping $s^k$.  For $p_d(z)=z^d$, this damping is exponential in $d$.  Equation~\eqref{eq:rescaled-weyl-total} includes the inner amplification and its precision dependence, so the resulting separation applies to the complete state-preparation query cost.  The weighted-shift family realizes the required uniform gap together with genuinely degree-$d$ action and therefore exhibits the claimed exponential advantage in block normalization, post-selection, and total state-preparation queries.

\section{Boundary semisimplicity}
\label{app:boundary}

\subsection{Accretive matrices}

\begin{lemma}[Boundary eigenvalues of an accretive matrix]
\label{lem:accretive-boundary}
Let $B=L+\ii H$ with $L=L^\dagger\succeq0$ and $H=H^\dagger$.  Every eigenvalue of $B$ lies in the closed right half-plane.  If $Bv=\ii\alpha v$ for a real $\alpha$, then
\begin{equation}
Lv=0,
\qquad
Hv=\alpha v,
\qquad
B^\dagger v=-\ii\alpha v.
\label{eq:accretive-boundary-relations}
\end{equation}
Consequently, every eigenvalue of $B$ on the imaginary axis is semisimple.
\end{lemma}

\begin{proof}
For a normalized eigenvector $Bv=\lambda v$,
$\RePart\lambda=\ip{v}{Lv}\geq0$.  If $\lambda=\ii\alpha$, then
$\ip{v}{Lv}=0$ and positivity gives $Lv=0$.  Hence $Hv=\alpha v$ and $B^\dagger v=-\ii\alpha v$.  If a generalized eigenvector $u$ satisfied $(B-\ii\alpha\Id)u=v$, then
\[
\norm{v}^2
=\ip{v}{(B-\ii\alpha\Id)u}
=\ip{(B-\ii\alpha\Id)^\dagger v}{u}=0,
\]
a contradiction.
\end{proof}

\subsection{Numerical-radius contractions}
\label{app:semisimple}

\begin{lemma}
\label{lem:boundary-semisimple}
Let $w(A)\leq1$, and suppose $Av=\lambda v$ with $\norm{v}=1$ and
$\abs{\lambda}=1$.  Then
\begin{equation}
A^\dagger v=\overline\lambda v.
\label{eq:boundary-adjoint}
\end{equation}
Consequently, eigenvalues of $A$ on the unit circle have no nontrivial Jordan blocks.
\end{lemma}

\begin{proof}
Set
\[
Y=\RePart(\overline\lambda A)
=\frac{\overline\lambda A+\lambda A^\dagger}{2}.
\]
The numerical-radius condition gives $Y\preceq\Id$, while
\[
\ip{v}{Yv}=\RePart(\overline\lambda\ip{v}{Av})=1.
\]
Hence $(\Id-Y)v=0$.  Using $Av=\lambda v$,
\[
v=Yv=\frac12(v+\lambda A^\dagger v),
\]
which proves \eqref{eq:boundary-adjoint}.

If a generalized eigenvector $u$ satisfied $(A-\lambda\Id)u=v$, then
\[
1=\norm{v}^2
=\ip{v}{(A-\lambda\Id)u}
=\ip{(A-\lambda\Id)^\dagger v}{u}=0,
\]
a contradiction.
\end{proof}

This semisimplicity explains why disk-algebra boundary values, rather than boundary derivatives, suffice at eigenvalues of modulus one.

\section{Two-dimensional proof of the qubitized Chebyshev identity}
\label{app:qubitization}

Let $\mathcal V=\mathcal V^\dagger$, $\mathcal V^2=\Id$, and let $\Pi$ be an orthogonal projector.  Define
\begin{equation}
X=\Pi\mathcal V\Pi|_{\Ran\Pi},
\qquad
R=2\Pi-\Id,
\qquad
W=R\mathcal V.
\label{eq:abstract-qubitization}
\end{equation}
Then $X=X^\dagger$ and $\norm{X}\leq1$.  Let $\ket{x}\in\Ran\Pi$ be a normalized eigenvector of $X$ with eigenvalue $x\in[-1,1]$.  For $\abs{x}<1$, define
\begin{equation}
\ket{x^\perp}
=\frac{(\Id-\Pi)\mathcal V\ket{x}}{\sqrt{1-x^2}}.
\label{eq:x-perp}
\end{equation}
In the basis $\{\ket{x},\ket{x^\perp}\}$,
\begin{equation}
\mathcal V=
\begin{pmatrix}
x&\sqrt{1-x^2}\\
\sqrt{1-x^2}&-x
\end{pmatrix},
\qquad
R=
\begin{pmatrix}1&0\\0&-1\end{pmatrix},
\label{eq:VR-matrices}
\end{equation}
and hence
\begin{equation}
W=
\begin{pmatrix}
x&\sqrt{1-x^2}\\
-\sqrt{1-x^2}&x
\end{pmatrix}.
\label{eq:W-matrix}
\end{equation}
Writing $x=\cos\varphi$, this is a planar rotation through angle $-\varphi$, so
\begin{equation}
\bra{x}W^m\ket{x}=\cos(m\varphi)=T_m(x).
\label{eq:scalar-chebyshev}
\end{equation}
The endpoints $x=\pm1$ follow by continuity.  Summing over the spectral decomposition of $X$ gives
\begin{equation}
\Pi W^m\Pi=T_m(X)\Pi.
\label{eq:abstract-chebyshev}
\end{equation}
Taking $\mathcal V=\mathcal V_j$, $\Pi=\Pi_X$, and $X=\widetilde X_j$ proves \eqref{eq:qubitized-chebyshev}.

\section{Operator-norm error and resource accounting}\label{app:error-resources}
\label{subsec:error-resources}

The ideal circuit acts on the actual principal block $\widetilde A$.  Since $\norm{A}\leq1$ and $\norm{\widetilde A}\leq1$, the telescoping identity gives
\begin{equation}
\norm{\widetilde A^m-A^m}\leq m\delta_A.
\label{eq:power-perturbation}
\end{equation}
For $P_d(z)=\sum_{m=0}^dc_mz^m$ and $p_d$ from \eqref{eq:pd-from-Pd},
\begin{equation}
\norm{p_d(\widetilde A)-p_d(A)}
\leq\frac{\delta_A}{2}\sum_{m=1}^d m\abs{c_m}.
\label{eq:polynomial-perturbation}
\end{equation}
Parseval's identity and Cauchy--Schwarz imply
\begin{equation}
\sum_{m=1}^d m\abs{c_m}
\leq\norm{P_d}_{\infty,\partial\Disk}
\sqrt{\frac{d(d+1)(2d+1)}6}.
\label{eq:coefficient-free}
\end{equation}

Let $\varepsilon_{\mathrm{circ}}$ denote the operator-norm difference between the implemented circuit and the ideal circuit.  If ideal unitaries $V_1,\ldots,V_q$ are replaced by $\widehat V_1,\ldots,\widehat V_q$, the telescoping product identity gives
\begin{equation}
\norm{\widehat V_q\cdots\widehat V_1-V_q\cdots V_1}
\leq\sum_{k=1}^q\norm{\widehat V_k-V_k}.
\label{eq:unitary-telescoping}
\end{equation}
Grouping the elementary errors by operation type yields the convenient bound
\begin{equation}
\varepsilon_{\mathrm{circ}}
\leq\varepsilon_{\mathrm{gqsp}}
+q_{\mathrm{sig}}\eta_{\mathrm{sig}}
+q_{\mathrm{rot}}\eta_{\mathrm{rot}},
\label{eq:circuit-error}
\end{equation}
where all errors are operator-norm errors of unitaries.  Here $q_{\mathrm{sig}}$ and $q_{\mathrm{rot}}$ are the numbers of signal calls and approximated one-qubit rotations, $\eta_{\mathrm{sig}}$ and $\eta_{\mathrm{rot}}$ are uniform per-operation bounds, and $\varepsilon_{\mathrm{gqsp}}$ is the residual GQSP implementation error.  The uniform angle state in \eqref{eq:angle-state} is prepared exactly by Hadamard gates in the ideal circuit; any physical implementation error may be included in the same telescoping sum.

\begin{theorem}[End-to-end operator-norm bound]
\label{thm:total-error}
Let $g\in\DiskAlg$, let $P_d(z)=\sum_{m=0}^dc_mz^m$ satisfy \eqref{eq:approximation-assumption}, and let $\widetilde G$ be the normalized success block of the implemented circuit in \Cref{thm:coherent-algorithm}.  Then
\begin{equation}
\norm{\widetilde G-\frac{g(A)}{C_p}}
\leq
\frac{\varepsilon_{\mathrm{app}}}{C_p}
+\frac{\delta_A}{2C_p}\sum_{m=1}^d m\abs{c_m}
+\varepsilon_{\mathrm{circ}}.
\label{eq:total-error-coefficients}
\end{equation}
In particular,
\begin{equation}
\norm{\widetilde G-\frac{g(A)}{C_p}}
\leq
\frac{\varepsilon_{\mathrm{app}}}{C_p}
+\frac{\delta_A}{2}
\sqrt{\frac{d(d+1)(2d+1)}6}
+\varepsilon_{\mathrm{circ}}.
\label{eq:total-error-simple}
\end{equation}
If $R>1$ and $g$ is holomorphic on a neighborhood of the closed disk $\{z:\abs{z}\leq R\}$, then with $M_g(R)=\max_{\abs{z}=R}\abs{g(z)}$ one also has
\begin{equation}
\norm{\widetilde G-\frac{g(A)}{C_p}}
\leq
\frac1{C_p}\left[
\varepsilon_{\mathrm{app}}
+\delta_AM_g(R)\frac{R}{(R-1)^2}
\right]
+\varepsilon_{\mathrm{circ}}.
\label{eq:total-error-analytic}
\end{equation}
\end{theorem}

\begin{proof}
The ideal success block is $p_d(\widetilde A)/C_p$.  Add and subtract
$p_d(A)/C_p$, apply \eqref{eq:constant-one-transfer} and
\eqref{eq:polynomial-perturbation}, and then add the circuit error.  Equation~\eqref{eq:total-error-simple} follows from \eqref{eq:coefficient-free} and
$\norm{P_d}_{\infty,\partial\Disk}\leq C_p$.

For the analytic bound, write $g(z)=\sum_{m\geq0}b_mz^m$.  Cauchy's estimate and \eqref{eq:power-perturbation} give
\[
\norm{g(\widetilde A)-g(A)}
\leq\delta_A\sum_{m=1}^\infty m\abs{b_m}
\leq\delta_AM_g(R)\frac{R}{(R-1)^2}.
\]
Because $w(\widetilde A)\leq1$, \Cref{thm:lifted-approximation} also gives
$\norm{p_d(\widetilde A)-g(\widetilde A)}\leq\varepsilon_{\mathrm{app}}$.  Adding the circuit error proves \eqref{eq:total-error-analytic}.
\end{proof}

\paragraph{A computable normalization bound.}
The coefficient estimate
\begin{equation}
\norm{P_d}_{\infty,\partial\Disk}\leq\sum_{m=0}^d\abs{c_m}
\label{eq:l1-bound}
\end{equation}
is immediate.  A sharper bound is obtained from a uniform grid
$t_k=2\pi k/K$.  Define
\begin{equation}
M_K=\max_{0\leq k<K}\abs{P_d(\ee^{\ii t_k})},
\qquad
D_1(P_d)=\sum_{m=1}^d m\abs{c_m}.
\label{eq:grid-quantities}
\end{equation}
Every $t$ lies within $\pi/K$ of a grid point, while
$\abs{\frac{\dd}{\dd t}P_d(\ee^{\ii t})}\leq D_1(P_d)$.  Hence
\begin{equation}
\norm{P_d}_{\infty,\partial\Disk}
\leq M_K+\frac{\pi}{K}D_1(P_d).
\label{eq:grid-bound}
\end{equation}
The right-hand side gives a certified computable upper bound on $C_p$.

\paragraph{Normalized output-state preparation.}
For a normalized input state $\ket{\psi}$, the ideal success probability is
\begin{equation}
P_{\mathrm{succ}}
=\frac{\norm{p_d(\widetilde A)\ket{\psi}}^2}{C_p^2}.
\label{eq:success-probability}
\end{equation}
When reflections about the input state and success subspace are available, amplitude amplification prepares the normalized output using
\begin{equation}
\mathcal{O}\!\left(
\frac{dC_p}{\norm{p_d(\widetilde A)\ket{\psi}}}
\right)
\label{eq:state-preparation}
\end{equation}
controlled matrix-oracle calls.  The distinction between block construction and normalized output-state preparation is explicit in this factor.

If
\begin{equation}
\epsilon_G=\norm{\widetilde G-\frac{g(A)}{C_p}}
<\frac{\norm{g(A)\ket{\psi}}}{C_p},
\label{eq:normalized-state-condition}
\end{equation}
then the normalized postselected state obeys
\begin{equation}
\norm{
\frac{\widetilde G\ket{\psi}}{\norm{\widetilde G\ket{\psi}}}
-
\frac{g(A)\ket{\psi}}{\norm{g(A)\ket{\psi}}}
}
\leq
\frac{2C_p\epsilon_G}
{\norm{g(A)\ket{\psi}}-C_p\epsilon_G}.
\label{eq:normalized-state-error}
\end{equation}
This follows from the elementary normalization inequality
$\norm{x/\norm{x}-y/\norm{y}}
\leq2\norm{x-y}/(\norm{y}-\norm{x-y})$.

\end{document}